\documentclass[11pt]{amsart}
\usepackage{harvard,enumerate,amssymb, amsmath,bbm}
\newtheorem{thm}{Theorem}[section]

\newtheorem{lem}{Lemma}[section]
\newtheorem{prop}{Proposition}[section]
\newtheorem{ass}{Assumption}[section]

\newtheorem{condition}{Condition}[section]

\theoremstyle{definition}

\newtheorem{rem}{Remark}[section]

\numberwithin{equation}{section}

\newcommand{\textupandbold}[1]{\textup{\textbf{#1}}}

\newcommand{\E}{\mathbb{E}}
\newcommand{\R}{\mathbb{R}}

\newcommand{\argmax}{\operatornamewithlimits{argmax}}

\pdfoutput=1
\newcommand{\indep}{{\bot\negthickspace\negthickspace\bot}}

\begin{document}
\title[]{Nonparametric Analysis of Finite Mixtures}

\author[Kitamura]{Yuichi Kitamura}
\address{Cowles Foundation for Research in Economics, Yale University, New
	Haven, CT 06520.}
\email{yuichi.kitamura@yale.edu}     
\author[Laage]{Louise Laage}
\address{Cowles Foundation for Research in Economics, Yale University, New
	Haven, CT 06520.}
\email{louise.laage@yale.edu}

\date{This Version: November 6, 2018.}

\thanks{This paper supersedes a previously circulated manuscript  entitled ``Nonparametric identifiability of finite mixtures.''   The authors thank Werner Ploberger and participants at various seminars for their comments.     Kitamura acknowledges financial support
  from the National
Science Foundation.  
  }

\begin{abstract} 
Finite mixture models are useful in applied econometrics. They can be used to model unobserved heterogeneity, which plays major roles in labor economics, industrial organization and other fields. Mixtures are also convenient in dealing with contaminated sampling models and models with multiple equilibria. This paper shows that finite mixture models are nonparametrically identified under weak assumptions that are plausible in economic applications. The key is to utilize the identification power implied by information in covariates variation. First, three  identification approaches are presented, under  distinct and non-nested sets of  sufficient conditions.  Observable features of data inform us  which of the three approaches is valid.  These results apply to general nonparametric switching regressions, as well as to structural econometric models, such as auction models with unobserved heterogeneity.  Second, some extensions of the identification results are developed.  In particular, a  mixture regression where the mixing weights depend  on the value of the regressors in a fully unrestricted manner is shown to be nonparametrically identifiable.  This means a finite mixture model with function-valued unobserved heterogeneity can be identified in a cross-section setting, without restricting the dependence pattern between the regressor and the unobserved heterogeneity.  In this aspect it is akin to fixed effects panel data models which permit unrestricted correlation between unobserved heterogeneity and covariates.    Third, the paper  shows that fully nonparametric estimation of the entire mixture model is possible, by forming a sample analogue of one of the new identification strategies. The estimator is shown to possess a desirable polynomial rate of convergence as in a standard nonparametric estimation problem, despite nonregular features of the model. 
\end{abstract}

\maketitle

\section{Introduction}\label{sec:introduction}

In empirical economics it is often crucially important to control for unobserved heterogeneity,  
and mixture models provide convenient ways to deal with it.  
This paper studies identification
problems in the presence of unobserved heterogeneity under weak assumptions, by exploring    
identification in nonparametric  finite mixture models.  We then propose a fully nonparametric estimation method. 

A generic mixture model takes the
following form.      
Consider a probability distribution function $F_\alpha(\cdot)$, indexed by a random
variable $\alpha$ that takes values on a sample space $\mathcal A$.
$\alpha$ is sometimes called a  
mixing variable or a latent variable.  It can be interpreted as a term
representing unobserved
heterogeneity.  Let $G$ denote the probability
distribution for $\alpha$.  Define    
\begin{equation}
F(z) = \int_{\mathcal A} F_\alpha (z)dG(\alpha)    
\label{mixture}
\end{equation}
The researcher observes $w$ distributed
according to $F$.  In other words, the mixture distribution $F(\cdot)$ 
is generated by mixing the component probability
measures $F_\alpha(\cdot), \alpha \in \mathcal A$ according to the
mixing distribution $G(\cdot)$.  In an important special case where
$G$ is discretely distributed and the space $\mathcal A$ is finite,
(\ref{mixture}) becomes  
\begin{equation}
\label{finite}
F(z) = \sum_{j=1}^J \lambda_j F_j(z), \quad \sum_{j=1}^J
\lambda_j = 1.
\end{equation} 
For example, suppose there are $J$ types of economic agents that have
type specific distributions $F_j(z), j = 1,...,J$.  If type
$j$ is drawn with probability $\lambda_j$, the resulting data 
obeys the finite mixture model (\ref{finite}).  
The $F$ defined in (\ref{finite}) is called a finite mixture
distribution function.  This is the main concern of the current paper.  Since the paper presents various results with different models, a brief discussion of the overall nature of our  contributions might be in order, as we now summarize in the following three points:

\
\

\noindent {\it (i) Relation to other identification results.} \quad   As we mention below, currently available nonparametric identification strategies for finite mixtures often require  either (A) multiple observations (a leading example being panel data) or (B) exclusion restriction  and/or specific conditions on the shapes of the component distribution functions $F_\alpha, \alpha \in \mathcal A$.  All the results in this paper concern identification in cross-section settings (i.e. the econometrician never observes an individual with a particular realization of the mixing/latent variable $\alpha$ more than once), therefore our identification strategy has little in common with the ones that belong to Category (A).  Some papers in Category (B) assume exclusion restrictions, then invoke an identification-at-infinity type argument by focusing on observations at the tails of the component distributions.   This paper does not rely on exclusion restrictions (and it even allows the mixture weights to depend on covariates in the model discussed in Section \ref{sec:FE}).   Some other papers in Category (B) rely on symmetry of $F_\alpha$, which we do not assume either. 

\   
\ 

\noindent {\it (ii) Source of identification.} \quad The primal identification power in this paper comes from what may be called ``componentwise shift-restriction" when a covariate is observed.  That is, under an independence assumption, each component distribution generates a set of cross restrictions over a family indexed by the covariate values.  Here the term ``shift-restriction" is adopted from   \citeasnoun{klein2002shift}, who consider semiparametric estimation of ordered response models (hence their paper is not about mixtures) though the identification strategy in the current paper not directly related to theirs: it is crucial to observe that in our case  {\it for each component distribution} we obtain continuous limit analogues of shift-restrictions  defined for a (possibly finite) set of covariate values.
These componentwise shift-restrictions --- and equally importantly, the fact that after aggregating such latent distributions, the resulting mixture distribution function {\it lacks}  the shift-restriction property under a ``non-parallel condition" described later  --- deliver  fully nonparametric identification.         

\
\

\noindent {\it (iii) On identification/estimation strategies.} \quad The ``componentwise shift-restriction" described above can be usefully exploited  after taking Fourier/Laplace transforms of the model.  We then take limits in the Fourier/Laplace domains.  As noted in {\it (i)} above, this is quite different from the approach based on exclusion restrictions together with nonparametric estimators with observations at the tails of the component distributions.   Moreover,  basing identification on the upper and lower tails generally limits the number of identifiable components, typically to the case with  $J=2$, whereas our approach can be used to identify models with arbitrary $J$ (Section \ref{sec:J3}).  The number of components $J$ itself will be identified in our  approach as well.  Alternatively, if we impose a large support restriction on  covariates we can in principle establish identification in a straightforward manner.  This would be a variant of the identification-at-infinity argument, and our approach does not share this feature either.  
  As we shall see in Section  \ref{sec:estimation} it is possible to estimate the entire mixture model  fully nonparametrically with standard polynomial convergence rates under mild assumptions.  This desirable property is achieved without focusing on observations at the tails of the component distributions,  nor a large support condition on the covariate.          

\
\

We now mention some literature on the use of mixture models in general, followed by existing  methods of identification for (finite) mixtures.   
As noted before, mixtures are commonly used in models with unobserved
heterogeneity, especially in labor economics and industrial
organization.    See, for example, \citeasnoun{cameron1998life},
\citeasnoun{keane-wolpin},  \citeasnoun{berry1996airline}, \citeasnoun{arcidiacono2011conditional}, and \citeasnoun{aguirregabiria2013identification}  for applications of finite mixture models
in these fields.
They are also used extensively in duration models with unobserved
heterogeneity; see   \citeasnoun{heckman1984method}, \citeasnoun{heckman1994econometric} and \citeasnoun{van2001duration}.  A somewhat different use of
mixtures can be found in    
models of regime changes, which can be viewed as finite mixture
models.  \citeasnoun{porter1983study}, for example, uses a switching 
simultaneous equations for an empirical IO model (see also \citeasnoun{ellison1994theories} and \citeasnoun{lee1984switching}).  
Some models with multiple equilibria can be regarded as 
mixtures as well (e.g.  \citeasnoun{berry2006identification},  \citeasnoun{echenique2009testing}).  Finally, contaminated models, as
analyzed by \citeasnoun{horowitz1995identification} and \citeasnoun{manski2003partial} can be
formulated as mixture models.  

The most common estimation method for mixture
models is parametric maximum likelihood (ML).  In the notation
introduced in (\ref{mixture}), ML requires parameterizing 
$F_\alpha(\cdot)$ and $G(\cdot)$ so that they are known up to a finite
number of parameters.  The EM algorithm often provides a convenient
way to calculate the 
ML estimator for a mixture model.

This paper considers 
nonparametric identification problems in finite
mixture models.  The goal of the paper is to show that it is possible
to treat the component distributions of a mixture model in a flexible
manner. It should be noted that \citeasnoun{jewell1982mixtures} and \citeasnoun{heckman1984method}
provide important identification results for mixture models in
semiparametric settings.  Again in the notation in 
(\ref{mixture}), these authors treat the component distributions $F_\alpha(\cdot)$ parametrically,
(so that it is parameterized as $F_\alpha(\cdot,\theta)$, say, by a finite
dimensional parameter
$\theta$) while treating $G$ nonparametrically.  They develop
nonparametric ML estimators (NPMLE) for this type of models.   Note
that NPMLE, in actual applications, 
yields nonparametric estimates for $G$ that are typically discrete distributions
with only a few support points.  This fact may suggest that 
considering finite
mixture distributions from the outset, as this paper does, is likely to be  flexible
enough for practical purposes.          

Identification problems of finite mixtures have attracted much
attention in the statistics literature.  Teicher's pioneering work
(Teicher \citeyear*{teicher1961identifiability},   \citeyear*{teicher1963identifiability}) initiated this research area.   \citeasnoun{rao1992identifiability}  provides a nice summary of this topic.  See, also, \citeasnoun{lindsay1995mixture}  for a comprehensive treatment of mixture models including their
identification issues.  Many results known in this area assume
parametric component distributions.  Indeed, as   \citeasnoun{hall2003nonparametric}   put it, ``(v)ery little is known of the potential for consistent nonparametric
inference in mixtures without training data.''  Nevertheless, a number of papers have appeared  on this subject, especially  after the first version of the current paper was circulated.  These include  approaches based on multiple outcomes (e.g. 
\citeasnoun{bonhomme2016non},   \citeasnoun{bonhomme2016estimating}, \citeasnoun{d2015identification},  \citeasnoun{kasahara2009nonparametric}), or identification results  based on exclusion restrictions, with/without tail restrictions on component distributions (e.g. \citeasnoun{adams2016finite}, \citeasnoun{compiani2016using}, \citeasnoun{henry2014partial},   \citeasnoun{henry2010identifying},  \citeasnoun{hohmann2013semiparametric},    \citeasnoun{jochmans2017inference}), or  methods based on symmetry restrictions (e.g.
\citeasnoun{butucea2014semiparametric},  \citeasnoun{hohmann2013two}). 

The main result of the
present paper is that nonparametric treatment of the component
distributions of a finite mixture model is possible in a cross-sectional setting, if appropriate
covariates are available.

\section{Mixture Model with Covariates}\label{sec:main}

Consider random vectors $z$ and $x$.  
Suppose the conditional distribution of $z$ given $x$ is given by a finite
mixture model of the following form: 
\begin{equation}\label{Jmodel}
F(z|x) = \sum_{j=1}^J \lambda_j F_j(z|x), \quad \lambda_j > 0, \quad j =
1,...,J, \quad   \sum_{j=1}^J\lambda_j
= 1.
\end{equation}
The main goal is to identify the mixing probability
weights $\lambda_j, j = 1,...,J$ and the conditional component
distributions $F_j(z|x)$ from the conditional mixture distribution
$F(\cdot|x)$, using nonparametric restrictions.  
Sections \ref{sec:switch} - \ref{sec:iv} consider the case where $J
= 2$.  The
above expression then becomes:      
\begin{equation}
F(z|x) = \lambda F_1(z|x) + (1 - \lambda) F_2(z|x), \quad \lambda \in (0,1].  
\label{model}
\end{equation}
The case with $\lambda = 0$ is ruled out as we seek identification only up to labeling.    Section \ref{sec:J3} considers an extension to the case with  $J \geq 3$.   

\section{Regression}\label{sec:switch}

This section develops basic nonparametric identification results for
(\ref{model}).  Suppose $z$ and $x$ reside in $\R$ and $\R^k$, respectively.
Define 
$$
m_j(x) = \int_\R z dF_j(z|x), j = 1,2, 
$$  
i.e. the mean regression functions of the component distributions.  
Let $F_{\epsilon|x}^j, j = 1,2$ denote the distribution functions of the random
variables  
\[
\epsilon_j = z_j - m_j(x), \quad j = 1,2.
\]
Note that by construction $\int \epsilon dF_{\epsilon|x}^j(\epsilon) = 0, j =
1,2$.  With this
notation $F_j(z|x) = F_{\epsilon|x}^j(z - m_i(x)), j = 1,2$, and 
the model (\ref{model}) can be written as 
\begin{equation}\label{mixmodel}
F(z|x) = \lambda F_{\epsilon|x}^1(z - m_1(x)) + (1 -
\lambda)F_{\epsilon|x}^2(z - m_2(x)).  
\end{equation}
Our goal in this section is then to identify the elements of the right hand side of
(\ref{mixmodel})
nonparametrically from the knowledge of $F(\cdot|x)$ evaluated at
various $x$.    
Note that the model (\ref{mixmodel}) is further interpreted as a switching
regression model:  
\begin{equation} \label{switch}
z = 
\begin{cases} 
 m_1(x) + \epsilon_1, \epsilon_1|x \sim F^1_{\epsilon|x} \quad & \text{with
	probability } \lambda
\\
m_2(x) + \epsilon_2, \epsilon_2|x \sim F^2_{\epsilon|x} \quad & \text{with
	probability } 1 - \lambda.
\end{cases}
\end{equation}
Models as described above are conventionally estimated using parametric
ML.  That is, the researcher specifies (1) parametric functions for
$m_1(x)$, $m_2(x)$, e.g. $m_1(x) =
\beta_1^\top x$,  $m_2(x) = \beta_2^\top x$, and 
(2) parametric distribution functions for $F_{\epsilon|x}^1$ and
$F_{\epsilon|x}^2$, e.g. $\epsilon_1|x \sim N(0,\sigma_1^2)$,
$\epsilon_2|x \sim N(0,\sigma_2^2)$.  Examples of such methods can
be found in \citeasnoun{quandt1972new} and \citeasnoun{kiefer1978discrete}; see also \citeasnoun{hamilton1989new}   for application of ML in a time series context.  The EM algorithm is often 
used in computing the ML estimator.

While the parametric approach is attractive and practical, the consistency of ML depends
crucially on whether the parametric model is correctly
specified or not.  For example, even if $m_1$ and $m_2$ have the
correct form, misspecifications in  $F_{\epsilon|x}^1$ and
$F_{\epsilon|x}^2$ would result in a failure of consistency.  This is
quite different from standard (possibly nonlinear) regression models,
for which many distribution free estimators are available.   This may
discourage applied researchers from using mixture models.
It also raises a more fundamental question:
Is the model (\ref{switch}) identified under weaker,
non/semi-parametric assumptions?  The results in this section provide
a positive answer to this question.

Before discussing how nonparametric identification is possible, it may be helpful
to see that a certain nonparametric restriction  \emph{fails} to
generate identification in the model.  Arguably the most common identification assumption for
the standard regression model (without mixtures) is the conditional
mean restriction.  In our case, by 
the construction of $F_{\epsilon|x}^1$ and $F_{\epsilon|x}^2$ we have $\int_\R \epsilon
dF_{\epsilon|x}^1(\epsilon) = 0$ and  $\int_\R \epsilon
dF_{\epsilon|x}^2(\epsilon) = 0$.  
The question is whether the knowledge of the conditional mixture
distribution $F(z|x)$ at various $x$, combined with these
``restrictions,'' uniquely determine $F_{\epsilon|x}^1$,
$F_{\epsilon|x}^2$, $m_1$, $m_2$, and $\lambda$.  The answer is
negative; at each $x$, we can split the mixture distribution $F(z|x)$
into increasing and right continuous $\R_+$-valued functions $a(z)$ and $b(z)$, say, so
that $F(z|x) = a(z) + b(z)$.  If we let $\lambda  = \int da(z)$, $m_1(x)=\frac{1}{\lambda}\int z da(z),$ $m_2(x) = \frac{1}{1 - \lambda}\int z db(z)$, $F_{\epsilon|x}^1(\epsilon) = a(\epsilon + m_1(x))/\lambda$   and $F_{\epsilon|x}^2(\epsilon) =
b(\epsilon + m_2(x))/(1-\lambda)$  they would satisfy all the
available restrictions and information at all $x$. 
Even if $m_1$ and
$m_2$ are completely parameterized, the model is not identified; 
``splitting'' of $F(z|x)$ is not unique.  

While it is straightforward to see the above identification failure, it
highlights the fact the conditional
mean zero condition allows ``too many'' ways to split the mixture
distribution, thereby failing to deliver identification.
Fortunately, however, there exists an alternative
nonparametric restriction which identifies the model (\ref{switch}).
In what follows we focus on independence
restrictions, i.e. independence of $(\epsilon_1,\epsilon_2)$ from $x$.  

\begin{rem}
Note that it suffices to assume that the independence restriction holds (i) for  just one element of the $k$-vector of covariates (wlog we assume that it is the first element)   (ii) over a small subset of the support of the element.  The dependence property between $\epsilon$'s and the elements of $x$ other than the first is completely left unspecified.  In this sense the independence requirement should be interpreted as a conditional independence assumption.  With a rich set of controls such a requirement might be regarded reasonable.    Note this point applies to all the other identification results in this paper as well.   
\end{rem}

\subsection{First identification result}\label{sec:1stid}
\medskip

Our first result is concerned with cases where
at least one element of the vector of covariates $x = (x^1,...,x^k)^\top$ is continuous.
Assume that the first $k^*$ elements $x^1,...,x^{k^*}$ are
continuous covariates.  
We establish nonparametric identifiability at $x = x_0$ utilizing local variation in
one of the $k^*$ continuous covariates.  It is
convenient to assume that the first 
element $x^1$ is such an element, which is assumed to be prior knowledge both for identification and estimation.   The following notation is useful in
considering local variations of $x^1$: for a point $x_0 = (x_0^1,...,x_0^k)^\top \in \R$,
define 
$$
N^1(x_0,\delta) = \{(x^1,x_0^2,...,x_0^k)^\top \in \R^k| x^1 \in (x_0^1 - \delta,x_0^1 + \delta)\}.
$$

 \begin{ass}\label{ass-a}
For some $\delta > 0$,  
\begin{enumerate}[\textup{(}i\textup{)}]  

\item\label{a-ind} $\epsilon_1|x \sim F_1$ and  $\epsilon_2|x \sim
  F_2$ at all $x \in N^1(x_0,\delta)$ where $F_1$ and $F_2$ do not depend on the value of $x$,

\item\label{a-nonpara} 
If $0 < \lambda < 1$, $m_1(x_0) - m_1(x) \neq m_2(x_0) - m_2(x)$, for all $x \in
N^1(x_0,\delta)$, $x \neq x_0$,

\item\label{a-cont} $m_1$ and $m_2$ are continuous in $x^1$ at $x_0$.

\end{enumerate}
\end{ass}
With the notation above, (\ref{mixmodel}) is written as:
\begin{equation}\label{mixmodel2}
F(z|x) = \lambda F_1(z - m_1(x)) + (1 -
\lambda)F_2(z - m_2(x)).  
\end{equation}
Note that the mixing distribution is allowed to be degenerate, i.e. $J
= 1$.  As a convention let $\lambda = 1$ if the mixing model is
degenerate.  That is, with degeneration (\ref{mixmodel}) becomes 
\begin{equation}
F(z|x) = F_1(z - m_1(x)).
\end{equation}
The parameter space of $\lambda$ is therefore $(0,1]$.

We first discuss identification of the functions $m_1(\cdot), m_2(\cdot)$ in a neighborhood of the point $x_0 \in  \R^k$.  To this end a set of regularity conditions for nonparametric
identification are stated in terms of moment generating functions.
Let 
$$
M_i(t) = \int_\R e^{t\epsilon}dF_i(\epsilon), \quad i = 1,2,  
$$
for all $t$ such that this integral exists. $M_1$ and $M_2$ are the moment generating functions of the
disturbance terms $\epsilon_1$ and $\epsilon_2$.  Define 
$$
D(x) := m_2(x) - m_1(x)
$$
on $\R^k$ 
and 
$$
h(c,t) : = e^{tD(x_0)(1 + c)}  \frac {M_2(t)} {M_1(t)}, \quad c \in \R, t \in \R.
$$
The following imposes a very weak regularity  condition on the behavior of these
moment generating functions.

\begin{ass}\label{ass-b}
	\begin{enumerate}[\textup{(}i\textup{)}]

		\item\label{b-0.5} The domains of $M_1(t)$ and $M_2(t)$ are $(-\infty,\infty)$, and

		\item\label{b-2}  For some $\varepsilon >0$ either $h(\pm\varepsilon,t) = O(1)$ or $1/ h(\pm\varepsilon,t) = O(1)$, or both hold as $t \rightarrow  +\infty$.   Moreover, the same  holds  as $t \rightarrow - \infty$.

	\end{enumerate}
\end{ass}
\begin{rem}\label{rem:idass}
Note that the requirement \eqref{b-2} for the asymptotic behavior of the ratio $\frac {M_2(t)} {M_1(t)}$ is very weak and reasonable, as it allows the ratio to grow, decline or remain bounded as $t$ diverges. 
\end{rem}

Let $M(t|x)$ denote the moment generating function of $z$
conditional on $x$,
that is, 
$$
M(t|x) := \int_\R e^{tz}dF(z|x), 
$$
whose domain, by (\ref{mixmodel}) and Assumption \ref{ass-b}(\ref{b-0.5}), is $(-\infty,\infty)$, and also let 
$$
R(t,x) := \frac{M(t|x)}{M(t|x_0)}.  
$$
Note that these functions are observable.   The domain of these functions are $\R^k \times (-\infty,\infty)$ by Assumption \ref{ass-b}(\ref{b-0.5}).

\begin{lem}\label{lem:slope}  Suppose Assumptions \ref{ass-a} and
	\ref{ass-b} hold.  Then there exists  $\delta' \in (0,\delta)$ such
	that for every  $x' \in N^1(x_0,\delta')$ 
	
	\begin{enumerate}[\textup{(}i\textup{)}]
		
		\item\label{lem:slope1} $\lim_{t \rightarrow \infty}\frac 1 t \log  R(t,x') = m_1(x') - m_1(x_0) $ or  $	\lim_{t \rightarrow \infty}\frac 1 t \log  R(t,x') = m_2(x') - m_2(x_0)$, 
		
		and 
		
		\item\label{lem:slope2} $\lim_{t \rightarrow -\infty}\frac 1 t \log  R(t,x') = m_1(x') - m_1(x_0) $ or  $	\lim_{t \rightarrow -\infty}\frac 1 t \log  R(t,x') = m_2(x') - m_2(x_0)$ 
		
		hold.

	\end{enumerate}
\end{lem}

\begin{proof}[\textupandbold{Proof of Lemma~\ref{lem:slope}}]
First consider the case with $0 < \lambda < 1$.  By the continuity
condition (Assumption \ref{ass-a}(\ref{a-cont})), there exist a $\delta' \in (0,\delta)$ such
that  
\begin{equation}
\label{contfunc}
|m_2(x') - m_2(x_0)| < \frac{ \epsilon|D(x_0)|} 2 \qquad \text{and} \qquad
|m_1(x') - m_1(x_0)| < \frac{ \epsilon|D(x_0)|} 2 
\end{equation}
for all $x' \in N^1(x_0,\delta')$.  By (\ref{mixmodel}) we have 
\begin{equation}\label{mixedmgf}
M(t|x) = \lambda e^{tm_1(x)}M_1(t) + (1 - \lambda) e^{tm_2(x)}M_2(t).
\end{equation}
Now we prove part \eqref{lem:slope1}, i.e., the result with $t \rightarrow \infty$.  
Suppose $h(\pm\epsilon,t) = O(1)$ holds.   
  Write
\begin{align*}
\frac 1 t \log R(x',t)
&= \frac 1 t \log \left ( \frac 
  { \lambda e^{tm_1(x')}M_1(t) + (1 - \lambda)
    e^{tm_2(x')}M_2(t)} 
  { \lambda e^{tm_1(x_0)}M_1(t) + (1 - \lambda)
    e^{tm_2(x_0)}M_2(t)} \right ) 
\\
&= m_1(x') - m_1(x_0) +  \frac 1 t   \log \left (  \frac{\lambda + (1 - \lambda)  e^{t[m_2(x') - m_1(x')] \frac{M_2(t)}{M_1(t)} }}{\lambda + (1 - \lambda)  e^{t[m_2(x_0) - m_1(x_0)] \frac{M_2(t)}{M_1(t)} }} \right )
\end{align*}
Note that (\ref{contfunc}) guarantees that $|m_2(x') - m_1(x')|$ is less than $|D(x_0)|(1 + \epsilon)$.  We have
$$
\lim_{t \rightarrow \infty}\frac 1 t \log R(x',t)
 = m_1(x') - m_1(x_0).
$$
If $1/h(\pm\epsilon,t) = O(1)$ instead, then write
$$
\frac 1 t \log R(x',t) = m_2(x') - m_2(x_0) +  \frac 1 t   \log \left (  \frac
{\lambda  e^{t[m_1(x') - m_2(x')] \frac{M_1(t)}{M_2(t)} }  + (1 - \lambda) }
{\lambda  e^{t[m_1(x_0) - m_2(x_0)] \frac{M_1(t)}{M_2(t)} }  + (1 - \lambda) }
\right )
$$
and  again by $|m_2(x') - m_1(x')| < |D(x_0)|(1 + \epsilon)$ we obtain 
 $$
 \lim_{t \rightarrow \infty}\frac 1 t \log R(x',t)
 = m_2(x') - m_2(x_0).
 $$
If both hold, then it has to be the case that $D(x_0)= 0$.  If, on top of that, $D(x') = m_2(x') - m_1(x') > 0$ then 
$$
\lim_{t \rightarrow -\infty}\frac 1 t \log R(x',t)
= m_1(x') - m_1(x_0)
$$
and 
$$
\lim_{t \rightarrow \infty}\frac 1 t \log R(x',t)
= m_2(x') - m_2(x_0).
$$
The analysis of the case with  $D(x') = m_2(x') - m_1(x') < 0$ is similar.

The proof of part \eqref{lem:slope2} is similar.  If $\lambda = 1$ (i.e. the mixing distribution is degenerate) we have 
 $$
 \frac 1 t \log R(x',t) = m_1(x') - m_1(x_0)
 $$ 
 thus the claim trivially holds.
\end{proof}

Lemma \ref{lem:slope} suggests that  the slopes of $m_1$
  and $m_2$ are identified as far as the following condition holds.  To state it, define
  $$
  \E[z|x] = \int z d F(z|x)
  $$
  and
 $$
 \lambda_c := 
 \frac 
 { \E[z|x] - \E[z|x_0]
 	 -  (1 + c) \lim_{t \rightarrow - \infty} \frac
 	1 t \log R(t,x) }
 { \lim_{t \rightarrow + \infty} \frac 1 t \log R(t,x) - 
 	(1 + c)\lim_{t \rightarrow - \infty} \frac 1 t
 	\log R(t,x)}.
 $$    
 Note these are well defined under Assumption \ref{ass-b} (\ref{b-0.5}).  The constant $\delta$ in the following condition will be specified in the statements of Lemmas \ref{lem:slope} and \ref{lem:id}.
 \begin{condition}\label{ass:thm1id}  Either 
	\begin{enumerate}[\textup{(}i\textup{)}]

	\item\label{ass:thm1id1} $
	\lim_{t \rightarrow \infty}\frac 1 t \log  R(t,x)  \neq \lim_{t \rightarrow -\infty}\frac 1 t \log  R(t,x) \text{ for some } x \in N^1(x_0,\delta) 
	$
	
	or 
	
	\item\label{ass:thm1id2}  $\lim_{c \downarrow 0} \lambda_{c} = 1$
	
	holds.  	
	
\end{enumerate} 	
\end{condition}
With this condition we have: 
\begin{lem}\label{lem:thm1}  Suppose Assumptions \ref{ass-a}, 
 	\ref{ass-b} and Condition \ref{ass:thm1id} hold.  Then there exists  $\delta' \in (0,\delta)$ such
 	that $F(\cdot|x), x \in N^1(x_0,\delta)$ uniquely
 	determines the value of $\lambda$,  and moreover, 
 	 $$(m_1(x) - m_1(x_0),m_2(x) - m_2(x_0)) \text{ if } \lambda \in (0,1)  $$ 
 	 up to labeling and 
 	 $$m_1(x) - m_1(x_0)\text{ if } \lambda = 1$$
 	 for all $x$
 	in $N^1(x_0,\delta')$　as well.
 \end{lem}
\begin{proof}
First consider the case with $\lambda \in (0,1)$.  Suppose Condition \ref{ass:thm1id}\eqref{ass:thm1id1} fails, i.e.	
$
\lim_{t \rightarrow \infty}\frac 1 t \log  R(t,x)  = \lim_{t \rightarrow -\infty}\frac 1 t \log  R(t,x) \text{ for every } x \in N^1(x_0,\delta'). 
$
In view of Lemma \ref{lem:slope} these limits are either equal to $m_1(x) - m_1(x_0)$ or  $m_2(x) - m_2(x_0)$.  Wlog suppose it is the former.  
Note 
\begin{equation}\label{mean}
\E[z|x]  = \lambda m_1(x) + (1 - \lambda) m_2(x),
\end{equation}
 therefore
\begin{align*}
\E[z|x] - \E[z|x_0] 
& = \lambda[(m_1(x) - m_1(x_0)) -
(m_2(x) - m_2(x_0))]  +  (m_2(x) - m_2(x_0)) \\
&= (1 - \lambda)[(m_2(x) - m_2(x_0)) -
(m_1(x) - m_1(x_0))]  +  (m_1(x) - m_1(x_0)).
\end{align*}
Using this 
\begin{align*}
\lambda_{c} 
&= \frac{(1 - \lambda)[(m_2(x) - m_2(x_0)) -
	(m_1(x) - m_1(x_0))]  +  (m_1(x) - m_1(x_0))- (1 + c)(m_1(x) - m_1(x_0))}{(m_1(x) - m_1(x_0)) - (1+c)(m_1(x) - m_1(x_0))}
\\
&= -\frac{(1 - \lambda) [(m_2(x) - m_2(x_0)) -
	(m_1(x) - m_1(x_0))]  }{c 	(m_1(x) - m_1(x_0) )   } + 1.
\end{align*}
Thus Condition  \ref{ass:thm1id}\eqref{ass:thm1id2} does not hold either.  In sum, if $\lambda \neq 1$ then  Condition \ref{ass:thm1id} reduces to its first part, i.e. Condition \ref{ass:thm1id}\eqref{ass:thm1id1}. 
 Lemma \ref{lem:slope} and  Condition \ref{ass:thm1id}\eqref{ass:thm1id1} imply  either 
 $$
 \lim_{t
 	\rightarrow + \infty} \frac 1 t \log R(t,x) = m_1(x)
 - m_1(x_0), \quad  \lim_{t
 	\rightarrow - \infty} \frac 1 t \log R(t,x) = m_2(x)
 - m_2(x_0)]
 $$
or 
$$
\lim_{t
	\rightarrow + \infty} \frac 1 t \log R(t,x) = m_2(x)
- m_2(x_0), \quad  \lim_{t
	\rightarrow - \infty} \frac 1 t \log R(t,x) = m_1(x)
- m_1(x_0)].
$$
Either way  the slopes are identified.  If the former holds, then 
\begin{align*}
\lambda_{c} 
&=  \frac{(1 - \lambda)[(m_2(x) - m_2(x_0)) -
	(m_1(x) - m_1(x_0))]  +  (m_1(x) - m_1(x_0))- (1 + c)(m_2(x) - m_2(x_0))}{(m_1(x) - m_1(x_0)) - (1+c)(m_2(x) - m_2(x_0))}
\\
& \rightarrow \lambda 
\end{align*}
as $c \downarrow 0$, which identifies $\lambda$.  In the latter case re-labeling delivers the result, with $\lambda$ replaced by $1-\lambda$. 

%
%

Next, consider the case with $\lambda = 1$.  Then Condition \ref{ass:thm1id}\eqref{ass:thm1id1} cannot hold; as noted before 
 $$
\frac 1 t \log R(x',t) = m_1(x') - m_1(x_0)
$$ 
(which identifies the slope).
On the other hand
\begin{align*}
\lambda_{\delta} 
&= \frac{ (m_1(x) - m_1(x_0))- (1 + \delta)(m_1(x) - m_1(x_0))}     {(m_1(x) - m_1(x_0)) - (1+\delta)(m_1(x) - m_1(x_0))}
\\
& = 1,
\end{align*}
so indeed Condition   \ref{ass:thm1id}\eqref{ass:thm1id2} is consistent with $\lambda = 1$.  Moreover this shows that the limit of $\lambda_{\delta}$ once again identifies $\lambda$.
\end{proof}

\begin{rem}
	Condition \ref{ass:thm1id} is the main regularity restriction for our first identifiability result.  Importantly, it is testable, as both $R(t,x)$ and $\lambda_\delta$ are observable.  
\end{rem}   

\begin{rem}
A sufficient condition.	
\end{rem}
	
The next Lemma gives a full identification result.  Let ${\mathcal
  F}(\R^p)$ denote the space of distribution functions on 
$\R^p$ for some $p \in \mathbb{N}$.  Define 
$$
\bar{\mathcal F}(\R^p) = \{F: \int u F(du) = 0, F \in  {\mathcal F}(\R^p)\},
$$
the set of distribution functions with mean zero.  The parameter space
of $(F_1(\cdot),F_2(\cdot))$ is given by $\bar 
{\mathcal F}(\R)^2$.  Also, for a set $\mathcal C \subset \R^k$ let $\mathcal V(\mathcal C)$ denote the space of all real valued functions on $\mathcal C$.
\begin{lem}\label{lem:id}  Suppose Assumptions \ref{ass-a}, \ref{ass-b} and Condition \ref{ass:thm1id} hold.  Then  there exists  $\delta' \in (0,\delta)$ such
	that $F(\cdot|x), x \in
  N^1(x_0,\delta)$ uniquely
  determines $(\lambda,F_1(\cdot),F_2(\cdot),m_1(\cdot),m_2(\cdot))$ in
  the set $(0,1] \times \bar{\mathcal F}(\R)^2 
  \times {\mathcal V(N^1(x_0,\delta'))}^2$  up to labeling.
\end{lem}

\begin{proof}[\textupandbold{Proof of Lemma~\ref{lem:id}}]
Define $\dot M(t,x) = \frac
\partial {\partial t} M(t,x)$, $\ddot M(t,x) = \frac {\partial^2}
{(\partial t)^2}M(t,x)$, $\ddot M_i(t) =\frac {\partial^2}
{(\partial t)^2} M_i(t), i = 1,2$ and $\ddot M_i(t) =  \frac {\partial^2}
{(\partial t)^2}M_i(t), i = 1,2$, whose existences follow from Assumption \ref{ass-b} (\ref{b-0.5}). 

Note 
\begin{equation}\label{mean}
\dot M(0|x) = \int z dF(z|x) = \lambda m_1(x) + (1 - \lambda) m_2(x).
\end{equation}
Using this, 
\begin{align*}
\dot M(0|x_0) - \dot M(0|x)   
& = \lambda[(m_1(x_0) - m_1(x)) -
(m_2(x_0) - m_2(x))]  +  (m_2(x_0) - m_2(x)) \\
&= (1 - \lambda)[(m_2(x_0) - m_2(x)) -
(m_1(x_0) - m_1(x))]  +  (m_1(x_0) - m_1(x)).
\end{align*}
By this and Assumption \ref{ass-a}(\ref{a-nonpara}), if $0 < \lambda <
1$, $\lambda$ is
identified from  
\begin{equation}\label{lambdaid}
\lambda = 
\frac 
{[\dot M(0|x_0) - \dot M(0|x)] -  \lim_{t \rightarrow - \infty} \frac
  1 t \log \frac {M(t|x_0)} {M(t|x)} }
{ \lim_{t \rightarrow + \infty} \frac 1 t \log \frac {M(t|x_0)} {M(t|x)} - 
 \lim_{t \rightarrow - \infty} \frac 1 t \log \frac {M(t|x_0)} {M(t|x)}}
\end{equation}
evaluated at an arbitrary $x \in N^1(x_0,\delta')$ (note $\delta'$ is
  defined in Lemma \ref{lem:slope}), since $\lim_{t
  \rightarrow + \infty} \frac 1 t \log \frac
{M(t|x_0)} {M(t|x)}$ and   
 $\lim_{t \rightarrow - \infty} \frac 1 t \log \frac
 {M(t|x_0)} {M(t|x)}$ identify 
the factors $[m_1(x_0)
- m_1(x)]$ and $ [m_2(x_0) - m_2(x)]$ by Lemma \ref{lem:slope} (here
and in what follows, we assume that $m_2(x_0) - m_1(x_0) < 0$; if
$m_2(x_0) - m_1(x_0) > 0$, $\lambda$ should be replaced by $(1-\lambda)$).  The
right hand side of (\ref{lambdaid}), however, is not well-defined
  ($ = 0/0$) if the mixing distribution is degenerate, i.e. $\lambda =
  1$.  To avoid the discontinuity, let 
$$
\lambda_\delta = 
\frac 
{[\dot M(0|x_0) - \dot M(0|x)] -  (1 + \delta) \lim_{t \rightarrow - \infty} \frac
  1 t \log \frac {M(t|x_0)} {M(t|x)} }
{ \lim_{t \rightarrow + \infty} \frac 1 t \log \frac {M(t|x_0)} {M(t|x)} - 
 (1 + \delta)\lim_{t \rightarrow - \infty} \frac 1 t
 \log \frac {M(t|x_0)} {M(t|x)}}, 
$$    
which approaches to $\lambda$ as $\delta \rightarrow 0$ whether
$\lambda < 1$ or not.  Thus
$\lambda$ is determined by  
$$
\lambda = \lim_{\delta \rightarrow 0} \lambda_\delta.
$$
Next, to show that $m_1(x_0)$ and $m_2(x_0)$ are identified, note the
basic relationship of the first and second order moments:
$$
\ddot M(0|x) = \lambda [m_1(x)^2 + \ddot M_1(0)] + (1 -
\lambda)[m_2(x)^2 + \ddot M_2(0)]. 
$$
Therefore
\begin{align*}
\ddot M(0|x_0)  - \ddot M(0|x) = & \lambda [m_1(x_0)^2 - m_1(x)^2]  + (1 -
\lambda)[m_2(x_0)^2 - m_2(x)^2]
\\
= & \lambda (2m_1(x_0) - [m_1(x_0) - m_1(x)])[m_1(x_0) - m_1(x)]  
\\
& + (1 -\lambda)(2m_2(x_0) - [m_2(x_0) - m_2(x)])[m_2(x_0) - m_2(x)].
\end{align*}
Let 
$$
C(x) = \left \{\ddot M(0|x_0)  - \ddot M(0|x) + \lambda [m_1(x_0) - m_1(x)]^2 +
(1 - \lambda) [m_2(x_0) - m_2(x)]^2 \right \}/2,   
$$
then
$$
C(x) = [m_1(x_0) - m_1(x)]\lambda m_1(x_0) + [m_2(x_0) - m_2(x)](1 -
\lambda) m_2(x_0).  
$$
Notice that $C(x)$ is already identified over $N^1(x_0,\delta')$ from the above argument and
Lemma \ref{lem:slope}.  Together with (\ref{mean}), 
\begin{equation}\label{mID}
\begin{bmatrix}
C(x) \\ \dot M(0|x_0)
\end{bmatrix}
= 
\begin{pmatrix}
 [m_1(x_0) - m_1(x)] &  [m_2(x_0) - m_2(x)]
\\
1 & 1 
\end{pmatrix}
\begin{pmatrix}
\lambda &  0
\\
0 & (1 - \lambda) 
\end{pmatrix}
\begin{bmatrix}
m_1(x_0) \\ m_2(x_0)
\end{bmatrix},
\end{equation}
for all $x \in N^1(x_0,\delta')$.  By Assumptions
\ref{ass-a}(\ref{a-nonpara}), this
can be uniquely solved for $m_1(x_0)$ and $m_2(x_0)$ (if $\lambda =
1$, the above equation can be solved directly to determine $m_1(x_0)$;
another way to proceed in the degenerate case is to solve (\ref{mID})
using the Moore-Penrose generalized inverse, which identifies
$m_1(x_0)$ and yields the solution that $m_2(x_0) = 0$).  As the slopes are already obtained in Lemma \ref{lem:thm1}, the levels of $m_1$ and $m_2$ over $N^1(x_0,\delta)$ are also identified.  The only components
remaining are $F_1$ and $F_2$.  By evaluating (\ref{mixedmgf}) at
$x_0$ and $x \in N^1(x_0,\delta')$, $x \neq x_0$, obtain    
\begin{equation}\label{mgfeq}
\begin{bmatrix}
M(t|x_0) \\ M(t|x)
\end{bmatrix}
= E(x_0,x,t) \Lambda
\begin{bmatrix}
M_1(t) \\ M_2(t)
\end{bmatrix},
\end{equation}
where
\[
E(x,x',t) = 
\begin{pmatrix}
e^{tm_1(x)} & e^{tm_2(x)}
\\
e^{tm_1(x')} & e^{tm_2(x')} 
\end{pmatrix}, 
\Lambda = 
\begin{pmatrix}
\lambda &  0
\\
0 & (1 - \lambda) 
\end{pmatrix}.
\]
If the mixing distribution is non-degenerate, 
\begin{align*}
\text{Det}(E(x_0,x,t))&  = e^{t[m_1(x_0) + m_2(x)]} - e^{t[m_1(x) + m_2(x_0)]}
\\
& =  e^{t[m_1(x_0) + m_2(x)]}\left(1 - e^{t\{[m_1(x) - m_1(x_0)] - [m_2(x)
    - m_2(x_0)] }\right )
\\
& \neq  0 
\end{align*}
for all $x \in N^1(x_0,\delta)$, $x \neq x_0$, $t \neq 0$, because of
Assumption \ref{ass-a}(\ref{a-nonpara}),
guaranteeing 
the invertibility of  $E(x_0,x',t)$.  Moreover, 
$$
E(x_0,x,t) = e^{t[m_1(x_0) + m_2(x_0)]}
\begin{pmatrix}
e^{-tm_2(x_0)} & e^{-tm_1(x_0)}
\\
e^{t\left\{[m_1(x) - m_1(x_0)] - m_2(x_0)\right \}} &
e^{t\left\{[m_2(x) - m_2(x_0)] - m_1(x_0)\right \}} 
\end{pmatrix}.  
$$
Therefore $E(x_0,x,t)$ for all $x \in N^1(x_0,\delta')$ and $t$ are
identified from the above argument and Lemma \ref{lem:slope}.
Evaluate (\ref{mgfeq}) at an arbitrary $x \in N^1(x_0,\delta')$ and
solve it to determine $M_1(\cdot)$ and $M_2(\cdot)$.  If $\lambda =
1$, solve (\ref{mgfeq}) directly to identify $M_1$ (or, alternatively,
use the Moore-Penrose generalized inverse as before). 
  Since distribution functions are uniquely
determined by their Laplace transforms (see, for example,  \citeasnoun{feller1968introduction},
p.233), $F_1(\cdot)$ and $F_2(\cdot)$ are uniquely determined.  This
completes the proof.        
\end{proof}

\begin{rem}\label{rem1:id}
To show the above lemma, some regularity conditions on the nature of
$m_1$, $m_2$, $F_1$ and $F_2$ (e.g.
Assumptions \ref{ass-a}(\ref{a-nonpara}),
\ref{ass-b}(\ref{b-0.5})-(\ref{b-2})) are imposed.  Note that such
restrictions are \emph{not} imposed on the parameter set $(0,1] \times \bar{\mathcal F}(\R)^2
  \times \mathcal V( N^1(x_0,\delta') )^2$.  The space of candidate parameters being searched over
  generally contains
  parameter values that violate, say, the non-parallel regression function
  condition as in Assumption \ref{ass-a}(\ref{a-nonpara}).  The
only restrictions imposed on the parameter space are the
independence restriction, which enables us to have $\bar{\mathcal F}(\R)^2$
as the space of the distributions of $\epsilon$'s, and the mean zero
property of $\epsilon$'s, which holds by construction. Lemmas \ref{lem:slope}
  and \ref{lem:id} claim that as far as the \emph{true parameter value}
  $(\lambda, F_1(\cdot), F_2(\cdot), m_1(\cdot), m_2(\cdot))$ satisfies
  the regularity conditions like Assumptions \ref{ass-a}(\ref{a-nonpara}),
\ref{ass-b}(\ref{b-0.5})-(\ref{b-2})), it is uniquely determined in
the unrestricted parameter space $(0,1] \times \bar{\mathcal F}(\R)^2
  \times  \mathcal V( N^1(x_0,\delta') )^2$.  This point
  should be clear from the proof. 
  It is of
  course much easier to establish nonparametric identification by 
  restricting the parameter space we search over, for example, by making
  the parameter space for $m_1$ and $m_2$ the space of pairs of functions
  that are non-parallel.  Such a result
  is not satisfactory from a practical point of view:
  \emph{imposing} conditions such as Assumption 
  \ref{ass-a}(\ref{a-nonpara}) in estimation is difficult in practice.
  This is the reason why this paper considers the more challenging problem
  which removes unnecessary restrictions on the parameter space.  
\end{rem}

\begin{rem}\label{rem1.5:id}
Note that Lemmas \ref{lem:slope} and \ref{lem:id} do not require
$\lambda <1$.  That is, if the true model has $J=1$, the model is
still 
correctly identified (to be a model with just one ``type'' of individuals).  
\end{rem}

\begin{rem}\label{rem2:id}
Some of the assumptions made above are crucial.  The main source of
identification is the independence assumption (Assumption
\ref{ass-a}(\ref{a-ind})), as discussed before.  Also Assumption
\ref{ass-a}(\ref{a-nonpara}) is essential.  If we have $m_1$ and $m_1$
that are completely parallel everywhere, it is easy to see that the ``shift restriction'' implied
by independence loses its identifying power. 
\end{rem}

\begin{rem}
On the other hand, some
of the assumption made here are ``regularity conditions''. First, Assumption
\ref{ass-b}(\ref{b-0.5}) imposes a rather strong assumption requiring that the
moment generating functions $M_1$ and $M_2$ of $F_1$ and $F_2$ exist over $\R$.
Second, Assumption \ref{ass-b}(\ref{b-2}) imposes a very mild condition: see Remark \ref{rem:idass}.  Assumption  \ref{ass:thm1id}  is important for this result, and as discussed earlier, it is testable.  It is satisfied by a large class of parameters, and interestingly, it even includes the case where $F_1$ and $F_2$ are completely identical.    
\end{rem}

\subsection{Second identification result}

\medskip

This section propose an alternative approach for identifying  \eqref{mixmodel}.  One advantage of this second identification result is that it is based on characteristic functions, so their existence is not an issue.  Like the first identification result, the key sufficient condition, which differs from the MGF based condition in the previous section, is testable,     Nonparametric identification holds under
the following alternative set of sufficient conditions.    
\begin{ass}\label{ass-c}
There exist three points $x_a, x_b, x_c$ in $\R^k$ such that  
\begin{enumerate}[\textup{(}i\textup{)}]  \label{ass:2nd}

\item\label{c-ind} 
 $\epsilon_1|x \sim F_1$ and  $\epsilon_2|x \sim F_2$ at all $x =
x_a, x = x_b, x = x_c$, where $F_1$ and $F_2$ do not depend on $x$, 

\item\label{c-nonpara} 
$m_1(x_a) - m_1(x_b) \neq m_2(x_a) - m_2(x_b)$, $m_1(x_a) - m_1(x_c) \neq m_2(x_a) - m_2(x_c)$, and $m_1(x_b)
- m_1(x_c) \neq m_2(x_b) - m_2(x_c)$. 

\end{enumerate}
\end{ass}
Assumption \ref{ass-c} is similar to Assumption \ref{ass-a}, though
here the continuity of $m_1$ and $m_2$ is not an issue.  Next
assumption imposes regularity conditions of the characteristic
functions of $F_1$ and $F_2$, defined by
$$
\phi_i(t) := \int_{\R} e^{it\epsilon} dF_i(\epsilon), \quad i = 1,2.
$$   
\begin{ass}\label{b-1}
$\lim_{t \rightarrow \infty}
\left| \frac {\phi_1(t)} {\phi_2(t)} \right | \rightarrow 0$  or    $ \left| \frac {\phi_2(t)} {\phi_1(t)} \right | \rightarrow 0$ or $\lambda = 1$. 
\end{ass}

It is interesting to compare Assumption \ref{b-1} with Condition \ref{ass:thm1id}.  The former gives a sufficient condition in terms of the characteristic function, whereas the latter the moment generating function.  
It holds, for
example,  if $F_1$ and $F_2$ are the CDFs of $N(0,\sigma_1^2)$,
$N(0,\sigma_2^2)$, $\sigma_1^2 \neq \sigma_2^2$.  \citeasnoun{teicher1963identifiability} 
uses an assumption similar to this.    Assumption \ref{b-1} rules out the case with $F_1 \equiv F_2$, which is allowed by Assumption \ref{ass:thm1id}.  Fortunately, just like  Condition \ref{ass:thm1id}, the new condition Assumption \ref{b-1} is verifiable through the observables, as is clear from the next lemma.  This means which of the two identification strategies to be used can be determined by the observable features of the data.  To state this more precisely, let $\phi(t|x)$ denote the characteristic functions of the conditional
mixture distribution $F(z|x)$, that is,
$$
\phi(t|x) := \int_\R e^{itz}dF(z|x),
$$
and for $x_0 \in \R^k$ define
$$
\rho(x,t) := \frac{\phi(t|x)}{\phi(t|x_0)},  \quad x \in \R^k.
$$
\begin{condition} \label{cond:id2} There exists $\epsilon > 0$ such that 
$$
	\lim_{t \rightarrow \infty}|\rho(x,t)|  = 1
	$$
	and
	$$
	\lim_{t \rightarrow \infty}     \frac {-i} a  \mathrm{Log}\left( \frac{\rho(x,t+a)}{\rho(x,t)} \right) = \mathrm{const.}
	$$
for every $x \in N^1(x_0,\epsilon)$ and $a \in (0,\epsilon]$  where the constant in the second condition may depend on $x$ and $\mathrm{Log}(z)$ denotes the principal value of the complex logarithm of $z \in \mathbb{C}$.
		\end{condition}

\begin{lem}\label{lem:id2cd}  If $m_1$ and $m_2$ are non-parallel on  $N^1(x_0,\epsilon)$,
Assumption \ref{b-1} and Condition \ref{cond:id2} are equivalent.
\end{lem}	
\begin{proof} Define $\delta(x) := m_2(x) - m_1(x_0)$.  Note 
\begin{eqnarray}
\rho(x,t) &=& e^{i t [m_1(x) - m_1(x_0)]}\frac{1 + \frac{1 - \lambda}{\lambda}  e^{i t \delta(x)}  \frac{\phi_2(t)}{\phi_1(t)} }{1 + \frac{1 - \lambda}{\lambda}  e^{i t \delta(x_0)}  \frac{\phi_2(t)}{\phi_1(t)} } \label{eq:1st_decomp}
\\
&=& e^{i t [m_2(x) - m_2(x_0)]}\frac{\frac{1 - \lambda}{\lambda}  e^{-i t \delta(x)}  \frac{\phi_1(t)}{\phi_2(t)} + 1}{\frac{1 - \lambda}{\lambda}  e^{-i t \delta(x_0)}  \frac{\phi_1(t)}{\phi_2(t)} + 1} \label{eq:2nd_decomp}
\end{eqnarray}
The treatment of the case with $\lambda = 1$ is trivial, thus we maintain that $\lambda \in (0,1)$ in the rest of the proof.  
It is enough to prove the necessity, since the sufficiency follows from \eqref{eq:1st_decomp} and \eqref{eq:2nd_decomp}, with the constant in the second condition  being either $m_1(x) - m_1(x_0)$ or $m_2(x) - m_2(x_0)$.   So suppose the necessity fails, i.e. Condition \ref{cond:id2} holds but also  
\begin{equation}\label{eq:rat}
\limsup_{t \to \infty} \left| \frac {\phi_1(t)} {\phi_2(t)} \right | = C, C \in (0,\infty]
\end{equation}
and
\begin{equation}
\limsup_{t \to \infty} \left| \frac {\phi_2(t)} {\phi_1(t)} \right | = C', C' \in (0,\infty].
\end{equation}
hold.   Then if either $C$ or $C'$ is finite (so suppose $C$ is) then there exists a sequence $\{t_k\}_{k=1}^\infty$ such that $\lim_{k \rightarrow \infty} t_k= \infty$ and $\lim_{k \rightarrow \infty} \left |\frac {\phi_1(t_k)} {\phi_2(t_k)} \right | = C$.   But then with the first part of Condition \ref{cond:id2} and  \eqref{eq:2nd_decomp} we have to have 
$$
\lim_{k \rightarrow \infty} \left |\frac{\frac{1 - \lambda}{\lambda}  e^{-i t_k \delta(x)}  \frac{\phi_1(t_k)}{\phi_2(t_k)} + 1}{\frac{1 - \lambda}{\lambda}  e^{-i t_k \delta(x_0)}  \frac{\phi_1(t_k)}{\phi_2(t_k)} + 1} \right|  = 1, x \in  N^1(x_0,\epsilon).
$$   
which holds only if 
$$
\lim_{k \rightarrow \infty} \left[ \mathrm{Arg}\left( \left(  \frac{\phi_1(t_k)}{\phi_2(t_k)} \right ) ^2\right) - 
 \left(t_k [\delta(x) - \delta(x_0)]  + 2 \pi \left \lfloor{\frac{1}{2} - \frac{t_k [\delta(x) - \delta(x_0)] }{2 \pi}}\right \rfloor \right)      \right  ] = 0 
$$
at every $x \in N^1(x_0,\epsilon)$.  Under the non-parallel hypothesis this is impossible.   Finally, if both $C$ and $C'$ are infinite, then  there exits two sequences  $\{t_k\}_{k=1}^\infty$ and  $\{s_k\}_{k=1}^\infty$  such that $\lim_{k \rightarrow \infty} t_k= \infty$, $\lim_{k \rightarrow \infty} s_k= \infty$, $\lim_{k \rightarrow \infty} \left |\frac {\phi_1(t_k)} {\phi_2(t_k)} \right | = \infty$ and  $\lim_{k \rightarrow \infty} \left |\frac {\phi_2(s_k)} {\phi_1(s_k)} \right | = \infty$.  With \eqref{eq:1st_decomp} and \eqref{eq:2nd_decomp}, these imply that  for sufficiently small $a$
$$
\lim_{k \rightarrow \infty}     \frac {-i} a  \mathrm{Log}\left( \frac{\rho(x,t_k+a)}{\rho(x,t_k)} \right) = m_1(x) - m_1(x_0)
$$
and
$$
\lim_{k \rightarrow \infty}     \frac {-i} a  \mathrm{Log}\left( \frac{\rho(x,s_k+a)}{\rho(x,s_k)} \right) = m_1(x) - m_1(x_0)
$$
hold simultaneously, which contradicts the second part of Condition \ref{cond:id2}.
\end{proof}

Finally, assume  
\begin{ass}\label{ass-var}
$\sigma_1^2 := \int \epsilon^2 d F_1(\epsilon)$ and $\sigma_2^2 := \int
\epsilon^2 d F_2(\epsilon)$ are finite. 
\end{ass}
Note that the next lemma holds if the set of the regressors values
includes at least three points.  It therefore allows, for example,
two regressors cases 
where one regressor is binary and the other is continuous. 
\begin{lem}\label{lem:altid}
Under Assumption \ref{b-1} (or Condition \ref{cond:id2}),  as well as Assumptions  \ref{ass-c}  and \ref{ass-var}, $F(\cdot|x)$
at $x = x_a, x_b$ and $x_c$ uniquely determine 
$(\lambda, m_1(x_a), m_1(x_b), m_1(x_c), m_2(x_a), m_2(x_b), m_2(x_c),
F_1(\cdot), F_2(\cdot))$ in the set $\R^7 \times \bar{\mathcal F}(\R)^2$
up to labeling.
\end{lem}
\begin{proof}[\textupandbold{Proof of Lemma~\ref{lem:altid}}]
The proof proceeds in two steps.  Step 1 considers the slopes of $m_1$
and $m_2$.  Using the results in Step 1, Step 2 establishes the
identification of all the parameters.   

\
\

\noindent (Step 1) 

By (\ref{mixmodel})
\begin{equation}
\phi(t|x) = \lambda e^{itm_1(x)}\phi_1(t) + (1 - \lambda)e^{itm_2(x)}\phi_2(t). 
\end{equation}
Suppose there exists an alternative set of parameters 
$$(\lambda^*, m_1^*(x_a), m_1^*(x_b), m_1^*(x_c), m_2^*(x_a), m_2^*(x_b), m_2^*(x_c),
F_1^*(\cdot), F_2^*(\cdot))$$ 
in  $\R^7 \times \bar{\mathcal F}(\R)^2$ such that 
\begin{equation}\label{basic}
F(z|x) = \lambda^* F_1^*(z - m_1^*(x)) + (1 -
\lambda^*)F_2^*(z - m_2^*(x)), \qquad x = x_a, x_b, x_c.  
\end{equation}
Let $\phi_1^*$ and $\phi_2^*$ denote the characteristic functions of
$F_1^*$ and $F_2^*$.  Then 
\begin{align}
\label{n1}
 &\lambda e^{itm_1(x_a)}\phi_1(t) + (1 - \lambda)e^{itm_2(x_a)}\phi_2(t) =
 \lambda^* e^{itm_1^*(x_a)}\phi_1^*(t) + (1 - \lambda^*)e^{itm_2^*(x_a)}\phi_2^*(t),  
\\
\label{n2}
 &\lambda e^{itm_1(x_b)}\phi_1(t) + (1 - \lambda)e^{itm_2(x_b)}\phi_2(t) =
 \lambda^* e^{itm_1^*(x_b)}\phi_1^*(t) + (1 - \lambda^*)e^{itm_2^*(x_b)}\phi_2^*(t),  
\\
\label{n3}
 &\lambda e^{itm_1(x_c)}\phi_1(t) + (1 - \lambda)e^{itm_2(x_c)}\phi_2(t) =
 \lambda^* e^{itm_1^*(x_c)}\phi_1^*(t) + (1 - \lambda^*)e^{itm_2^*(x_c)}\phi_2^*(t).  
\end{align}
Let $\alpha$ and $\beta$ be arbitrary two indices from the index set
$\{a,b,c\}$.  For a function $f:\R^k \rightarrow \R$, 
let $\Delta_{\alpha \beta}f$ denote the differences of the
values of $f$
at $x_\alpha$
and $x_\beta$, that is,
$\Delta_{\alpha\beta}f = f(x_\alpha) - f(x_\beta)$. Define the following function of
$t$ that also depends on functions $f_1:\R^k \rightarrow \R$, $f_2:\R^k \rightarrow
\R$ and indices $\alpha$ and $\beta$:  
\begin{align*}
H(t;f_1,f_2,\alpha,\beta) 
& =  e^{itf_2(x_\alpha)}\left(1 -
  e^{it(\Delta_{\alpha\beta} (f_1 - f_2)}\right)
\\
& = e^{itf_2(x_\alpha)} \left(1 - e^{it\{[f_1(x_\alpha) -
    f_1(x_\beta)]-[f_2(x_\alpha) - f_2(x_\beta)]\}}\right). 
\end{align*}
Now, multiply (\ref{n2}) by $e^{it\Delta_{ab}m_2^*}$ then subtract
both sides from (\ref{n1}) to obtain
\begin{equation}
\label{n4'}
\lambda H(t;m_2^*,m_1,a,b)\phi_1(t) + (1 -
\lambda)H(t;m_2^*,m_2,a,b)\phi_2(t) 
= \lambda^*H(t;m_2^*,m_1^*,a,b) \phi_1^*(t). 
\end{equation}
Repeat this with (\ref{n2}) and  $e^{it\Delta_{ab}m_2^*}$ replaced by 
(\ref{n3}) and $e^{it\Delta_{ac}m_2^*}$: 
\begin{equation}
\label{n5'}
\lambda H(t;m_2^*,m_1,a,c)\phi_1(t) + (1 - \lambda)H(t;m_2^*,m_2,a,c)\phi_2(t) 
= \lambda^*H(t;m_2^*,m_1^*,a,c) \phi_1^*(t). 
\end{equation}
(\ref{n4'}) and (\ref{n5'}) imply 
\begin{align}
\label{n4''}
\lambda  H(t;m_2^*,m_1,a,b) H(t;m_2^*,m_1^*,a,c)\phi_1(t) 
+ (1 - \lambda) H(t;m_2^*,m_2,a,b)H(t;m_2^*,m_1^*,a,c)\phi_2(t) 
\\
= \lambda^*H(t;m_2^*,m_1^*,a,b)H(t;m_2^*,m_1^*,a,c) \phi_1^*(t),
\nonumber
\end{align}
and
\begin{align}
\label{n5''}
\lambda H(t;m_2^*,m_1^*,a,b)H(t;m_2^*,m_1,a,c)\phi_1(t) + (1 -
\lambda)H(t;m_2^*,m_1^*,a,b)H(t;m_2^*,m_2,a,c)\phi_2(t)
\\  
= \lambda^*H(t;m_2^*,m_1^*,a,b)H(t;m_2^*,m_1^*,a,c) \phi_1^*(t),  
\nonumber
\end{align}
yielding 
\begin{align*}
\lambda  H(t;m_2^*,m_1,a,b)H(t;m_2^*,m_1^*,a,c) \phi_1(t) 
+ (1 - \lambda) H(t;m_2^*,m_2,a,b)H(t;m_2^*,m_1^*,a,c)\phi_2(t) 
\\
= \lambda  H(t;m_2^*,m_1^*,a,b)H(t;m_2^*,m_1,a,c) \phi_1(t) + (1 -
\lambda)H(t;m_2^*,m_1^*,a,b)H(t;m_2^*,m_2,a,c)\phi_2(t),
\nonumber
\end{align*}
or 
\begin{align}
\label{nend}
&\lambda \left [ H(t;m_2^*,m_1^*,a,b)H(t;m_2^*,m_1,a,c) -  H(t;m_2^*,m_1,a,b)H(t;m_2^*,m_1^*,a,c) 
  \right ] \phi_1(t)
\\
&= 
(1 - \lambda) 
\left [ 
H(t;m_2^*,m_1^*,a,b)H(t;m_2^*,m_2,a,c) - 
 H(t;m_2^*,m_2,a,b)H(t;m_2^*,m_1^*,a,c)
\right ] \phi_2(t).
\nonumber
\end{align}
Divide both sides of (\ref{nend}) by $e^{m_1^*(x_a)}$ and rewriting:
\begin{align*}
&\lambda e^{itu_1}\left [(1 - e^{it u_{11}})(1 - e^{it u_{12}}) -
  (1 - e^{itu_{13}})(1 - e^{itu_{14}}) \right ] \phi_1(t)
\\
&=  (1 - \lambda) e^{itu_2}\left [(1 - e^{it u_{21}})(1 - e^{it u_{22}}) -
  (1 - e^{itu_{23}})(1 - e^{itu_{24}}) \right ] \phi_2(t) \qquad
  \text{for all } t 
\end{align*}
where $u_1 = m_1(x_a)$, $u_2 = m_2(x_a)$, 
$u_{11} = \Delta_{ab}(m_2^* - m_1^*)$, 
$u_{12} = \Delta_{ac}(m_2^* - m_1)$,  
$u_{13} = \Delta_{ab}(m_2^* - m_1)$,   
$u_{14} = \Delta_{ac}(m_2^* - m_1^*)$, 
$u_{21} = \Delta_{ab}(m_2^* - m_1^*) = u_{11}$ , 
$u_{22} = \Delta_{ac}(m_2^* - m_2)$,  
$u_{23} = \Delta_{ab}(m_2^* - m_2)$,   
$u_{24} = \Delta_{ac}(m_2^* - m_1^*) = u_{14}$. 

First, consider the non-degenerate case, i.e. $\lambda \neq 1$. Define 
$$
L_1(t) = (1 - e^{it u_{11}})(1 - e^{it u_{12}}) -
  (1 - e^{itu_{13}})(1 - e^{itu_{14}})
$$
and 
$$
L_2(t) = (1
- e^{it u_{21}})(1 - e^{it u_{22}}) - 
  (1 - e^{itu_{23}})(1 - e^{itu_{24}}),  
$$
then 
\begin{equation}\label{L1L2}
L_1(t) =  e^{it(u_2 - u_1)}\frac {1 - \lambda} \lambda \frac {\phi_2(t)} {\phi_1(t)}
L_2(t)  \qquad  \text{for all } t.   
\end{equation}
We now use the condition  
\begin{equation}\label{ratio}
\lim_{t \rightarrow \infty} \frac {\phi_2(t)}{\phi_1(t)} = 0
\end{equation}
from Assumption \ref{b-1}
(the treatment of the case with $\lim_{t \rightarrow \infty} \frac
{\phi_1(t)}{\phi_2(t)} = 0$ is essentially identical).  
The following argument shows that 
\begin{equation}\label{L1iszero}
L_1(t) = 0  \qquad  \text{for all } t. 
\end{equation}
Suppose (\ref{L1iszero}) is false, i.e. suppose the set $A = \{t:
 L_1(t) \neq
 0, \quad t \in \R\}$ is non-empty.   Pick an arbitrary point $t_0$
 from $A$.  Then there exists an $\epsilon > 0$ such that $|L_1(t_0)| \geq \epsilon >
 0$.  But since $\lim_{t
 \rightarrow \infty}  e^{it(u_2 - u_1)}\frac {1 - \lambda} \lambda \frac {\phi_2(t)}
 {\phi_1(t)} L_2(t) = 0$ under (\ref{ratio}), together with
 (\ref{L1L2}), there exists $t_1(\epsilon) \in \R$ 
 such that 
\begin{equation}\label{L2bound}
|L_1(t)| < \frac \epsilon 2 \qquad \text{for all } t > t_1(\epsilon).
\end{equation}
 Because of the definition of $t_0$, it must be the case that $t_0 \leq t_1(\epsilon)$.
 Now, since $L_1(\cdot)$ is a sum of periodic functions, it is almost
 periodic (see, e.g.  \citeasnoun{dunford1958linear})).  Therefore there
 exists a positive number $l(\epsilon)$ such that for all $\tau \in
 \R$ one can find a $\xi(\tau,\epsilon,l(\epsilon)) \in [\tau,\tau+ l(\epsilon)]$ such that 
\begin{equation}\label{AP}
|L_1(t) - L_1(t +  \xi(\tau,\epsilon,l(\epsilon)))| < \frac \epsilon 2
 \qquad  \text{for all } t \in \R.       
\end{equation} 
In particular, evaluating (\ref{AP}) at $t = t_0$ and $\tau = -t_0 + t_1(\epsilon)$;
\begin{equation}\label{AP'}
|L_1(t_0) - L_1(t_0 +  \xi^*)|
 < \frac \epsilon 2
\end{equation} 
where $\xi^* =  \xi(-t_0+t_1(\epsilon),\epsilon,l(\epsilon)))$.
But $\xi^* \in [-t_0 + t_1(\epsilon), -t_0 +
t_1(\epsilon) + l(\epsilon)]$, therefore $t_0 +
\xi^* \leq t_0 - t_0 + t_1(\epsilon) =
t_1(\epsilon)$.  By (\ref{L2bound}), 
\begin{equation}\label{L1bound}
|L_1(t_0 +  \xi^*)| < \frac \epsilon 2
\end{equation}
Using the triangle inequality, (\ref{AP'}) and (\ref{L1bound}),
conclude that 
\begin{align*}
|L_1(t_0)| & \leq |L_1(t_0) - L_1(t_0 + \xi^*)| + |L_1(t_0 + \xi^*)|
\\
& < \epsilon.
\end{align*}
But the $\epsilon$ was originally defined so that $|L_1(t_0)| \geq
\epsilon$, contradicting the last inequality.  Since the choice of
$t_0 \in A$ was arbitrary, (\ref{L1iszero}) is now proved.   

Next, as $\lambda \neq 0$,
(\ref{L1L2}) and (\ref{L1iszero}) imply that 
$$
\phi_2(t) L_2(t) = 0 \qquad \text{for all } t.
$$ 
But by the basic properties a characteristic function, $\phi_2(\cdot)$
is continuous and $\phi_2(1) = 0$.  Therefore for a $d > 0$,
$\phi_2(t) \neq 0$ for all $t \in [-d,d]$.  It follows that $L_2(t) = 0$ for
all $t \in [-d,d]$.  Moreover, $L_2(t)$ is analytic on the entire
complex plane, and $[-d,d]$ obviously has an accumulation point,
therefore by the identity theorem of analytic functions, $L_2(t) = 0$ for all
$t \in \R$.  
In summary, $L_1(t) = L_2(t) = 0$ for all $t \in \R$, or:    
\begin{equation}\label{e1}
(1 - e^{it u_{11}})(1 - e^{it u_{12}}) -
  (1 - e^{itu_{13}})(1 - e^{itu_{14}}) = 0
\end{equation}
and 
\begin{equation}\label{e2}
(1 - e^{it u_{21}})(1 - e^{it u_{22}}) -
  (1 - e^{itu_{23}})(1 - e^{itu_{24}}) = 0
\end{equation}
for all $t$.  These conditions in turn identify the slopes of $m_1$
and $m_2$, as shown by the subsequent argument.  

Consider the following set of conditions  
\begin{align}
\label{c1}
\Delta_{ab}(m_2^* - m_1^*) = \Delta_{ab}(m_2^* - m_1) \text{ and } 
\Delta_{ac}(m_2^* - m_1^*) = \Delta_{ac}(m_2^* - m_1) \tag{C1}, 
\\
\label{c2}
\Delta_{ab}(m_2^* - m_1^*) = \Delta_{ac}(m_2^* - m_1^*) \text{ and } 
\Delta_{ab}(m_2^* - m_1) = \Delta_{ac}(m_2^* - m_1) \tag{C2},
\\
\label{c3}
\Delta_{ab}(m_2^* - m_1^*) = \Delta_{ab}(m_2^* - m_2) \text{ and } 
\Delta_{ac}(m_2^* - m_1^*) = \Delta_{ac}(m_2^* - m_2) \tag{C3}, 
\\
\label{c4}
\Delta_{ab}(m_2^* - m_1^*) = \Delta_{ac}(m_2^* - m_1^*) \text{ and } 
\Delta_{ab}(m_2^* - m_2) = \Delta_{ac}(m_2^* - m_2) \tag{C4}.  
\end{align}
Then by (\ref{e1}) and (\ref{e2}), if $u_{jk} \neq 0$ for all $j = 1,2, k =
1,2,3,4$, one of the following four cases has to be true:

\noindent (D1):  (\ref{c1}) and (\ref{c3}) hold;

\noindent (D2):  (\ref{c1}) and (\ref{c4}) hold; 

\noindent (D3):  (\ref{c2}) and (\ref{c3}) hold;   

\noindent (D4):  (\ref{c2}) and (\ref{c4}) hold. 

\noindent First, consider (D1).  (\ref{c1}) and (\ref{c3}) imply $\Delta_{ab}m_1^*
= \Delta_{ab}m_1$ and  $\Delta_{ab}m_1^*
= \Delta_{ab}m_2$, respectively, thereby yielding  $\Delta_{ab}m_1
= \Delta_{ab}m_2$, which violates Assumption
\ref{ass-c}(\ref{c-nonpara}).  Next, turn to (D2).  From (\ref{c1})
get $\Delta_{ab} m_1 = \Delta_{ab} m_1^*$ and $\Delta_{ac}m_1 =
\Delta_{ac}m_1^*$, therefore $\Delta_{bc}m_1 = \Delta_{bc}m_1^*$. But
(\ref{c4}) also implies  $\Delta_{bc} m_2^* = \Delta_{bc} m_1^*$ and $\Delta_{bc}m_2^* =
\Delta_{bc}m_2$, hence  $\Delta_{bc}m_1 = \Delta_{bc}m_2$, violating  Assumption
\ref{ass-c}(\ref{c-nonpara}).  Since (D3) is identical to (D2) except
for the switched roles of $m_1$ and $m_2$, it also violates Assumption
\ref{ass-c}(\ref{c-nonpara}).   Finally, (D4) also leads to a
violation of   Assumption \ref{ass-c}(\ref{c-nonpara}), because
the second equations of (\ref{c2}) and (\ref{c4}) yield  $\Delta_{bc}
m_1 = \Delta_{bc} m_2$.  As (D1)-(D4) are impossible, some of the $u_{jk}$'s should
be non-zero.  To consider the cases with some non-zero $u_{jk}$, it is useful to introduce the following
classification (note that for $i = 1,2$, if $u_{ij} = 0$ for $j = 1$ or 2 (3 or 4), then $u_{ij} =
0$ for $j = 3$ or 4 (1 or 2),  

\noindent Case (i): $u_{11} = 0$ 

\noindent Case (ii): $u_{12} = u_{13} = 0$

\noindent Case (iii): $u_{14} = 0$ 

\noindent Case (iv): $u_{21} = 0$ 

\noindent Case (v): $u_{22} = u_{23} = 0$

\noindent Case (vi): $u_{24} = 0$ 

First consider Case (i).  Then 
$
H(t,m_2^*,m_1^*,a,b) = e^{itm_1^*(x_a)}(1 - e^{it(\Delta(m_2^* -
  m_1^*))})   = 0.    
$
Therefore (\ref{n4'}) becomes 
\begin{equation}
\lambda H(t;m_2^*,m_1,a,b)\phi_1(t) + (1 -
\lambda)H(t;m_2^*,m_2,a,b)\phi_2(t) 
= 0, 
\end{equation} 
or 
\begin{equation}
(1 - e^{it\Delta_{ab}(m_2^* - m_1)})
+ \frac {1 - \lambda} \lambda \frac {\phi_2(t)} {\phi_1(t)}e^{-itm_1(x_a)}H(t;m_2^*,m_2,a,b) 
= 0.
\end{equation} 
Let $t \rightarrow \infty$, then again by (\ref{ratio}), the third 
term goes to zero.  Since the first term is periodic, it must be the
case that $u_{13} = 0$ for all $t \in \R$.  Since $1 - \lambda \neq 0$
in the current analysis of the non-degenerate case, 
$H(t;m_2^*,m_2,a,b)\phi_2(t) 
= 0$ for all $t$, or, 
$$
(1 - e^{itu_{23}})\phi_2(t) = 0 \qquad \text{for all }t. 
$$        
As argued before, this means
$$
1 - e^{itu_{23}} = 0 \qquad \text{for } t \in [-d,d]
$$
for some $d > 0$.  But this is possible iff $u_{23} = 0$.  In sum,
$u_{11} = 0$ automatically implies that $u_{13} = u_{23} =  0$ as well.
But the latter condition means $\Delta_{ab}(m_2^* - m_1) = 0$ and
$\Delta_{ab}(m_2^* - m_2) = 0$, which in turn imply $\Delta_{ab}(m_1 -
m_2) = 0$, thereby violating Assumption \ref{ass-c}(\ref{c-nonpara}).

Next, consider Case (ii).  This case means that 
\begin{align}\label{IDeq1}
\Delta_{ab}m_2^* &= \Delta_{ab}m_1,
\\
\Delta_{ac}m_2^* &= \Delta_{ac}m_1.\nonumber
\end{align}        
On top of this, (\ref{e2}) has to hold at the same time.  First,
suppose all $u_{2k}, k = 1,2,3,4$ in (\ref{e2}) are non-zero.  Then
(\ref{c3}) and/or (\ref{c4}) has to hold.  Suppose (\ref{c3}) holds.
Then
\begin{align}\label{IDeq2}
\Delta_{ab}m_1^* &= \Delta_{ab}m_2,
\\
\Delta_{ac}m_1^* &= \Delta_{ac}m_2.\nonumber
\end{align}    
(\ref{IDeq1}) and (\ref{IDeq2}) imply that slopes of $m_1^*$ and
$m_2^*$ have to coincide with those of $m_2$ and $m_1$, respectively,
proving a part of the identification result.  Next, suppose
(\ref{c4}) holds.  In particular, the second equation of (\ref{c4}),
together with (\ref{IDeq1}) means that $\Delta_{ab}(m_1 - m_2) =
\Delta_{ac}(m_1 - m_2)$, or $\Delta_{bc}m_1 = \Delta_{bc}m_2$,
violating Assumption \ref{ass-c}(\ref{c-nonpara}).  To complete the
analysis of Case (ii), now suppose some of $u_{2k}, k = 1,2,3,4$ in
(\ref{e2}) are zero.  If $u_{21} = 0$, then $u_{11} = 0$, but we have
already shown that the latter condition leads to a violation of
Assumption \ref{ass-c}(\ref{c-nonpara}).  Next, suppose $u_{22} = 0$,
i.e., $\Delta_{ac}m_2^* = \Delta_{ac}m_2$.  But with the
second equation of (\ref{IDeq1}),  $\Delta_{ac}m_1 = \Delta_{ac}m_2$,
again violating Assumption \ref{ass-c}(\ref{c-nonpara}).  If $u_{23} =
0$ or $u_{24} = 0$, it means at least one of $u_{21}$ and $u_{22}$
must be zero, so the above argument covers the cases.  This completes
the analysis of Case (ii); in sum, Case (ii) implies (\ref{IDeq1}) and (\ref{IDeq2}).  

 Case (iii) is identical to Case (i), with
the roles of the indices $b$ and $c$ switched, therefore it violates
Assumption \ref{ass-c}(\ref{c-nonpara}).  Case (iv) is identical to
Case (i).   Note that case (v) is identical to Case (ii) with the
role of the functions $m_1$ and $m_2$ reversed.  But the treatment of
Case (ii) only uses Equations (\ref{e1}) and (\ref{e2}), which are
equivalent to 
(\ref{e2}) and (\ref{e1}), respectively, after switching $m_1$ and $m_2$.  Therefore
the above treatment of Case (ii) applies with $m_1$ and $m_2$
reversed; that is, Case (v) implies that             
\begin{align}\label{IDeq3}
\Delta_{ab}m_1^* &= \Delta_{ab}m_1,
\\
\Delta_{ac}m_1^* &= \Delta_{ac}m_1.\nonumber
\end{align}
and    
\begin{align}\label{IDeq4}
\Delta_{ab}m_2^* &= \Delta_{ab}m_2,
\\
\Delta_{ac}m_2^* &= \Delta_{ac}m_2.\nonumber
\end{align}
Finally, Case (vi) is identical to Case (vi).    

The above arguments prove that if the mixture model is non-degenerate, the only
possible cases are either (A): (\ref{IDeq1}) and (\ref{IDeq2}) hold,
or (B): (\ref{IDeq3}) and (\ref{IDeq4}) hold.   
That is, the slopes of $m_1$ and $m_2$ are identified, up to
labeling.  

Next consider the case where the mixture model is degenerate,
i.e. $\lambda = 1$.  Then (\ref{basic}) is now written as
\begin{equation}\label{basic2}
F_1(z - m_1(x)) = \lambda^* F_1^*(z - m_1^*(x)) + (1 -
\lambda^*)F_2^*(z - m_2^*(x)).  
\end{equation}
Define 
${\sigma_1^*}^2 = \int \epsilon^2 F_1^*(d\epsilon)$ and
${\sigma_2^*}^2 = \int \epsilon^2 F_2^*(d\epsilon)$.           
Taking the conditional variance of both sides given $x$,
\begin{align*}
\sigma_1^2 &= \lambda^*(m_1^*(x)^2 + {\sigma_1^*}^2)  + (1 - \lambda^*)
(m_2^*(x)^2 + {\sigma_2^*}^2) - [\lambda^*m_1^*(x) + (1 - \lambda^*)m_2^*(x)]^2 
\\
&=\lambda^*(1 - \lambda^*)[m_1^*(x) - m_2^*(x)]^2 +  \lambda^*{\sigma_1^*}^2  + (1 - \lambda^*)
{\sigma_2^*}^2 \qquad \text{ at } x = x_a, x_b \text{ and } x_c. 
\end{align*} 
This equation is used to establish identification for the degenerate
case.  In particular, it admits two solutions:
\begin{align}
\label{varid1}
&\lambda^* = 1, \quad {\sigma_1^*}^2 = \sigma_1^2,
\\
\label{varid2}
&[m_1^*(x_a) - m_2^*(x_a)]^2 = [m_1^*(x_b) - m_2^*(x_b)]^2 = [m_1^*(x_c) - m_2^*(x_c)]^2.
\end{align}
(\ref{varid1}) obviously leads to full identification: integrating
both sides of (\ref{basic2}) gives $m_1^*(x) = m_1(x)$, and this
trivially determines $F_1^*(z) = F_1(z)$ for all $z$. (\ref{varid2})
implies that, for at least one pair of points, $(x,x')$, say, out of
the three points $\{x_a,x_b,x_c\}$, the following holds: 
\begin{equation}\label{parID}
m_1^*(x) - m_2^*(x) = m_1^*(x') - m_2^*(x').
\end{equation}
Unlike the case with $\lambda < 1$, this does not fully determine the
slopes of $m_1$ and $m_2$ over  $\{x_a,x_b,x_c\}$; it will be done in
(Step 2).   

\
\

\noindent (Step 2)

\
\

We now argue that $\lambda$ is identified whether the model is
degenerate or not.  Let $m_j^*(x), j = 1,2, x = x_1, x_b, x_c$ be (arbitrary) six
numbers that satisfy (\ref{basic}).  By (Step 1), in the
case  $\lambda \neq 1$,  they have
to satisfy (\ref{IDeq1}) and (\ref{IDeq2}), or, (\ref{IDeq3}) and
(\ref{IDeq4}).  Similarly, in the case  $\lambda = 1$, they have to
satisfy (\ref{parID}) (the case with (\ref{varid1}) is trivial).     
For an arbitrary pair of points $(x,x')$ from
the three support points ${x_a,x_b,x_c}$, define
$$
\lambda(x,x') = \lim_{\delta \downarrow 0} \frac{\int z F^*(dz|x) -
\int z F^*(dz|x') - (1 + \delta)(m_2^*(x) - m_2^*(x'))}{(m_1^*(x) -
m_1^*(x')) - (1 + \delta)(m_2^*(x) - m_2^*(x'))}.  
$$
Then $\lambda$ is uniquely determined from the values $m_j^*(x), j =
1,2, x = x_1, x_b, x_c$ by
\begin{equation}\label{lambdaID}
\max_{(x,x') = (x_a,x_b), (x_a,x_c), (x_b,x_c)}
\lambda(x,x'),
\end{equation}
using an argument as in the proof of Lemma \ref{lem:id}, up to
labeling  It holds whether $\lambda < 1$ or not.  (Note that
the maximization in the line above is unnecessary if $\lambda \neq 1$,
since $\lambda(x,x')$ is identical for all pairs $(x, x')$ in
that case.)  Let $(\bar x,\bar x')$ be a maximizer of
(\ref{lambdaID}), which is possibly not unique.

Now, evaluating (\ref{mID}) at
$(\bar x,\bar x')$ and $(\bar x',\bar x)$, instead of $(x,x_0)$ and
solving for $m_1$ and $m_2$, obtain $m_1(\bar x)$, $m_2(\bar x)$,
$m_1(\bar x')$ and $m_2(\bar x')$ ($m_1(\bar x)$ and $m_1(\bar x')$ in
the degenerate case).      
     
To identify $F_1$ and $F_2$, use      
\begin{equation}\label{cfeq}
\begin{bmatrix}
\phi(t|\bar x) \\ \phi(t|\bar x')
\end{bmatrix}
= G(\bar x,\bar x',t) \Lambda
\begin{bmatrix}
\phi_1(t) \\ \phi_2(t)
\end{bmatrix},
\end{equation}
where
\[
G(x,x',t) = 
\begin{pmatrix}
e^{itm_1(x)} & e^{itm_2(x)}
\\
e^{itm_1(x')} & e^{itm_2(x')} 
\end{pmatrix}, 
\Lambda = 
\begin{pmatrix}
\lambda &  0
\\
0 & (1 - \lambda) 
\end{pmatrix},
\]
instead of (\ref{mgfeq}) in the proof of Lemma \ref{lem:id}.  
Then
\begin{align*}
\text{Det}(G(\bar x,\bar x,t))&  = e^{it[m_1(\bar x) + m_2(\bar x')]}
- e^{it[m_1(\bar x') + m_2(\bar x)]}
\\
& =  e^{it[m_1(\bar x) + m_2(\bar x')]}\left(1 - e^{it\{[m_1(\bar x) -
    m_1(\bar x')] - [m_2(\bar x')
    - m_2(\bar x)] }\right )
\\
& \neq  0 \quad \text{ for all } t \neq \frac {2\pi j}{[m_1(\bar x) -
    m_1(\bar x')] - [m_2(\bar x')
    - m_2(\bar x)]}, \quad j \in \mathbb{Z}.
\end{align*}
under Assumption \ref{ass-c}(\ref{c-nonpara}) if $\lambda < 1$,
therefore $G(\bar x,\bar x,t)$ is invertible (and all of
its elements are identified).  This determines $\phi_1(t)$ and
$\phi_2(t)$ for all  $t \neq  0  \text{ for all } t \neq \frac {2\pi j}{[m_1(\bar x) -
    m_1(\bar x')] - [m_2(\bar x')
    - m_2(\bar x)]}, j \in \mathbb{Z}$ (as before, solve (\ref{cfeq}) directly or
using the Moore-Penrose inverse in the degenerate case to determine
$\phi_1$). 
 Since  $\phi_1(t)$ and
$\phi_2(t)$ are continuous, they are identified on $\R$.  This
identifies $F_1$ and $F_2$.  

The foregoing argument shows that $\lambda$, $F_1(\cdot)$, $F_2(\cdot)$ and $m_1(x)$
and $m_2(x)$ evaluated at two points (i.e. $\bar x$ and $\bar x'$ defined right
after (\ref{lambdaID})) out of the three support points
$\{x_a,x_b,x_c\}$, are identified  Note that  $m_1$
and $m_2$ at the third point (= $\tilde x$, say) is identified by the
relation
$$ 
\phi(t|\tilde x) = \lambda e^{itm_1(\tilde x)}\phi_1(t) + (1 -
\lambda)e^{itm_2(\tilde x)}\phi_2(t) \quad \text{ for all }t.             
$$
Let 
\begin{align*}
g_\tau (t) & = \frac {\phi(t + \tau|\tilde x)}{\lambda \phi_1(t + \tau)}
\\
& =  e^{i(t + \tau)m_1(\tilde
    x)} + \frac{1 - \lambda}{\lambda} e^{i(t+ \tau)m_2(\tilde x)} \frac
  {\phi_2(t + \tau)}{\phi_1(t + \tau)}, 
\end{align*}
Then under (\ref{ratio}) the second term converges to zero as $\tau
    \rightarrow \infty$, and if we write, for all $c$, 
\begin{align*}
h(c) & = \lim_{\tau \rightarrow \infty} \frac{g_\tau (t+c)}{g_\tau (t)} = e^{ic m_1(\tilde x)},
\end{align*}    
then $m_1(\tilde x)$ is uniquely determined by the formula $m_1(\tilde x) = \frac{ -i h'(c)}{h(c)}  $.

If the model is non-degenerate,
    $m_2(\tilde x)$ is identified from $e^{itm_2(\tilde x)} = \frac
    {\phi(t|\tilde x) - \lambda e^{itm_1(\tilde x)}\phi_1(t)}{(1 - \lambda)\phi_2(t)}$.
\end{proof} 

\begin{rem}\label{rem:idoverarea}  
Once identification is achieved at some values of $x$, as implied by Lemmas
\ref{lem:id} and \ref{lem:altid}, the complete knowledge of $M_1$ and $M_2$ is
available.  Since the identity for 
conditional characteristic functions or conditional moment generating functions
as in (\ref{mixedmgf}) holds for all $t$, it can be used to
determine $m_1$ and $m_2$ even at points where they fail
to satisfy the non-parallel condition (i.e. Assumption \ref{ass-a}(\ref{a-nonpara}) or
\ref{ass-c}(\ref{c-nonpara})).  Suppose $F(\cdot|x)$ is known on a set
${\mathcal X} \in \R^k$.   Assume that, for example, Assumptions
\ref{ass-a}(\ref{a-ind}) and \ref{ass-b} hold.  Then $F(\cdot|x),
x\in \mathcal X$ uniquely determines 
$(\lambda,F_1(\cdot),F_2(\cdot),m_1(x),m_2(x))$ for all $x \in
\mathcal X$ up to labeling, unless $\lambda = 1 - \lambda = \frac 1
2$ and $F_1(z) = F_2(z)$ for all $z \in \R$.
\end{rem}

\subsection{Third identification result}\label{subsec:third}

\medskip

We now propose an identification strategy that has an approach similar to the first identification result, though differs from it in some important ways.  It uses one sided limit (e.g. $t$ tending to {\it positive} infinity) of MGFs and also characteristic functions.   Unlike our first result, it for instance addresses the case where $F_1$ and $F_2$ are CDFs of $N(0,\sigma_1^2)$, $N(0,\sigma_2^2)$, $\sigma_1^2 \neq \sigma_2^2$.  Moreover, the identification strategy for the distribution functions avoids Laplace inversion, a problematic step in practice.  For these reasons it is the identification strategy in this section that will be used to construct our estimator in Section \ref{sec:estimation}.

Recall our definition of the function $h(\cdot,\cdot)$ (see Assumption \ref{ass-b}) in the statements of the following assumption.  
\begin{ass}\label{ass:thm3}
	\begin{enumerate}[\textup{(}i\textup{)}]

		\item\label{ass:thm31} The domains of $M_1(t)$ and $M_2(t)$ include  $[0,\infty)$ and  for some $\varepsilon >0$ either $h(\pm\varepsilon,t) = O(1)$ or $1/ h(\pm\varepsilon,t) = O(1)$ or both hold as $t \rightarrow  +\infty$,
		
		or

		\item\label{ass:thm32} The domains of $M_1(t)$ and $M_2(t)$ include $(-\infty,0]$ and  for some $\varepsilon >0$ either $h(\pm\varepsilon,t) = O(1)$ or $1/ h(\pm\varepsilon,t) = O(1)$ or both hold as  $t \rightarrow - \infty$.

	\end{enumerate}
\end{ass}

Note that this assumption does not demand  the MGFs $M_1$ and $M_2$ to be defined on the whole real line, sometimes a restrictive assumption.

\begin{lem}\label{lem:3rdslope}   Suppose Assumptions \ref{ass-a}, \ref{b-1} and \ref{ass:thm3} hold.  Then there exists  $\epsilon \in (0,\delta)$
such that	for every $x \in N^1(x_0,\epsilon)$ and $a \in (0,\epsilon]$

	\begin{enumerate}[\textup{(}i\textup{)}]
		
		\item\label{lem:3rdslope1} $\lim_{t \rightarrow \infty}\frac 1 t \log  R(t,x') = m_1(x') - m_1(x_0) $ or  $	\lim_{t \rightarrow \infty}\frac 1 t \log  R(t,x') = m_2(x') - m_2(x_0)$   if Assumption \ref{ass:thm3}\eqref{ass:thm31} holds, and
		
\noindent  $\lim_{t \rightarrow -\infty}\frac 1 t \log  R(t,x') = m_1(x') - m_1(x_0) $ or  $	\lim_{t \rightarrow -\infty}\frac 1 t \log  R(t,x') = m_2(x') - m_2(x_0)$  if Assumption \ref{ass:thm3}\eqref{ass:thm32} holds instead.

		\item\label{lem:3rdslope2}  	$\lim_{t \rightarrow \infty}     \frac {-i} a  \mathrm{Log}\left( \frac{\rho(x,t+a)}{\rho(x,t)} \right) = m_1(x') - m_1(x_0) $ or  $\lim_{t \rightarrow \infty}     \frac {-i} a  \mathrm{Log}\left( \frac{\rho(x,t+a)}{\rho(x,t)} \right) =  m_2(x') - m_2(x_0)$.

	\end{enumerate}
\end{lem}

\begin{proof}  The proof of Part \eqref{lem:3rdslope1} is essentially in the proof of 
Lemma \eqref{lem:slope}.  For Part \eqref{lem:3rdslope2},  note that the ratios on the right hand side of  \eqref{eq:1st_decomp} and by \eqref{eq:2nd_decomp} converge to $1$ as $s \to \infty$.   Since
$$\rho(s,x)  = e^{i s \nabla} \frac{\frac{\lambda}{1-\lambda} e^{i s (m_1(x_1) - m_2(x_1))} \frac{\phi_1(s)}{\phi_2(s)} + 1}{\frac{\lambda}{1-\lambda} e^{i s (m_1(x_0) - m_2(x_0))} \frac{\phi_1(s)}{\phi_2(s)} + 1},$$
and under Assumption  \ref{b-1}   the ratio on the right hand side converges to $1$ as $s \to \infty$. Therefore we have
\begin{align*}
\lim_{s \rightarrow \infty} \frac{-i}{a}  \mbox{ Log} \left( \frac{\rho(x,s+a)}{\rho(x,s)}    
\right)
& = \frac{-i}{a} \mbox{ Log}  (e^{ia\nabla})\\
& = \frac{1}{a} \left( a \nabla + 2 \pi \left \lfloor{\frac{1}{2} - \frac{a \nabla}{2 \pi}}\right \rfloor \right) ,
\end{align*}
where $\mbox{ Log}$ corresponds to the principal value of the log. This limit is a piecewise continuous function of $a$, constant equal to $\nabla$ only when $a$ is small enough to guarantee $a\nabla \in (-\pi,\pi).$
And if $\lambda=1$, $\frac{\phi(s|x_1)}{\phi(s|x_0)}= e^{i s \Delta}$ so that $\lim_{s \rightarrow \infty} \frac{-i}{a}  \mbox{ Log} \left( \frac{\phi(s+a|x_1)}{\phi(s+a|x_0)} \left( \frac{\phi(s|x_1)}{\phi(s|x_0)} \right)^{-1}  \right)
 = \frac{1}{a} \left( a \Delta + 2 \pi \left \lfloor{\frac{1}{2} - \frac{a \Delta}{2 \pi}}\right \rfloor \right) .$ By assumption, $m_1(x_1) - m_1(x_0) \neq m_2(x_1) - m_2(x_0)$ that is, $\Delta \neq \nabla$ therefore if the former limit is equal to $\Delta$, one knows $\lambda =1$ and there is no $m_2$.
\end{proof}

  The constant $\delta$ in the following condition is specified in Assumption  \ref{ass-a}.
\begin{condition}\label{ass:thm3id}  Either 
	\begin{enumerate}[\textup{(}i\textup{)}]

		\item\label{ass:thm3id1}   there exists $\epsilon \in (0,\delta)$	such that $
		\lim_{t \rightarrow \infty}\frac 1 t \log  R(t,x)  \neq 
		\lim_{s \rightarrow \infty} \frac{-i}{a}  \mbox{\rm  Log} \left( \frac{\rho(x,s+a)}{\rho(x,s)}    
		\right) \text{ for every} x \in N^1(x_0,\epsilon)$ and $a \in (0,\epsilon]$   if Assumption \ref{ass:thm3}\eqref{ass:thm31} holds
		
		or 
		
			\item\label{ass:thm3id1}   there exists $\epsilon \in (0,\delta)$	such that $
		\lim_{t \rightarrow \infty}\frac 1 t \log  R(t,x)  \neq 
		\lim_{s \rightarrow \infty} \frac{-i}{a}  \mbox{\rm Log} \left( \frac{\rho(x,s+a)}{\rho(x,s)}    
		\right) \text{ for every } x \in N^1(x_0,\epsilon)$ and $a \in (0,\epsilon]$   if Assumption \ref{ass:thm3}\eqref{ass:thm32} holds

		or 
		
		\item\label{ass:thm3id2}  $\lim_{\delta \downarrow 0} \lambda_{\delta} = 1$
		
		holds.  	
		
	\end{enumerate} 	
\end{condition}
 The above condition is verifiable with information in the observables as $\rho$, $R$ and $\lambda_{\delta}$ are all observed.

\begin{lem}\label{lem:thm3}    Suppose Assumptions \ref{ass-a}, \ref{b-1}, \ref{ass:thm3} and Condition \ref{ass:thm3id} hold.     Then there exists  $\delta' \in (0,\delta)$ such
that $F(\cdot|x), x \in N^1(x_0,\delta)$ uniquely
determines the value of $\lambda$,  and moreover, 
$$(m_1(x) - m_1(x_0),m_2(x) - m_2(x_0)) \text{ if } \lambda \in (0,1)  $$ 
up to labeling and 
$$m_1(x) - m_1(x_0)\text{ if } \lambda = 1$$
for all $x$
in $N^1(x_0,\delta')$　as well.
\end{lem}

\begin{proof}  Similar to the proof of Lemma \ref{lem:thm1}.
	
	\end{proof}

Note that we once again needed the non-parallel regression function condition. Once the increments of the regression functions are identified, their levels  as well as the mixture weight $\lambda$ are obtained using the same procedure as in the first identification result.  Thus we have:

\begin{lem}\label{lem:rest3rdresult}   Suppose Assumptions \ref{ass-a}, \ref{b-1}, \ref{ass:thm3} and Condition \ref{ass:thm3id} hold.  Then there exists  $\delta' \in (0,\delta)$ such
	that $F(\cdot|x), x \in N^1(x_0,\delta)$ uniquely determines $(\lambda,m_1(\cdot),m_2(\cdot))$ in
  the set $(0,1] \times  \mathcal V( N^1(x_0,\delta') )^2$ up to labeling.
\end{lem}

To identify the distribution functions $(F_1(.),F_2(.))$, we now propose another method which will be used to construct our estimator and avoids Laplace inversion.  The main benefit of this is that it let us nonparametrically estimate the distribution functions without resorting to empirical MGF inversion, which is hard to handle in terms of obtaining polynomial rates of convergence.     
We will use previous identification of $\lambda$ and $m_1$ and $m_2$ evaluated at two points only, $x_1$ and $x_0$.

The idea is the following.
Equation (\ref{mixmodel2}) gives $F(z|x) = \lambda F_1(z - m_1(x)) + (1 -\lambda)F_2(z - m_2(x)), \forall (x,z) \in \R^{k+1},$  implying $ \forall (x,y) \in \R^{k+1},$
\begin{equation}\label{Ftrick}
F(m_1(x) +y |x) = \lambda F_1(y) + (1 -\lambda)F_2(m_1(x)-m_2(x) + y).
\end{equation}

Applying Equation (\ref{Ftrick}) to $(x_0,y)$ and $(x_1,y)$ and taking the difference, we obtain
\begin{align}\label{Fdiff}
F(m_1(x_1) +y |x_1) & - F(m_1(x_0) +y |x_0) = \nonumber \\
& = (1 -\lambda) \left( F_2(m_1(x_1)-m_2(x_1) + y) - F_2(m_1(x_0)-m_2(x_0) + y) \right),
\end{align}
which means that $\forall y \in \R,$ $ F_2(m_1(x_1)-m_2(x_1) + y) - F_2(m_1(x_0)-m_2(x_0) + y) $ is identified. Using recursively identification of this increment and the fact that the conditional cumulative distribution function $F_2$ converges to $1$ at infinity, we obtain identification of $F_2(z),$ $\forall z \in \R.$ Writing $g(x)=m_1(x)-m_2(x)$ and $\delta(x,x')=g(x)-g(x'),$ we assume that $\delta(x_1,x_0)>0.$ Note that $\delta(x_1,x_0) = \Delta - \nabla.$ Now, apply, for a given $z \in \R,$ Equation (\ref{Fdiff}) to $y=z - g(x_0)$ to obtain
$$F_2(z + \delta(x_1,x_0)) - F_2(z) = \frac{1}{1 -\lambda} (F(z + m_1(x_1) - g(x_0) |x_1) - F(z + m_2(x_0) |x_0)),$$
and, more generally, $\forall j \in \mathbb{N},$
\begin{align*}
F_2(z + (j+1) \delta(x_1,x_0)) - F_2(z+ j \delta(x_1,x_0)) =  \frac{1}{1 -\lambda}  \mathrm{\{} & F(z + j \delta(x_1,x_0) + m_1(x_1) - g(x_0) |x_1)\\
 & -  F(z + j \delta(x_1,x_0) + m_2(x_0) |x_0)) \mathrm{\}}.
\end{align*}
Using $\lim_{j \to \infty} F_2(z + (j+1) \delta(x_1,x_0)) = 1,$ the identifying equation for $F_2(.)$ is

\begin{align}\label{eq-idF}
F_2(z) =  1 \nonumber - \frac{1}{1-\lambda} \sum_{j=0}^{\infty} & F(z + j \delta(x_1,x_0) + m_1(x_1) - g(x_0) |x_1) \\
&   -  F(z + j \delta(x_1,x_0) + m_2(x_0) |x_0)),
\end{align}
where the infinite sum is a convergent series of positive terms.

Finally the equation $F(z|x)= \lambda F_1(z - m_1(x)) + (1 - \lambda) F_2(z - m_2(x))$ identifies $F_1(.)$ as
$$F_1(z) = \frac{1}{\lambda} \left[ F(z + m_1(x)) - (1 - \lambda) F_2(z + m_1(x) - m_2(x) \right]. $$ 

\begin{lem}\label{lem:F3rdresult}  Suppose Assumptions \ref{ass-a}, \ref{b-1}, \ref{ass:thm3} and Condition \ref{ass:thm3id} hold.  Then there exists  $\delta' \in (0,\delta)$ such
	that $F(\cdot|x), x \in N^1(x_0,\delta)$ uniquely determines  $(F_1(\cdot), F_2(\cdot))$ in
  the set $\bar{\mathcal F}(\R)^2$.
\end{lem}


\section{A model with ``fixed effects''}\label{sec:FE}  The model we have focused on so far assumes that heterogeneity is exogenously determined.  With $J=2$, a draw $(z,x)$ is generated from the first type of population or from the second with {\it fixed} probabilities $\lambda$ and  $1 - \lambda$.  This section relaxes this assumption.     We assume that the binary probability distribution over the two types/population can depend on $x$ in a completely unrestricted, nonparametric manner.  In terms of the switching regression formulation, this means:  
\begin{equation} \label{FEreg}
	z = 
	\begin{cases} 
		m_1(x) + \epsilon_1, \epsilon_1|x \sim F_1 \quad & \text{with
			probability } \lambda(x)
		\\
		m_2(x) + \epsilon_2, \epsilon_2|x \sim F_2  \quad & \text{with
			probability } 1 - \lambda(x).
	\end{cases}
\end{equation}
where $x$ and $\epsilon_1$ ($\epsilon_2$) are, as before, assumed to be independent.   Equivalently, we can write 
\begin{equation}\label{FEmodel}
	F(z|x) = \lambda(x) F_1(z - m_1(x)) + (1 -
	\lambda(x))F_2(z - m_2(x)).  
\end{equation}
The goals is now to identify the 5-tuple of functions $(\lambda(\cdot),m_1(\cdot),m_2(\cdot),F_1(\cdot),F_2(\cdot))$ from the joint distribution of $(z,x)$.  

This model is of a particular interest in terms its implications.  As in the rest of the paper, we often interpret the difference between $(m_1(\cdot),F_1(\cdot)  )$  and  $(m_1(\cdot),F_1(\cdot)  )$ as a representation of unobserved heterogeneity.   In a standard panel data regression model  often such heterogeneity is represented by a scalar, and when it is assumed to be independent of the regressor it would be representing random effects, whereas if it is allowed to be correlated with the regressor in an arbitrary manner it becomes a fixed effects model.  In certain applications fixed effects models are highly desirable.  Panel data often offers approaches to deal with fixed effects, a leading case being a linear model with additive scalar-valued fixed effects.     The model \eqref{FEreg} (or equivalently \eqref{FEmodel}) is in this sense analogous to these fixed effects models.  Unobserved heterogeneity in \eqref{FEreg} is function-valued (i.e. $m$ and $F$), as opposed to, say, an additive scalar.  Its distribution, represented by $\lambda(x)$, is dependent on $x$ in a fully unrestricted way, accommodating arbitrary correlation between the unobserved heterogeneity and the  regressor, so it resembles a panel data fixed effects model in this aspect.  In this section we show that  \eqref{FEreg} is nonparametrically identified, without requiring panel data, when the finite mixture modeling of unobserved heterogeneity is appropriate.  Moreover, unlike in the standard panel data fixed effects model, the distribution of unobserved heterogeneity conditional on $x$ is identified fully nonparametrically.  This means we identify the entire model, enabling the researcher to calculate desired counterfactuals.

We replace Assumption \ref{ass-a}  with
 \begin{ass}\label{ass:FE1}
	For some $\delta > 0$,  
	\begin{enumerate}[\textup{(}i\textup{)}]  
		
		\item\label{a-ind} $\epsilon_1|x \sim F_1$ and  $\epsilon_2|x \sim
		F_2$ at all $x \in N^1(x_0,\delta)$ where $F_1$ and $F_2$ do not depend on the value of $x$,
		
		\item\label{a-nonpara} 
		If $0 < \lambda(x_0) < 1$, $m_1(x_0) - m_1(x) \neq m_2(x_0) - m_2(x)$, for all $x \in
		N^1(x_0,\delta)$, $x \neq x_0$,

		\item\label{a-cont} $\lambda$, $m_1$ and $m_2$ are continuous in $x^1$ at $x_0$.

\end{enumerate}
\end{ass}
We maintain Assumption \ref{ass-b}, which, as noted before, is a weak regularity condition.
Define 
%
$$
K_{+\infty,t}(x) :=  R(t,x)\exp \left( - t\lim_{s \rightarrow +\infty} \frac 1 s \log(R(s,x))  \right),
$$ 
$$
K_{-\infty, t}(x) :=  R(t,x)\exp \left( - t\lim_{s \rightarrow -\infty} \frac 1 s \log(R(s,x))  \right),
$$ 
$$
K_{+\infty}(x) := \lim_{t \rightarrow +\infty} K_{+\infty,t}(x)
$$ 
and 
$$
K_{-\infty}(x) := \lim_{t \rightarrow -\infty} K_{ -\infty,t }(x).
$$ 
%
%

Note that the limits in these definitions are well-defined  over a neighborhood of $x_0$.

We replace Condition \ref{ass:thm1id} with:

\begin{condition}\label{cond:FE}  Either 
	\begin{enumerate}[\textup{(}i\textup{)}]

		\item\label{ass:thm4id1} $
		\lim_{t \rightarrow \infty}\frac 1 t \log  R(t,x)  \neq \lim_{t \rightarrow -\infty}\frac 1 t \log  R(t,x) \text{ for some } x \in N^1(x_0,\delta). 
		$
		
		or 
		
		\item\label{ass:thm4id2}  $K_{+\infty,t}(x) = 1$ for every $t \in \R$ and $x \in N^1(x_0,\delta)$

		holds for some $\delta > 0$.  	
		
	\end{enumerate} 	
\end{condition}

\begin{lem}\label{lem:thmFE}  Suppose Assumptions 
	\ref{ass-b},  \ref{ass:FE1} and Condition \ref{cond:FE} hold.  Then there exists  $\delta' \in (0,\delta)$ such
	that $F(\cdot|x), x \in N^1(x_0,\delta)$ uniquely
	determines  $\lambda(x)$,  and moreover, 
	$$(m_1(x) - m_1(x_0),m_2(x) - m_2(x_0)) \text{ if } \lambda(x)\lambda(x_0) \in (0,1)  $$ 
	up to labeling and 
	$$m_1(x) - m_1(x_0)\text{ if } \lambda(x)\lambda(x_0) = 1$$
	for all $x$
	in $N^1(x_0,\delta')$　as well.
\end{lem}

\begin{proof}  See appendix. 
	\end{proof}

The next result shows that the model that allows $\lambda$ to be arbitrarily dependent on $x$ is nonparametrically identified.  Note that the mixture can be degenerate (i.e. $\lambda(x) = 1$) for some values of $x$, and this can be also inferred from the observables.  As in the previous identification results presented in Lemmas  \ref{lem:id}, \ref{lem:altid} and \ref{lem:rest3rdresult}, its main sufficient condition (i.e. Condition   \ref{cond:FE}) is verifiable in terms of observables.    

\begin{lem}\label{lem:thmFEid}  Suppose Assumptions \ref{ass:FE1}, \ref{ass-b} and Condition \ref{cond:FE} hold.  Then $F(\cdot|x), x \in
	N^1(x_0,\delta)$ uniquely
	determines $(\lambda(x_0),F_1(\cdot),F_2(\cdot),m_1(x_0),m_2(x_0))$ in
	the set $(0,1] \times \bar{\mathcal F}(\R)^2
	\times \R^2$ up to labeling.
\end{lem}

\begin{proof}  Given Lemma \ref{lem:thmFE} the only remaining task is to identify the levels of $m_1$ and $m_2$ at $x_0$, $F_1$ and $F_2$.  
	Using the notation introduced in the proof of Lemma \ref{lem:thmFE}, with an additional definition  
	$$
	\dot \lambda(x) = \lambda(x) - \lambda(x_0),
	$$ 
	write
	$$
	\E [z|x] - \E[z|x_0] =  \lambda(x)[\dot m_1(x) - \dot m_2(x)] + \dot \lambda(x)[m_1(x_0) - m_2(x_0)].  
	$$
	If $\lambda(x) \neq \lambda(x_0)$ then we can proceed as in the proof of Lemma \ref{lem:id} to show that $m_1(x_0)$ and $m_2(x_0)$ are identified.  Accordingly, consider the case $\lambda(x) \neq \lambda(x_0)$.   Define 
	$$
	c(x) := \frac{	\E [z|x] - \E[z|x_0] -  \lambda(x)[\dot m_1(x) - \dot m_2(x)] }{ \dot \lambda(x) },
	$$
	which is observable by Lemma  \ref{lem:thmFE}, then $c(x) = m_1(x_0) - m_2(x_0)$, and we obtain 
	$$
	\begin{bmatrix}
	c(x) \\
	\E [z|x]
	\end{bmatrix}
	 = 
	 \begin{pmatrix}
	 1 & 1 \\
	 \lambda(x) & 1 - \lambda(x)
	 \end{pmatrix}
	 \begin{bmatrix}
	 m_1(x_0) \\
	 m_2(x_0)
	 \end{bmatrix}.
	 $$
	 Since the determinant of the matrix on the right hand side is unity, once again $m_1(x_0)$ and $m_2(x_0)$ are identified.  Finally, we proceed as in  as in the proof of Lemma \ref{lem:id} to identify $F_1$ and $F_2$, though here the 2-by-2 matrix in the following display  does not factorize:
	 \begin{equation}\label{eq:mgfid}
	 \begin{bmatrix}
	 M(t|x) \\
	 M(t|x_0)
	 \end{bmatrix}
	 = 
	 \begin{pmatrix}
	 \lambda(x)e^{tm_1(x)} & (1 - \lambda(x) )  e^{tm_1(x)} \\
	 \lambda(x_0) e^{tm_1(x)}  & ( 1 - \lambda(x_0) ) e^{tm_2(x_0)}
	 \end{pmatrix}
	 \begin{bmatrix}
	 m_1(x_0) \\
	 m_2(x_0)
	 \end{bmatrix}, \quad \text{ \rm for every } t \in \R.
	 \end{equation}
 Nevertheless, its determinant is, if $\lambda(x_0) \neq 1$ 
\begin{eqnarray*}
\lambda(x)(1 - \lambda(x_0) ) e^{t[m_1(x) + m_2(x_0)]} &-& \lambda(x_0)(1 - \lambda(x) ) e^{t[m_1(x_0) + m_2(x)]}  
\\
=&& \lambda(x_0)(1 - \lambda(x_0) )   e^{t[m_1(x_0) + m_2(x)]}     \left\{ \frac{\lambda(x)}{\lambda(x_0)}  e^{t[\dot m_1(x) - \dot m_2(x)  ]}  - \frac{1 - \lambda(x)}{1 -  \lambda(x_0)}   \right \}
\end{eqnarray*}
which is non-zero for almost all $t$ under the non-parallel condition.
Therefore \eqref{eq:mgfid} uniquely  determines $M_1$ and $M_2$, hence $F_1$ and $F_2$.  The treatment of the case with $\lambda(x_0) = 1$ is straightforward.  
\end{proof}

\section{Instrumental Variables}\label{sec:iv}
The identification results developed in the preceding sections can be used to identify nonparametric finite mixture regression with endogenous regressors.   Suppose we observe a triple of random variables $(y,w,x)$ taking its value in $\mathcal Y \times \mathcal W \times \mathcal X$ where    $\mathcal Y \subset \R$, $\mathcal W \in \R^p$ and $\mathcal X \in  \R^k$.   Also let 
$$
z := \binom y  w.   
$$  
In a manner similar to Section \ref{switch}, consider a switching regression model:
\begin{equation} \label{switchIV}
y = 
\begin{cases} 
g_1(w) + \eta_1, \quad (y,w,x,\eta_1) \sim F_1 \quad & \text{with
	probability } \lambda
\\
g_2(w) + \eta_2, \quad (y,w,x,\eta_2) \sim F_2 \quad & \text{with
	probability } 1 - \lambda.
\end{cases}
\end{equation}
Unlike in the previous sections, however,  we no longer assume that $\eta$'s and $w$ are uncorrelated or independent.    
Instead, we assume 
\begin{equation}\label{eq:npiv}
\int \eta d F_1(\eta|x)   = \int \eta d F_2(\eta|x) = 0, 
\end{equation}
that is
$$
\E [\eta_1|x] = \E [\eta_2|x] = 0.
$$
Here and thereafter the notation $F_i(\star_1,\star_2,...)$ and  $F_i(\star_1,\star_2,...|\star)$ denote the joint distribution of $\star_1,\star_2,...$ and the conditional distribution of $\star_1,\star_2,...$ given $\star$ when the joint distribution is given by $F_i, i = 1,2$.  Consider linear operators 
\begin{equation}\label{eq:operator}
T_1[f](x) = \int f(w) dF_1(w|x), \quad T_2[f](x) = \int f(w) dF_2(w|x)
\end{equation}
and assume that these operators are invertible.  

The main goal is to identify $g_1$ and $g_2$.  Here $x$
plays the role of instrumental variables.       
As before, define $m_1(x) = \int z dF_1(z|x)$ and $m_2(x) = \int z
dF_2(z|x)$.  Note that $m_1: \R^k \rightarrow \R^{p+1}$ and  $m_2: \R^k
\rightarrow \R^{p+1}$.  For $j = 1,...,p+1$, let $m_{j,1}(\cdot)$ and $m_{j,2}(\cdot)$ denote the $j-$th
elements of $m_1(\cdot)$ and $m_2(\cdot)$, respectively.  Define the
$p+1$-dimensional vectors of random variables 
$\epsilon_j  = z - m_i(x), (z,x) \sim F_i(z,x), j = 1,2$.  Consistent with the previous notation let
$F_i(\epsilon_i|x), i = 1,2,$ denote the conditional distribution of
$\epsilon_1$ and $\epsilon_2$ under $F_1$ and $F_2$.  

By
construction, 
$$
\int \epsilon dF_i(\epsilon|x) = 0, i = 1,2.
$$        
If we further assume that $F_i(\epsilon|x), j = 1,2$ do not depend on
$x$, an appropriate extension of the theory developed in Section \ref{sec:switch} can be used to
identify $m_{p+1,1}(x)$, $m_{p+1,2}(x)$,  $F_1(z|x)$ and $F_2(z|x)$, which in turn, also identify the operators $T_1$ and $T_2$.  By \eqref{switchIV}, \eqref{eq:npiv} and \eqref{eq:operator} we have
$$
m_{p+1,1}(x) = T_1[g_1](x), \quad m_{p+1,2}(x) = T_2[g_2](x).
$$
 Then by their invertibility $g_1$ and $g_2$ are identified as .
To formalize this idea, consider the following assumptions:       
\begin{ass}\label{ass-e}
For some $\delta > 0$,  
\begin{enumerate}[\textup{(}i\textup{)}]  

\item\label{e-ind} $\epsilon_1|x \sim F_1^{\epsilon}$ and  $\epsilon_2|x \sim
  F_2^\epsilon$ at all $\mathcal X$ where $ F_1^{\epsilon}$ and $ F_1^{\epsilon}$ do not depend on the value of $x$; 

\item\label{e-nonpara} 
$m_{j1}(x_0) - m_{j1}(x) \neq m_{j2}(x_0) - m_{j2}(x)$, for all $x \in
N^1(x_0,\delta)$, $x \neq x_0$ and for all $j$, $j = 1,...,p+1$;

\item\label{e-cont} $m_1$ and $m_2$ are continuous at $x_0$.

\end{enumerate}
\end{ass}
To state a multivariate extension of Assumption \ref{ass-b}, define
the multivariate moment generating function  
$$
M_i({\bf t}) = \int e^{{\bf t}^\top\eta} dF_i(\eta), \quad i =
1,2, \quad {\bf t} \in \R^{p+1}.
$$
Let ${\bf e}_j$ denote the unit vector whose $j-$th element is 1.  Accommodating the identification strategy in Section \ref{sec:1stid} require some modification as follows.  Define 
$
{D(x)} := m_2(x) - m_1(x) 
$ as before, though now $D: \R^k \rightarrow \R^p $ is vector-valued.  Also let   
$$
h_j(c,t) : = e^{ {t {\bf e}_j}'  D(x_0)(1 - c)}  \frac {M_2(t{\bf e}_j)} {M_1(t{\bf e}_j)}, \quad c \in \R_{++  }, t \in \R
$$
and 
$$
R({\bf t},x)  := \frac{M({\bf t} | x)}{M({\bf t} | x_0)}, \qquad {\bf t} \in \R^{p+1}.
$$
  
\begin{ass}\label{ass-f}
\begin{enumerate}[\textup{(}i\textup{)}]

\item\label{f-0.5} The domains of $M_1({\bf t})$ and $M_2({\bf t})$ are
  $(-\infty,\infty)^{p+1}$;
  
\item\label{f-1} For some $\varepsilon >0$ either $h_j(\pm\varepsilon,t) = O(1)$ or $1/ h_j(\pm\varepsilon,t) = O(1)$, or both hold as $t \rightarrow  +\infty$  for each $j \in \{1,...,p+1\}$.   Moreover, the same holds as  $t \rightarrow - \infty$



\end{enumerate}
\end{ass}

\begin{condition}\label{cond:iv}  Either 
	\begin{enumerate}[\textup{(}i\textup{)}]

		\item\label{cond:iv1} $
		\lim_{t \rightarrow \infty}\frac 1 t \log  R(t{\bf e}_j,x)  \neq \lim_{t \rightarrow -\infty}\frac 1 t \log  R(t{\bf e}_j,x) \text{ for each } j \in \{1,...,p\} \text{ and  for some } x \in N^1(x_0,\delta) 
		$
		
		or 
		
		\item\label{cond:iv2}  $\lim_{c \downarrow 0} \lambda_{c} = 1$
		
		holds.  	
		
	\end{enumerate} 	
\end{condition}

By modifying the proofs of Lemmas \ref{lem:slope} and \ref{lem:id}
appropriately to deal with $\R^p$-valued random variables, we
can show that $(\lambda,F_1^{\epsilon},F_2^\epsilon,m_1(x_0),m_2(x_0))$ is identified
under Assumptions \ref{ass-e} and
\ref{ass-f}.  Then,
as noted in Remark \ref{rem:idoverarea}, $F(z|x), x \in 
{\mathcal X}, z \in \R^p$ uniquely determines
$(\lambda,F_1^\epsilon,F_2^\epsilon,m_1(x),m_2(x))$ in $\R \times {\bar {\mathcal F}(\R^p)^2
  \times \R^{2p}}$ for all $x \in \mathcal X$ up to labeling.
Therefore each component distribution of $z$ is obtained by
$$
F_1(z|x) = F_1^\epsilon(z - m_1(x)), \quad F_2(z|x) = F_2^\epsilon(z - m_2(x)). 
$$ 
We now have:
\begin{thm}
Suppose Assumptions \ref{ass-e}, \ref{ass-f} and Condition  \ref{cond:iv} hold.
Then $g_1(\cdot)$ and $g_2(\cdot)$ are identified.  
\end{thm}
\begin{rem}
It is possible to further introduce flexibility into the model \eqref{switchIV} by allowing unrestricted dependence between unobserved heterogeneity and the instrument $x$.  This can be achieved by making $\lambda$ in \eqref{switchIV} an arbitrary function of $x$.  Applying the results in Section \ref{sec:FE} to identify  $m_{p+1,1}(x)$, $m_{p+1,2}(x)$,  $F_1(z|x)$ and $F_2(z|x)$ and proceeding as above, we recover $g_1$ and $g_2$ nonparametrically.    
 	\end{rem}

\section{Mixtures with arbitrary  $J$}\label{sec:J3}

Previous sections studied the identifiability for mixtures with $J =
2$.  It is desirable, however, to be able to deal with mixtures with many components in some
applications, especially when mixtures are used to represent
unobserved heterogeneity.  This section shows that nonparametric identification can be established for 
general $J$, possibly greater than 2, and moreover we show that the number  $J$ itself is also identifiable.

The basic setup in this section is analogous to the one considered in
Section \ref{sec:switch}, though the conditional distribution of $z
\in \R$
given $x \in \R^n$ consists of $J$ components, $J \in \mathbb{N}$, as in (\ref{Jmodel}).               
As before, define 
$$
m_j(x) = \int_\R z dF_j(z|x), j = 1,2,...,J.
$$
Define also 
\[
\epsilon_j = z_j - m_j(x), \quad j = 1,2,...,J.
\]
Later we impose independence between $\epsilon_j, j =
1,...,J$ and $x$, which enables us to write $F(z|x)$ as
\begin{equation}\label{Jcond}
F(z|x) = \sum_{j=1}^J \lambda_jF_j(z - m_j(x)).  
\end{equation}
For later use, define
$M_j(t) = \int e^{t\epsilon}F_j(d\epsilon), j = 1,...,J$. 
This section shows that the parameter 
$(\{\lambda_j\}_{j=1}^J,\{F_j(\cdot)\}_{j=1}^j,\{m_j(\cdot)\}_{j=1}^J)$
is identifiable under suitable conditions. 

At an intuitive level, the argument developed in Section \ref{sec:switch} still offers
a valid picture behind the identifiability result here.  The independence of $\epsilon$ from
$x$ leads to a shift restriction: the shapes of the
distributions of $\{\epsilon_j\}_{j=1}^J$ have to remain invariant along the
$J$ regression functions.  This restriction, with other conditions, nails down the
true parameters uniquely. 
Moving from $J = 2$ to $J \geq 3$, however,
involves rather different theoretical arguments as developed subsequently. 
Recall that Section \ref{sec:switch} presented alternative
conditions that guarantee the identifiability of two-component mixture
models, as summarized by Lemma \ref{lem:id}, Lemma \ref{lem:altid} and Lemma \ref{lem:thm3}.  
%
This section  proves the
nonparametric identifiability of (\ref{Jcond}) under conditions that
are similar to the ones used in Lemma \ref{lem:id}, which seems least prohibiting of the three to generalize.  Even so, this generalization calls for multistep identification argument with recursive procedures, as will be seen shortly.

To see how the treatment of general mixtures differs from the $J=2$
case, consider the case $J=3$.  Instead of 
Equation (\ref{mixedmgf}), we now have    
\begin{equation}\label{j_is_3}
M(t|x) = \lambda_1 e^{tm_1(x)}M_1(t) + \lambda_2 e^{tm_2(x)}M_2(t) +
\lambda_3 e^{tm_3(x)}M_3(t), \quad \lambda_1 + \lambda_2 + \lambda_3 = 1.
\end{equation}
Wlog, suppose $m_1(x_0) > m_2(x_0) > m_3(x_0)$ at a point $x_0$ in
$\R^k$.  
Take a point $x'$ in the neighborhood of $x_0$ and consider the case $m_1(x') - m_1(x_0) \geq 0$ (if this term is
negative, the roles of $m_1$ and $m_3$ get interchanged).  The
method used in the proof of Lemma \ref{lem:slope} to identify the $J=2$ model still
works for the slopes
of $m_1$ and $m_3$.  Following the
proof, take the ratio of the conditional moment generating functions
at $x_0$ and a point in its neighborhood, $x'$, say, then take its
logarithm followed by a normalization by $t$:  
\begin{align*}
\frac 1 t \log \left ( \frac {M(t|x')} {M(t|x_0)} \right ) 
&= \frac 1 t \log \left ( \frac 
  { \lambda_1 e^{tm_1(x')}M_1(t) + \lambda_2
    e^{tm_2(x')}M_2(t) +  \lambda_3 e^{tm_3(x')}M_3(t)}
  { \lambda_1 e^{tm_1(x_0))}M_1(t) + \lambda_2
    e^{tm_2(x_0)}M_2(t) +  \lambda_3 e^{tm_3(x_0))}M_3(t)} \right)
\\
&= \frac 1 t \log \left ( \frac 
  { e^{t[m_1(x') - m_1(x_0)]} + \frac {\lambda_2}{\lambda_1} 
    e^{t[m_2(x') - m_1(x_0)]} \frac {M_2(t)}{M_1(t)}   
    +  \frac {\lambda_3}{\lambda_1} 
    e^{t[m_3(x') - m_1(x_0)]} \frac {M_2(t)}{M_1(t)}       } 
  { 1 + \frac{\lambda_2}{\lambda_1}
    e^{t[m_2(x_0) - m_1(x_0)]}\frac {M_2(t)}{M_1(t)}
    +  \frac {\lambda_3}{\lambda_1} 
    e^{t[m_3(x_0) - m_1(x_0)]} \frac {M_2(t)}{M_1(t)}        }
   \right ).  
\end{align*}
Suppose the ratios of  $M_1(t)$, $M_2(t)$ and $M_3(t)$ do not
explode exponentially, and $m_1$, $m_2$ and $m_3$ are continuous so
that $m_2(x') - m_1(x_0)$ and $m_3(x') - m_1(x_0)$ are negative.  Then
as $t$ approaches to infinity, the above expression approaches to the
slope $m_1(x') - m_1(x_0)$ if it is
non-negative (though it yields the identical result if the slope is
negative as well,
as seen in the proof of Lemma \ref{lem:slope}).  Similarly, by
taking the limit $t \rightarrow - \infty$, the slope of $m_3$ is identified.    
This argument, however, leaves the slope of the middle term $m_2$ undetermined. And in
the general case of $J \geq 3$, $J-2$ slopes remain to be
determined.  The approach in Lemma \ref{lem:slope} does fall short of
achieving its goal when applied to models with $J \geq 3$.

It is, however, possible to identify the slope of $m_2$ by proceeding
as follows.  Suppose, evaluated at
$x$, the regression functions satisfy the inequality $m_1(x) >
m_2(x) > m_3(x)$.  Pick a point $y$ in a neighborhood of $x$.  Multiply
(\ref{j_is_3}) by $e^{-t{[m_1(x) - m_1(y)]}}$ to obtain: 
\begin{equation}\label{purgex}
e^{-t{[m_1(x) - m_1(y)]}}M(t|x) = \lambda_1 e^{tm_1(y)}M_1(t) +
\lambda_2 e^{t\{m_2(x) - [m_1(x) - m_1(y)]\}}M_2(t) +
\lambda_3 e^{t\{m_3(x) - [m_1(x) - m_1(y)]\}}M_3(t).
\end{equation}
This purges $x$ out of the first term on the right hand side.  Note
$[m_1(x) - m_1(y)]$ can be identified by applying the argument in Lemma \ref{lem:slope}
to the $J = 3$ model (\ref{j_is_3}), as demonstrated above.  Therefore
the left hand side of the above equation is known.
  
The above step enables us to eliminate all unknown parameters
associated with the first mixture component.  To see this, suppose
$m_j, j = 1,2,3$ are differentiable in at least one of the 
$k$ elements of $x = (x^1,x^2,...,x^k)$.  In what follows we assume
that it is differentiable in the first element $x^1$ without loss of
generality.  As before, we assume that this is a prior knowledge.    Let $D_x$
denote the partial differentiation operator with respect to the first component of x, 
i.e. $D_xf(x) =  \frac \partial {\partial x^1}f(x)$.  
Differentiating both sides of the above equation by $x^1$ and rearranging, 
\begin{align}\label{stage1}
D_x\left [e^{-t{[m_1(x) -
      m_1(y)]}}M(t|x)\right ]&  =
t\lambda_2 [D_xm_2(x) - D_xm_1(x)] e^{t\{m_2(x) - [m_1(x) - m_1(y)]\}}
      M_2(t)
\\
\nonumber
 &+
t\lambda_3 [D_xm_3(x) - D_xm_1(x)] e^{t\{m_3(x) - [m_1(x) - m_1(y)]\}}M_3(t). 
\end{align} 
Note that operating $D_x$ eliminates the unknown function $M_1(t)$ out of the
right hand side of (\ref{stage1}).  We now have
$$
\frac {\partial}{\partial t} \log\left(D_x\left [e^{-t{[m_1(x) -
      m_1(y)]}}M(t|x)\right ]\right) = \frac {A_1}{A_2}, 
$$
say, where
\begin{align*}
A_1 &= \frac 1 t + \{m_2(x) - [m_1(x) - m_1(y)]\} 
+  \frac {\frac {\partial}{\partial t} M_2(t)}{M_2(t)} 
\\
&+ \frac {\lambda_3}
{\lambda_2} \frac {[D_xm_3(x) - D_xm_1(x)]}{ [D_xm_2(x) - D_xm_1(x)]} 
\left (\frac 1 t + \{m_3(x) - [m_1(x) - m_1(y)]\}\right ) e^{t[m_3(x)
  - m_2(x)]}\frac {M_3(t)}{M_2(t)} 
\\
&+ \frac {\lambda_3}
{\lambda_2} \frac {[D_xm_3(x) - D_xm_1(x)]}{ [D_xm_2(x) - D_xm_1(x)]} 
e^{t[m_3(x) - m_2(x)]}\frac {\frac {\partial}{\partial t}M_3(t)}{M_2(t)}
\end{align*}
and 
$$
A_2 = 1 +  \frac {\lambda_3}
{\lambda_2} \frac {[D_xm_3(x) - D_xm_1(x)]}{ [D_xm_2(x) - D_xm_1(x)]} 
e^{t[m_3(x) - m_2(x)]}\frac {M_3(t)}{M_2(t)}.
$$
Note that the factor $D_xm_2(x) - D_xm_1(x)$ is non-zero if the two regression
functions are not parallel at $x$, which makes the division by the
factor valid.   As far as $\frac
{M_3}{M_2}$ and  $\frac {D_xM_3}{M_2}$ do not 
explode exponentially, all the terms above except for the second and third terms
of $A_1$ and the first term of $A_2$ converge to zero as $t
\rightarrow \infty$.  It follows that 
\begin{equation}\label{difft}
\lim_{t \rightarrow \infty} \left \{\frac {\partial}{\partial t} \log\left(D_x\left [e^{-t{[m_1(x) -
      m_1(y)]}}M(t|x)\right ]\right) \right \} =  \{m_2(x) - [m_1(x) - m_1(y)]\} 
+  \frac {\frac {\partial}{\partial t} M_2(t)}{M_2(t)}. 
\end{equation}
The only unknown component in the above equation is $\frac {\frac {\partial}{\partial t}
  M_2(t)}{M_2(t)}$, but this term depends only on $t$, so
it can be differenced out: repeat the above argument with replacing
$x \in \R^k$ with a point $z \in \R^k$ so close to $x$ that  $m_1(z) >
m_2(z) > m_3(z)$.  This yields    
$$
\lim_{t \rightarrow \infty} \left \{\frac {\partial}{\partial t} \log\left(D_z\left [e^{-t{[m_1(z) -
      m_1(y)]}}M(t|z)\right ]\right) \right \} =  \{m_2(z) - [m_1(z) - m_1(y)]\} 
+  \frac {\frac {\partial}{\partial t} M_2(t)}{M_2(t)}. 
$$  
The slope of $m_2$ is 
$$
m_2(x) - m_2(z) = \lim_{t \rightarrow \infty} \frac {\partial}{\partial t} \log \left(\frac {D_x\left [e^{-t{[m_1(x) -
      m_1(y)]}}M(t|x)\right]}{D_z\left [e^{-t{[m_1(z) -
      m_1(y)]}}M(t|z)\right ]}\right) + (m_1(x) -
m_1(z)).
$$
The terms such as $m_1(x) - m_1(z)$ on the right hand side are
identified by the method developed in Lemma \ref{lem:slope}, as noted earlier.  The
equation above shows the identifiability of the slope of $m_2$.      

We have already noted that the identifiability of the slope of $m_3$
basically follows from Lemma \ref{lem:slope}.  It is nevertheless
instructive to present an alternative way to identify it by carrying on the
foregoing analysis one step further.
This will illustrate the basic idea behind our
general identification theory for $J \in \mathbb{N}$.     


Let us return to Equation (\ref{stage1}), changing the notation and writing $x_a$ for $x$, $x_b$ for $y$.  As before, $\Delta_{ab}f$ stands for $f(x_a) - f(x_b)$.  The first step is to purge
$x_a$ from the first term on the right hand side, as we did in Equation
(\ref{purgex}), as follows:
\begin{align*}
\frac {e^{-t[\Delta_{ab}m_2 - \Delta_{ab}m_1]}}{t[D_{x_a}m_2(x_a) - D_{x_a}m_1(x_a)]}
D_{x_a}\left [e^{-t{\Delta_{ab}m_1}}M(t|x_a)\right ]&  =
\lambda_2 e^{t\{m_2(x_b)\}}M_2(t)
\\
\nonumber
 &+
\lambda_3 \frac{D_{x_a}m_3(x_a) - D_{x_a}m_1(x_a)}{D_{x_a}m_2(x_a) -
  D_{x_a}m_1(x_a)} e^{t\{m_3(x_a) - \Delta_{ab}m_2\}}M_3(t),   
\end{align*}
which yields
\begin{align}\label{stage1for3}
D_{x_a}\left[\frac {e^{-t[\Delta_{ab}m_2 - \Delta_{ab}m_1]}}{t[D_{x_a}m_2(x_a) - D_{x_a}m_1(x_a)]}
D_{x_a}\left [e^{-t{\Delta_{ab}m_1}}M(t|x_a)\right ] \right]&  
\\
\nonumber
=
\lambda_3 \left \{\left[D_{x_a} + t(D_{x_a}m_3(x_a)
    - D_{x_a}m_2(x_a))\right] \frac{D_{x_a}m_3(x_a) - 
    D_{x_a}m_1(x_a)}{D_{x_a}m_2(x_a) - 
  D_{x_a}m_1(x_a)} \right \}
&
e^{t\{m_3(x_a) - \Delta_{ab}m_2\}}M_3(t).   
\end{align}
Notice that again this eliminates an unknown moment generating
function, this time  $M_2(t)$.  Differentiating the
above expression with respect to $t$ and following the line of argument presented
above, the slope of $m_3$ is given by 
$$
\Delta_{ac}m_3 = \lim_{t \rightarrow \infty} \frac {\partial}{\partial t} \log \left[
\left(
\frac
{
\frac {e^{-t[\Delta_{ab}m_2 - \Delta_{ab}m_1]}}{t[D_{x_a}m_2(x_a) - D_{x_a}m_1(x_a)]}
D_{x_a}\left [e^{-t{\Delta_{ab}m_1}}M(t|x_a)\right ]
}
{
\frac {e^{-t[\Delta_{cb}m_2 - \Delta_{cb}m_1]}}{t[D_{x_c}m_2(x_c) - D_{x_c}m_1(x_c)]}
D_{x_c}\left [e^{-t{\Delta_{cb}m_1}}M(t|x_c)\right ]
}
\right) \right]
+ 
\Delta_{ac}m_2.
$$     

Let us now turn to the identifiability of the general model
(\ref{Jcond}) for a generic $J,$ at a point $x_a \in \R^k$. The general setting is the same as in Section \ref{sec:switch}: the first $k*$ elements $x^1,...,x^{k*}$ of the vector of covariates $x$ are continuous covariates, and we will again use local variations in $x^1$.
\begin{ass}\label{ass:jmodel}
For some $\delta > 0$,  
\begin{enumerate}[\textup{(}i\textup{)}]  

\item\label{jmodel-ind} $\epsilon_j|x \sim F_j, j = 1,...,J$ at all $x
  \in N^1(x_a,\delta)$ where $F_j, j = 1,...,J$ do not depend on the value of
  $x$;  

\item\label{jmodel-cont} $m_j, j = 1,...,J$ are continuous in $x^1$at $x_a$;

\item\label{jmodel-dble} $m_j, j = 1,...,J$ are J times differentiable on $B(x_a,\delta)$ 
at least in one of the $k^*$ continuous covariates of $x$;

\end{enumerate}
\end{ass}
Though Condition \eqref{jmodel-dble} imposes $J$-th order differentiability in one argument for simplicity of presentation, this is not essential: it is sufficient to assume that there exists at least one multi-index $\alpha := (\alpha_1,...,\alpha_k) \in \mathbb{Z}^k, \alpha_1 + \cdots \alpha_k = J$ such that the derivative $D^\alpha m(x) = \left(\frac{\partial}{\partial x^1}\right)^{\alpha_1} \cdots \left(\frac{\partial}{\partial x^k}\right)^{\alpha_k} m(x) $ is well-defined for every $x$ in  $B(x_a,\delta)$.  See Remark \ref{rem:dble} for further discussions.   


The independence assumption (\ref{jmodel-ind}) enables us to write the observable conditional
distribution in the form (\ref{Jcond}).  The continuity assumption
(\ref{jmodel-cont}) was also assumed in Lemma \ref{lem:slope}.  The
differentiability condition (\ref{jmodel-dble}) may not be essential
for the proof of the Lemma,
though replacing derivatives in the proof with differences leads to
extremely complex case-by-case analysis. Note that differentiability
in only one element of $x$ suffices.  Without loss of generality in what follows we
assume that the $m_j, j=1,...,J$ are differentiable in the first element
$x^1$.  Recall that $D_1$ is the differentiation operator with respect
to $x^1$. From now on, we will use the notation $$m_{k,j}(x)=m_k(x)-m_j(x).$$ 

\begin{ass}\label{ass:jmodel-diff}

\begin{enumerate}[\textup{(}i\textup{)}]  

\item\label{jmodel-apart} $\displaystyle\min_{k \neq j} |m_{k,j}(x_a)| > \Delta $, $\Delta>0$;

\item\label{jmodel-diff}  $D_1 m_{j}(x_a), j = 1,...,J$ takes $J$
  distinct values in $\R$; 
  
\item\label{jmodel-domain} The domains of $M_1(t)$ and $M_2(t)$ are $(-\infty,\infty)$;

\item\label{jmodel-lim} For some $\epsilon > 0$ , $ \displaystyle\lim_{t \rightarrow \infty}  e^{t(\epsilon - \Delta) }
\frac{M_j(t)}{M_k(t)} = 0$ and $ \displaystyle\lim_{t \rightarrow \infty}  e^{t(\epsilon - \Delta) }
\frac{\frac {\partial}{\partial t} M_j(t)}{M_k(t)} = 0$ for all $k,j=1,...,J$.

\end{enumerate}
\end{ass}

Part (\ref{jmodel-apart}) of the assumption is not restrictive.  
As before, our goal is to establish identification up to labeling, so
we can assume that 
\begin{equation}\label{ordering}
m_1(x_a) > m_2(x_a) > ... > m_{J}(x_a)     
\end{equation}
without loss of generality: this does not impact the validity of Assumption \ref{ass:jmodel-diff}. Part (\ref{jmodel-diff}) is an
infinitesimal version of the non-parallel regression function
conditions used in the previous sections. 

Under these assumptions, we first prove identifiability of the slope $\Delta_{ab}m_1$, using the method developed in Section \ref{switch} , for all $x_b$ in a chosen neighborhood of $x_a$. Note that we know $\lambda_1 \neq 0$. 
By the continuity and differentiability assumptions (Assumption \ref{ass:jmodel} (\ref{jmodel-cont}) and (\ref{jmodel-dble})), there exists $\delta' > 0$, $\delta' < \delta$, such that for all $x_b \in  N^1(x_a, \delta')$ and for all $j=1,...,J$, $|m_j(x_b)-m_j(x_a)|<\frac{\epsilon}{2} $, and $D_1 m_j(x_b), j=1,...,J$ take J distinct values. Here we use the fact that twice differentiability of the regression functions implies that they are $\mathcal{C}^1$.
Then, as in the proof of Lemma \ref{lem:slope}, in the case $ m_1(x_b)-m_1(x_a)>0$, we write
$$
\frac 1 t \log \left ( \frac {M(t|x_b)} {M(t|x_a)} \right ) 
= \frac 1 t \log \left ( \frac 
  { e^{t[m_1(x_b) - m_1(x_a)]} + \sum_{j=2}^{J} \frac{\lambda_j}{\lambda_1}\frac{M_j(t)}{M_1(t)}e^{t[m_j(x_b) - m_1(x_a)]}}
  {1 + \sum_{j=2}^{J} \frac{\lambda_j}{\lambda_1}\frac{M_j(t)}{M_1(t)}e^{t[m_j(x_a) - m_1(x_a)]}} \right),
$$
and in the case $ m_1(x_b)-m_1(x_a)<0$ , we write 
$$
\frac 1 t \log \left ( \frac {M(t|x_b)} {M(t|x_a)} \right ) 
= \frac 1 t \log \left ( \frac 
  { 1 + \sum_{j=2}^{J} \frac{\lambda_j}{\lambda_1}\frac{M_j(t)}{M_1(t)}e^{t[m_j(x_b) - m_1(x_b)]}}
  {e^{t[m_1(x_a) - m_1(x_b)]} + \sum_{j=2}^{J} \frac{\lambda_j}{\lambda_1}\frac{M_j(t)}{M_1(t)}e^{t[m_j(x_a) - m_1(x_b)]}} \right).
$$
Similarly, since $m_j(x_b) - m_1(x_a)$, $m_j(x_a) - m_1(x_a)$, $m_j(x_b) - m_1(x_b)$, and $m_j(x_a) - m_1(x_b)$ are less than $\epsilon - \Delta$, this gives in both cases,
$$  \forall x_b \in U, \lim_{t \rightarrow \infty} \frac 1 t \log \left ( \frac {M(t|x_b)} {M(t|x_a)} \right )  = \Delta_{ba}m_1. $$ Hence the slope $\Delta_{ab}m_1$ is identifiable for all $x_b \in N^1(x_a, \delta')$.
\\

Now we focus on the identifiability of the slopes $\Delta_{ab}m_j$ for all $j=2,...,J$ and $x_b$ in an appropriate neighborhood of $x_a$.

Pick a point $x_b \neq x_a$ in $\R^k$. 
For notational convenience, define the operator $A(x_a,x_b,t,k)$  
\begin{equation}\label{operator}
A(x_a,x_b,t,k) (f)(x_a) = \frac{\partial}{\partial
  x_a^1}\left[\frac{e^{-t[\Delta_{ab}m_k -
      \Delta_{ab}m_{k-1}]}}{R_{k}(t,x_a)} f(x_a)\right], k = 2,3,...,J.   
\end{equation}
where $f: \R^k \rightarrow \R$ is a function that is differentiable in
its first argument, and $R_k(t,x)$ is a (rational) function in $t$. Its precise 
definition will be given shortly. The operator $A(x_a,x_b,t,k)$
generalizes the procedure performed on $D_{x_a}\left
  [e^{-t{\Delta_{ab}m_1}}M(t|x_a)\right ]$ in Equation
(\ref{stage1for3}) to eliminate unknown parameters in (\ref{stage1}). 
Operate $A(x_a,x_b,t,k)$, $k = 2,3,...$ sequentially on $D_{x_a}\left
  [e^{-t{\Delta_{ab}m_1}}M(t|x_a)\right ]$ to define the expressions      
\begin{equation}\label{Qdefinition}
Q_k(x_a,t) = A(x_a,x_b,t,k-1) A(x_a,x_b,t,k-2) \cdots A(x_a,x_b,t,2)\frac{\partial}{\partial
  x_a^1} [e^{-t\Delta_{ab}m_1}M(t|x_a)], \quad k = 2,3,...,J.
\end{equation}
By construction $Q_k(x_a,t)$ satisfies the following recursive formula:
\begin{equation}\label{eq:recur}
Q_{k+1}(x_a,t) =  A(x_a,x_b,t,k)Q_{k}(x_a,t), \quad Q_2(x_a,t) = \frac{\partial}{\partial
	x_a^1} [e^{-t\Delta_{ab}m_1}M(t|x_a)]. 
\end{equation}
The definition of the operator $A(x_a,x_b,t,k)$, as explained further later,
is motivated by two facts: (i) the
factor $e^{-t[\Delta_{ab}m_k - \Delta_{ab}m_{k-1}]}$ purges $x_a$ out
of the exponent in the leading term of $Q_{k}(x_a,t)$ and (ii)
division by the polynomial $R_k(t,x_a)$ then makes the leading term
$\lambda_ke^{-tm_k(x_b)}M_k(t)$, which is completely free from $x_a$ and
therefore eliminated by $D_{x_a}$.  Once this is
done, taking the log-derivative
with respect to $t$ as in (\ref{difft}) terms and taking the limit $t
\rightarrow \infty$ yields
$\Delta_{ab}m_{k}$ up to an unknown additive factor
$\frac{\frac{\partial}{\partial t}M_k(t)}{M_k(t)}$, which can be
differenced out.  

Subsequent arguments establish the identifiability of  $\Delta_{ab}m_k, k =
  2,...,J$ for all $x_b$ in a
  neighborhood of $x_a$.  We proceed in two steps.  Step 1 shows that, with an
  appropriate choice of $R_k(t,x_a)$ in (\ref{operator}),   
$Q_k(x_a,t), k = 2,3,...,J$ have following representations:
\begin{equation}\label{qequation1}
Q_{k}(x_a,t) =  \sum_{j=k}^J
\lambda_j R_k^j(t,x_a)
  e^{t[m_j(x_a) -
  \Delta_{ab}m_{k-1}]}M_j(t), \quad k = 2,3,...,J, 
\end{equation}
where $R_k^j(t,x_a), k = 2,3...,J, j = k,k+1,...,J$ are polynomials in
$t$ with the property that $R_k^k(t,x_a) = R_k(t,x_a)$; a formal
definition of these polynomials are provided later.  
The representations (\ref{qequation1}) are useful, partly because the unknown functions $M_j(t), j = 1,..,k-1$
do not appear in $Q_{k}(x_a,t)$.  Step 2 uses the representations
(\ref{qequation1}) to show that it is possible to identify the slope $\Delta_{ab}m_k, k =
  2,...,J$ using the knowledge of $\Delta_{ab}m_1$, $Q_k(x_a,t)$ and
  $Q_k(x_b,t)$, $k = 2,...,J$ for all $x_b$ in a
  neighborhood of $x_a$.  

The identifiability of the rest of the model (at $x_a$) is then established using the
knowledge of $\Delta_{ab}m_k, k = 1,2,...,J$ and conditional moments
of $z$ given $x_a$.  

Let us start with Step 1, which derives the representation
(\ref{qequation1}) and will be summarized in Lemma \ref{lem:qrep}.  Note that the
definitions of the polynomials
$R_k(t,x_a), k = 2,...,J$ and $R_k^j(t,x_a), k = 2,...,J, j =
k,k+1,...,J$ are given in the course of our derivation.

\noindent {\bf Step 1:}  \quad Start from $k = 2$.  Define 
$$
R_2^j(t,x_a) = tD_{x_a}(m_j(x_a) - m_1(x_a)), j = 2,...,J,
$$
then       
\begin{align*}
Q_{2}(x_a,t) &= \frac{\partial}{\partial
  x_a^1} [e^{-t\Delta_{ab}m_1}M(t|x_a)] 
\\
&= \sum_{j=2}^J
\lambda_j(tD_{x_a}[m_j(x_a) - m_1(x_a)])e^{t[m_j(x_a) - \Delta_{ab}m_1]}M_j(t). 
\\
&= \sum_{j=2}^J
\lambda_jR_2^j(t,x_a)e^{t[m_j(x_a) - \Delta_{ab}m_1]}M_j(t), 
\end{align*}
yielding the desired representation for the case of $k = 2$.
Let $R_2(t,x_a)$ (used in the definition of
$A(x_a,x_b,t,2)$) be $R_2^2(t,x_a) = tD_{x_a}[m_2(x_a) -
m_1(x_a)]$. With this choice 
\begin{align*}
Q_{3}(x_a,t) &= A(x_a,x_b,t,2)Q_2(x_a,t)
\\
&=  \frac{\partial}{\partial
  x_a^1} \left [\frac{e^{-t[\Delta_{ab}m_2 - \Delta_{ab}m_1]}}{R_2(t,x_a)}Q_2(x_a,t)\right]
\\
&=  \sum_{j=3}^J
\lambda_j 
\left \{
D_{x_a}\frac{R_2^j(t,x_a)}{R_2(t,x_a)} + t
  \frac{R_2^j(t,x_a)}{R_2(t,x_a)}D_{x_a}[m_j(x_a) - m_2(x_a)]
\right \}
  e^{t[m_j(x_a) -
  \Delta_{ab}m_2]}M_j(t)
\\
&=  \sum_{j=3}^J
\lambda_j R_3^j(t,x_a)
  e^{t[m_j(x_a) -
  \Delta_{ab}m_2]}M_j(t), \quad
  \text{ say}, 
\end{align*}
and the $j=2$ term in the summation drops out.  Moreover, 
this result implies that $R_3(x_a,t)$ should be 
$$
R_3(x_a,t) = R_3^3(x_a,t) = D_{x_a}\frac{R_2^3(t,x_a)}{R_2(t,x_a)} + t
  \frac{R_2^3(t,x_a)}{R_2(t,x_a)}D_{x_a}[m_3(x_a) - m_2(x_a)]. 
$$ 
Note that the above
step requires that $R_2(t,x_a)$ is non-zero: this issue will be
discussed shortly.

The fact that the rest of $Q_{k}(x_a,t), k = 4,...,J$ have the
representations as in (\ref{qequation1}) can be shown by induction: suppose
(\ref{qequation1}) holds for $k = h$, that is
\begin{equation*}
Q_{h}(x_a,t) =  \sum_{j=h}^J
\lambda_j R_h^j(t,x_a)
  e^{t[m_j(x_a) -
  \Delta_{ab}m_{h-1}]}M_j(t). 
\end{equation*}
Define 
$$
R_{h+1}^j(t,x_a) = D_{x_a^1}\left(\frac{R_h^j(t,x_a)}{R_h^h(t,x_a)}\right)
+ t\frac{R_h^j(t,x_a)}{R_h^h(t,x_a)}D_{x_a^1}[m_j(x_a) - m_h(x_a)],
\qquad j = h+1,...,J.
$$
In what follows we sometimes write 
$$
R_k^j := R_k^j(t,x)
$$
and
$$
m_{k,l} := m_k(x) - m_l(x).
$$
as short hand.
Let $R_h(t,x_a) = R_h^h(t,x_a)$, then using this and the definition of
the operator $A(x_a,x_b,t,h)$ in (\ref{operator}), obtain
\begin{align*}
Q_{h+1}(x_a,t) &= A(x_a,x_b,t,h)Q_h(x_a,t)
\\
&= \sum_{j=h}^J
\lambda_j  A(x_a,x_b,t,h)R_h^j(t,x_a)
  e^{t[m_j(x_a) -
  \Delta_{ab}m_{h-1}]}M_j(t)
\\
& = 
D_{x_a}\lambda_he^{-tm_h(x_b)}M_h(t) + D_{x_a}\sum_{j=h+1}^J
\lambda_j  \frac{R_h^j(t,x_a)}{R_h^h(t,x_a)}
  e^{t[m_j(x_a) -
  \Delta_{ab}m_h]}M_j(t)
\\
& = \sum_{j=h+1}^J
\lambda_j  
\left\{
D_{x_a}\left(\frac{R_h^j(t,x_a)}{R_h^h(t,x_a)}\right)
+ t\frac{R_h^j(t,x_a)}{R_h^h(t,x_a)}D_{x_a}[m_j(x_a) - m_h(x_a)]
\right \}
  e^{t[m_j(x_a) -
  \Delta_{ab}m_h]}M_j(t)
\\
& = \sum_{j=h+1}^J
\lambda_j
R_{h+1}^j(t,x_a)
  e^{t[m_j(x_a) -
  \Delta_{ab}m_h]}M_j(t),
\end{align*}
which is the desired result.  The next lemma summarizes the foregoing
argument.  Notice that it relies on the assumption that
$R_k(t,x_a) = R_k^k(t,x_a), k = 2,3,...J$ are non-zero, and later we
show that the set 
\begin{equation}\label{setS}
S(x_a) = \{t|R_k(t,x_a)
\neq 0 \text{ for all } k\}
\end{equation}
is non-empty.
\begin{lem}\label{lem:qrep}
Define
$
R_2^j(t,x_a) = tD_{x_a}(m_j(x_a) - m_1(x_a)), j = 2,...,J,
$
and 
$
R_{k+1}^j(t,x_a) = D_{x_a^1}\frac{R_k^j(t,x_a)}{R_k^k(t,x_a)}
+ t\frac{R_k^j(t,x_a)}{R_k^k(t,x_a)}D_{x_a^1}[m_j(x_a) - m_k(x_a)],
k = 3,...,J, j = k+1,...,J.
$
Let $R_k(t,x_a) = R_k^k(t,x_a), k = 2,...,J$ in (\ref{operator}). 
Then
$
Q_k(x_a,t) = A(x_a,x_b,t,k-1) A(x_a,x_b,t,k-2) \cdots A(x_a,x_b,t,2) D_{x_a^1}[e^{-t\Delta_{ab}m_1}M(t|x_a)], k = 2,...,J
$
have the representations (\ref{qequation1}) on $S(x_a)$.  
\end{lem}  

\medskip

\noindent {\bf{Step 2:}} This step shows that the knowledge of the
function $Q_k(x,t)$ at $x = x_a$ and $x = x_b$ identifies $\Delta_{ab}m_k - \Delta_{ab}m_{k-1}$.  The main result is:
\begin{lem}\label{lem:diffid}
$\forall  x_b \in N^1(x_a, \delta')$,
$$
\lim_{t \rightarrow \infty} \frac{\partial}{\partial
  t} \log \left(\frac {Q_k(x_a,t)} {
    Q_k(x_b,t)} \right)= \Delta_{ab}m_k - \Delta_{ab}m_{k-1}, k = 2,3,...,J.
$$
\end{lem}

Lemmas \ref{lem:qrep} and \ref{lem:diffid} will then be useful to prove the
identifiability of $\Delta_{ab}m_k, k = 2,...,J$, for all $x_b$ in a neighborhood of $x_a$, since we already identified  $\Delta_{ab}m_1$.  The following
propositions are useful in proving Lemma
\ref{lem:diffid}.
In what follows $\deg_t(f)$ and $\text{lc}_t(f)$ denote the degree and
the leading coefficients of a polynomial $f(t)$ with respect to $t$.

\begin{prop}\label{prop:rational}
Suppose $x \in N^1(x_a, \delta')$.  Then 
$R_k(t,x)$ is a rational function of $t$ for sufficiently large $t$
and takes the following form:  
$$
R_k(t,x) = \frac {P_k(t,x)}{P_{k-1}(t,x)^2}
$$  
where $P_k(t,x), k \geq 3$ are polynomials in $t$ such that
$$
\deg_t(P_k(t,x)) = 2^{k-2} -1 
$$
and
$$
\mathrm{lc}_t(P_k(t,x)) =
(\Pi_{g=1}^{k-1}D_{x}(m_k(x) -
m_g(x)))\Pi_{j=2}^{k-1}\{(\Pi_{h=1}^{j-1}D_{x}(m_j(x) - m_h(x)))^{2^{k
  - j - 1}}\}. 
$$
\end{prop}
The proof of the proposition is given in the Appendix. 

\begin{rem}  The formula for $R_k(t,x)$ given in Proposition
  \ref{prop:rational} and the fact that $P_k(t,x)$ is a polynomial in $t$
 imply that $R_k \neq 0$ for sufficiently large
$t$ for $k = 2,3,...,J$.  Consequently $S(x_a)$ in (\ref{setS}) includes (for example) the set
 $[c,\infty)$ for some constant $c$ and therefore it is not empty.
 This is important in applying Lemma \ref{lem:qrep}.  
\end{rem}
 
\begin{prop}\label{prop:goto0}
$$
\lim_{t \rightarrow \infty} \frac{\partial}{\partial
  t} \log R_k(x,t) = 0
$$
for all $t \in \R$ and $x \in N^1(x_a,\delta')$.
\end{prop}

\begin{proof}[\textupandbold{Proof of Proposition~\ref{prop:goto0}}]
By the expression of $ R_k(x,t)$ given in Proposition \ref{prop:rational}, 
\begin{align*}
\lim_{t \rightarrow \infty} \frac{\partial}{\partial
  t} \log R_k(x,t)  & = \lim_{t \rightarrow \infty} \frac{\partial}{\partial
  t} \log  \frac {P_k(t,x)}{P_{k-1}(t,x)^2}  
\\
& =  \lim_{t \rightarrow \infty} \frac{\partial}{\partial
  t} \log P_k(t,x) - 2 \lim_{t
  \rightarrow \infty} \frac{\partial}{\partial
  t}  \log P_{k-1}(t,x).
\end{align*}
Since the Proposition shows that $R_k(x,t)$, $P_k(x,t)$ and $P_{k-1}(x,t)$ are well defined
for large $t$, so are the above limits.  But Proposition
\ref{prop:rational} also implies that $P_k(t,x)$ and $P_{k-1}(t,x)$
are polynomials in $t$ with finite 
degree, therefore the two terms are zero. 
\end{proof}

Now we are ready to prove the main result in Step 2, that is, Lemma
\ref{lem:diffid}.   
\begin{proof}[\textupandbold{Proof of Lemma~\ref{lem:diffid}}]

By Lemma \ref{lem:qrep} and Proposition \ref{prop:rational}, 
\begin{equation}
Q_{k}(x_a,t) =  \sum_{j=k}^J
\lambda_j R_k^j(t,x_a)
  e^{t[m_j(x_a) -
  \Delta_{ab}m_{k-1}]}M_j(t), \quad k = 2,3,...,J, 
\end{equation}
holds for sufficiently large $t$.  Then 
\begin{align*}
\frac{\partial}{\partial
  t} Q_k(x_a,t) &=  \sum_{j=k}^J
\lambda_j \left( \frac{\partial}{\partial
  t} R_k^j(t,x_a) + [m_j(x_a) -
  \Delta_{ab}m_{k-1}] R_k^j(t,x_a)  \right)
  e^{t[m_j(x_a) -
  \Delta_{ab}m_{k-1}]}M_j(t)
\\
&+ \sum_{j=k}^J
\lambda_j R_k^j(t,x_a)
  e^{t[m_j(x_a) -
  \Delta_{ab}m_{k-1}]}D_tM_j(t).
\end{align*}
and, for $k \leq J,$
\begin{align*}
\frac{\partial}{\partial
  t} \log(Q_k(x_a,t)) &= \frac{\frac{\partial}{\partial
  t} Q_k(x_a,t)}{Q_k(x_a,t)}
\\
&= 
\frac{
\frac{\frac{\partial}{\partial
  t} R_k(t,x_a)}{R_k(t,x_a)} + m_k(x_a) -
  \Delta_{ab}m_{k-1} + \frac{\frac{\partial}{\partial
  t} M_k(t)}{M_k(t)} 
}
{
1 + \sum_{j=k+1}^J \frac
{\lambda_j}{\lambda_k}\frac{R_k^j(t,x_a)}{R_k(t,x_a)}e^{tm_{j,k}(x_a)}\frac{M_j(t)}{M_k(t)}
}
\\
&+
\frac{
 \sum_{h=k+1}^J
\left[
(m_h(x_a) -
  \Delta_{ab}m_{k-1}) 
 \frac
{\lambda_h}{\lambda_k}\frac{R_k^h(t,x_a)}{R_k(t,x_a)}e^{tm_{h,k}(x_a)}\frac{M_h(t)}{M_k(t)}
\right]
}
{
1 + \sum_{j=k+1}^J \frac
{\lambda_j}{\lambda_k}\frac{R_k^j(t,x_a)}{R_k(t,x_a)}e^{tm_{j,k}(x_a)}\frac{M_j(t)}{M_k(t)}
}
\\
&+
\frac{
 \sum_{h=k+1}^J
\left[
 \frac
{\lambda_h}{\lambda_k}\frac{\frac{\partial}{\partial
  t} R_k^h(t,x_a)}{R_k(t,x_a)}e^{tm_{h,k}(x_a)}\frac{M_h(t)}{M_k(t)}
+
 \frac
{\lambda_h}{\lambda_k}\frac{R_k^h(t,x_a)}{R_k(t,x_a)}e^{tm_{h,k}(x_a)}\frac{\frac{\partial}{\partial
  t} M_h(t)}{M_k(t)}  
\right]
}
{
1 + \sum_{j=k+1}^J \frac
{\lambda_j}{\lambda_k}\frac{R_k^j(t,x_a)}{R_k(t,x_a)}e^{tm_{j,k}(x_a)}\frac{M_j(t)}{M_k(t)}
}.
\end{align*}
 
Using the notation in the proof of
Proposition {\ref{prop:rational}}, for all $h >k$,
\begin{align*}
\frac{R_k^h(t,x_a)}{R_k(t,x_a)} & = \frac{P_k^h(t,x_a)/(P_{k-1}^{k-1}(t,x_a))^2}{P_k^k(t,x_a)(P_{k-1}^{k-1}(t,x_a))^2}
\\
& = \frac{P_k^h(t,x_a)}{P_k^k(t,x_a)}.
\end{align*}
As noted in the Proof of Proposition \ref{prop:rational}, both
$P_k^h(t,x_a)$ and $P_k^k(t,x_a)$ are polynomials in $t$,
$P_k^k(t,x_a) \neq 0$ for sufficiently large $t$, and their degrees are equal. Hence their ratio goes to a constant as t goes to infinity:
$$\lim_{t \rightarrow \infty} \frac{R_k^h(t,x_a)}{R_k(t,x_a)} = c_{h,k,x_a}.$$

For a similar reason, using Proposition \ref{prop:goto0}, 
$$\lim_{t \rightarrow \infty} \frac{\frac{\partial}{\partial
  t} R_k^h(t,x_a)}{R_k(t,x_a)} = 0.$$
Then, using Assumption (\ref{jmodel-diff}) (\ref{jmodel-lim}), since $m_{h,k}(x_a) < -\Delta$, we know that the second and third lines of the expression of  converge to
zero as $t$ goes to $+ \infty$, and we have
\begin{equation}\label{eq:limQ}
\lim_{t \rightarrow \infty} \frac{\partial}{\partial
  t} \log(Q_k(x_a,t)) = m_k(x_a) -
  \Delta_{ab}m_{k-1} + \frac{\frac{\partial}{\partial
  t} M_k(t)}{M_k(t)}.
\end{equation}
Note that \ref{eq:limQ} holds for all $x_b \in \R^{k}$. Let us take $x_b \in N^1(x_a, \delta')$. Note that we can then also write  $\frac{\partial}{\partial
  t} \log(Q_k(x,t))$ taking $x = x_b$: the $\Delta_{ab}m_h$ terms are equal to $0$ and, again since $m_{h,k}(x_b)$ is less than $\epsilon-\Delta$, we have  
 $$\lim_{t \rightarrow \infty} \frac{\partial}{\partial
  t} \log(Q_k(x_b,t)) = m_k(x_b)  + \frac{\frac{\partial}{\partial
  t} M_k(t)}{M_k(t)},$$
so that, for all $x_b \in N^1(x_a, \delta')$, we have then
$$\lim_{t \rightarrow \infty} \frac{\partial}{\partial
  t} \log(\frac{Q_k(x_a,t)}{Q_k(x_b,t)} ) = \Delta_{ab}m_k -
  \Delta_{ab}m_{k-1}.$$

\end{proof}

To sum up, Lemma \ref{lem:diffid} together with the proof of identifiability of $\Delta_{ab}m_1$ allow, by induction, the identifiability of the slopes $\Delta_{ab}m_k$ for all $x_b \in N^1(x_a, \delta')$ and for all $k=1,...,J$:

$$\Delta_{ab}m_1 = \lim_{t \rightarrow \infty} \frac 1 t \log \left ( \frac {M(t|x_a)} {M(t|x_b)} \right ),$$
$$\Delta_{ab}m_k = \sum_{j=2}^k \lim_{t \rightarrow \infty} \frac{\partial}{\partial
  t} \log(\frac{Q_k(x_a,t)}{Q_k(x_b,t)} ) + \Delta_{ab}m_1. $$

We now state the complete identification result. For the sake of clarity, we name the point of identification $x_0$ instead of $x_a$.

\begin{ass}\label{ass:jmodel-matrix} There exists
$ X=(x_1,...,x_{J-1}) \in N^1(x_0, \delta')^{J-1}$ such that
$$
A(x_0,X) = 
\begin{pmatrix}
\Delta_{0,1} m_1 - \Delta_{0,1} m_J  & &  \ldots  & &  \Delta_{0,1} m_{J-1} - \Delta_{0,1} m_J\\
\vdots   & &   \ddots  & &  \vdots\\
\Delta_{0,J-1} m_1 - \Delta_{0,J-1} m_J  & &  \ldots  & &  \Delta_{0,J-1} m_{J-1} - \Delta_{0,J-1} m_J
\end{pmatrix}
$$
is invertible.  
\end{ass}

In the above assumption, the notation $\Delta_{0,i} m_j$ denotes $m_j(x_0) - m_j(x_i)$.

\begin{lem}\label{lem:fullidJ}  Suppose Assumptions \ref{ass:jmodel}, \ref{ass:jmodel-diff}
  and \ref{ass:jmodel-matrix} hold.  Then $F(\cdot|x), x \in
  B(x_0,\delta')$ uniquely
  determines $((\lambda_j)_{j=1..J-1},(F_j(\cdot))_{j=1..J},(m_j(x_0))_{j=1..J})$ in
  the set $(0,1)^{J-1} \times \bar{\mathcal F}(\R)^J
  \times \R^J$ up to labeling.
\end{lem}

\begin{proof}[\textupandbold{Proof of Lemma~\ref{lem:fullidJ}}]
Reproducing what was done in the Proof of Lemma \ref{lem:id}, since 
$$
\dot M(0|x_0) - \dot M(0|x)   
 = \sum_{i=1}^J \lambda_i [(m_i(x_0) - m_i(x)) -
(m_J(x_0) - m_J(x))]  +  (m_J(x_0) - m_J(x)) ,
$$
we can write
$$
\begin{pmatrix}
\dot M(0|x_0) - \dot M(0|x_1)\\
\vdots \\
\dot M(0|x_0) - \dot M(0|x_{J-1})
\end{pmatrix}
= A(x_0,X) . 
\begin{pmatrix}
\lambda_1\\
\vdots \\
\lambda_{J-1}
\end{pmatrix}
+
\begin{pmatrix}
\Delta_{0,1}m_J\\
\vdots \\
\Delta_{0,J-1}m_J
\end{pmatrix}.
$$
As Assumption \ref{ass:jmodel-matrix} guarantees the invertibility of $A(x_0,X)$ , and since the slopes of the $(m_j)_{j=1..J}$ were all previously identified, the $(\lambda_j)_{j=1..J-1}$ are identified with the formula
$$
\begin{pmatrix}
\lambda_1\\
\vdots \\
\lambda_{J-1}
\end{pmatrix}
=
A(x_0,X)^{-1}  \left[
\begin{pmatrix}
\dot M(0|x_0) - \dot M(0|x_1)\\
\vdots \\
\dot M(0|x_0) - \dot M(0|x_{J-1})
\end{pmatrix}
-
\begin{pmatrix}
\Delta_{0,1}m_J\\
\vdots \\
\Delta_{0,J-1}m_J
\end{pmatrix} \right].
$$
To identify $(m_j(x_0))_{j=1..J})$, we use the function $$C(x) = \left \{\ddot M(0|x_0)  - \ddot M(0|x) + \lambda [m_1(x_0) - m_1(x)]^2 +(1 - \lambda) [m_2(x_0) - m_2(x)]^2 \right \}/2$$ used in the Proof of Lemma \ref{lem:id}, where we can show that
$$
C(x_k) = \sum_{i=1}^J \lambda_i \; m_i(x_0) \; \Delta_{0,k}m_i,
$$
which gives
$$
\begin{pmatrix}
C(x_1) \\
\vdots \\
C(x_{J-1})\\
\dot M(0|x_0)
\end{pmatrix}
=
B(x_0, X). diag(\lambda_1,...,\lambda_J) 
\begin{pmatrix}
m_1(x_0)\\
\vdots \\
m_{J}(x_0)
\end{pmatrix},
$$
where
$$
B(x_0, X) = 
\begin{pmatrix}
\Delta_{0,1} m_1  & &  \ldots  & &  \Delta_{0,1} m_{J} \\
\vdots   & &   \ddots  & &  \vdots\\
\Delta_{0,J-1} m_1 & &  \ldots  & &  \Delta_{0,J-1} m_{J} \\
1 & & \ldots & & 1
\end{pmatrix} \text{ is observable}.
$$
$diag(\lambda_1,...,\lambda_J)$ is invertible as $\lambda_j, \, j=1..J$ are assumed to be nonzero. Since $\det B(x_0, X) = \det A(x_0,X)$, $B(x_0, X)$ is invertible. Therefore, we obtain the following identification result:
$$
\begin{pmatrix}
m_1(x_0)\\
\vdots \\
m_{J}(x_0)
\end{pmatrix}
=
diag(\lambda_1^{-1},...,\lambda_J^{-1})
B(x_0,X)^{-1}
\begin{pmatrix}
C(x_1) \\
\vdots \\
C(x_{J-1})\\
\dot M(0|x_0)
\end{pmatrix}.
$$
What now remain to be identified are the $(F_j(\cdot))_{j=1..J}$: we will again use a technique similar to what was done in the proof of Lemma \ref{lem:id}, but using Assumption \ref{ass:jmodel-matrix}.
As $ M(t|x)= \sum_{i=1}^J \lambda_i e^{tm_i(x)}M_i(t)$, considering $J$ generic points $(c_i)_{i=1..J} \in B(x_0, \delta')^J$, we have
$$
\begin{pmatrix}
M(t|c_1) \\
\vdots \\
M(t|c_{J})
\end{pmatrix}
=
D(t, c_1,...,c_J) \;
diag(\lambda_1,...,\lambda_J) \;
\begin{pmatrix}
M_1(t) \\
\vdots \\
M_J(t)
\end{pmatrix},
$$
$$ \text{where }
D(t, c_1,...,c_J)=(e^{tm_j(c_i)})_{1 \leq i,j \leq J}.
$$
We prove in the appendix (Proposition \ref{prop:zeros}) that there is a vector of $(J-1)$ points $X^{(J)}=(x_1^{(J)},...,x_{J-1}^{(J)}) \in B(x_0, \delta')^{J-1}$, such that $\mathcal{Z}= \left\lbrace t \in \R | \det D(t, x_0, x_1^{(J)},...,x_{J-1}^{(J)})=0 \right\rbrace$ is finite. Hence, we can invert $ D(t, x_0, x_1^{(J)},...,x_{J-1}^{(J)})$ for all $t \in \R \backslash \mathcal{Z}$. 
Note that we can write 
$$
D(t, x_0, x_1^{(J)},...,x_{J-1}^{(J)})= e^{t\sum_{i=1}^J m_i(x_0)}
\begin{pmatrix}
1 & & \ldots & & 1 \\
e^{-t \left( \Delta_{0,1} m_1 + \sum_{i=2}^J m_i(x_0)\right) }  & &  \ldots  & & e^{-t \left( \Delta_{0,1} m_J + \sum_{i=1}^{J-1} m_i(x_0)\right)} \\
\vdots   & &   \ddots  & &  \vdots\\
e^{-t \left( \Delta_{0,J-1} m_1 + \sum_{i=2}^J m_i(x_0)\right)} & &  \ldots  & & e^{-t \left( \Delta_{0,J-1} m_J + \sum_{i=1}^{J-1} m_i(x_0)\right)} \\
\end{pmatrix},
$$
and since $(x_1^{(J)},...,x_{J-1}^{(J)}) \in B(x_0, \delta')^{J-1}$, by the above result and Lemma \ref{lem:diffid}, $D(t, x_0, x_1^{(J)},...,x_{J-1}^{(J)})$ is identified.
Therefore $(M_i(t))_{i=1..J}$ are identified for all $t \in \R \backslash \mathcal{Z}$ and since the $(M_i(t))_{i=1..J}$ have domain $(-\infty,+\infty)$, we know that they are continuous (see, e.g, \citeasnoun{gut2013probability} Theorem 8.3 p190) on $\R$. As for each $M_i$, there is a unique continuous extension on $\R$ of its restriction to $\R \backslash \mathcal{Z}$, the $J$ functions are identified. By the same argument of uniqueness of the Laplace transform for a distribution function, this leads to the identification of the $F_i$.
\end{proof}

Having showed identification of our model assuming knowledge of $J$, we now consider the case where $J$ is unknown, and show it is identified, using the observable sequence of functions $(Q_k)_{k=1,...}$. As we see below, the number of mixture components $J$ is equal to the largest  $j$ for which the function $Q_j$  not identically $0$ in $t$. Therefore one can sequentially compute the $\Delta_{ab}m_j$ using $Q_j$, for increasing $j$. Once there exists $j_0$ such that $Q_{j_0} =0$, then $J=j_0-1$.

\textbf{
\begin{prop}\label{prop:idJ}
$$J = \max \left\lbrace j \geq 1 | \exists t_0 \in \R, Q_j(x_a,t_0) \neq 0 \right\rbrace . $$
\end{prop}
}
 
\begin{proof}[\textupandbold{Proof of Proposition~\ref{prop:idJ}}]

$$Q_J = \lambda_J R_J(t,x_a) e^{t(m_J(x_a)-\Delta_{ab}m_{J-1})}M_J(t),$$
therefore 
$$Q_{J+1}(x_a,t)= \lambda_J \frac{\partial}{\partial
  x_a^1}\left[\frac{R_{J}(t,x_a)}{R_{J}(t,x_a)} e^{-t m_J(x_b)} M_J(t)\right] = 0, \text{ for all } t \in \R. $$
We actually see that we cannot calculate any $\Delta_{ab} m_{J+1}$ with the method of Lemma \ref{lem:diffid} because of the logarithm: the identification process must be stopped here.

Reciprocally, if $j_0 \leq J$, then for some $t_0 \in \R, Q_{j_0}(x_a, t_0) \neq 0$. 
Indeed, $j_0 \leq J \Rightarrow \forall j_0 \leq k \leq J,\, \lambda_k > 0 $
and we can write
$$
Q_{j_0}(x_a,t) = \lambda_{j_0} R_{j_0}(t,x_a)M_{j_0}(t)e^{tm_{j_0}(x_a)-\Delta_{ab}m_{j_0-1} } \left( 1 + \sum_{j=k_0+1}^J \frac
{\lambda_j}{\lambda_{j_0}}\frac{R_{j_0}^j(t,x_a)}{R_{j_0}(t,x_a)}e^{tm_{j,j_0}(x_a)}\frac{M_j(t)}{M_{j_0}(t)} \right).
$$
By proposition \ref{prop:rational}, we know that $\deg_tR_{j_0}^j = 1$, so there is a constant $b_{x_a, x_b, j, j_0} > 0$ such that 
$$\frac{R_{j_0}^j(t,x_a)}{R_{j_0}(t,x_a)} \xrightarrow[t\to\infty]{} b_{x_a, x_b, j, j_0}.$$
Using Assumption \ref{ass:jmodel-diff} (\ref{jmodel-lim}),	 since $m_{h,k}(x_a) < -\Delta$, each term in the sum on the right hand side goes to 0 as $t$ goes to $\infty$, implying that for large enough $t$, the term in parenthesis is strictly positive, that is, nonzero.

\end{proof}


\begin{rem}\label{rem:dble}
Note that it is not essential for our identification strategy to assume to impose Assumption \ref{ass:jmodel} \eqref{jmodel-dble} $ m$ is $J$-times differentiable in one argument, as stated right after the assumption.
   Note that the use of the differentiation operator $\frac{\partial}{\partial x^1}$ in the linear operator $A$ is motivated by the fact that it eliminates terms that do not involve $x_a$, therefore with respect to which argument we differentiate is unimportant.  The same identification argument applies if at each application of the operator $A$ in the recursive formula \eqref{eq:recur} time we  use  $\frac{\partial}{\partial x^\ell}$ with a different $\ell \{1,...,k\}$ instead of keeping on using the same differential operator  $\frac{\partial}{\partial x^1}$ as in the current proof.      What we need is, as noted before, that $m$ can be differentiated up to a $J$-th order multi-index.   This is  less stringent than Assumption \ref{ass:jmodel} \eqref{jmodel-dble}, though we chose to state the result in the current form for notational simplicity.  
\end{rem}

\section{Application to Identifiability of Auction Models with Unobserved Heterogeneity}\label{sec:auction}

  It is of great interest to demonstrate that the preceding identification results potentially apply to nonparametric analysis of auction models with unobserved heterogeneity.  As recognized in the recent literature, failing to properly taking account for unobserved heterogeneity in empirical auction models can lead to grossly misleading policy implications and counterfactual analyses.  The reader is referred to \citeasnoun{haile2018unobserved} for various approaches to nonparametric identifiabilty in auction models when unobserved heterogeneity is present.    Here we focus on  application of the preceding mixture identification  results to models with auction-specific unobserved heterogeneity.  In particular, we focus on   a symmetric affiliated auction model as considered in \citeasnoun{milgrom1982theory}.    Suppose that valuations have the following multiplicative form, with $J$ unknown types of auctions      
\begin{equation}\label{eq:valuation}
V^k = \Gamma_j(x)U_j^k \quad \text{ with probability } \lambda_j, 1 \leq j \leq J
\end{equation}
where $V^k$ is the valuation of bidder $k$, $1 \leq k \leq I$, who knows the number of bidders $I$, observed characteristics $x$, unobserved heterogeneity (i.e. unobserved type of auction) $j$, and a signal $S^k$.  The function $\Gamma_j(x)$ depends on the two characteristics $x$ and $j$.  The term $U_j^k$ can be interpreted as the ``homogenized valuation" for bidder $k$, as used in \citeasnoun{haile2003nonparametric}.  Let $B^k$ denote the bid of bidder $k$.    The observables in this application is $(I,B^1,...,B^I,x)$.  The rest remain unobserved.

We maintain that there are finite number of types in terms of auction heterogeneity.  It is then  possible to establish identification under quite weak assumptions.  In the following result note that (i) valuations can be affiliated, and  (ii) unobserved heterogeneity is treated flexibly, as not only it can affect valuations through the index function $\Gamma_j$ in an unrestricted way, the distribution of the homogenized valuation $U_j^k$ is allowed to depend on $j$ freely.  Property (i) is important, as many preceding nonparametric identification results for auction with unobserved heterogeneity focus on the independent private value (IPV) model, as they tend to impose independence assumptions across valuations, with the exception of \citeasnoun{CHS}.  For example, Property (i) implies that the result in this section applies to the common values model.   Property (ii) about the flexible treatment of homogenized valuations is apparently new.        	
	
Assume 
		\begin{equation}\label{eq:aucdep}
	(U_j^1,\dots ,U_j^I,S^1,\dots ,S^I) \indep x | I
	\end{equation}
	for every $j \in \{1,...,J\}$.  Note that standard approaches to deal with unobserved  heterogeneity do so through the index function $\Gamma_j$, and would not allow $(U^1,...,U^I)$ to depend on $j$ . 
	Define
	$$
	w(S,I,x,j) :=  \mathrm E\left[V^k | S^k = \max_{i \neq k, 1 \leq i \leq I}S^i = S, I, x, j\right]
	$$
	which corresponds to the expected value of a bidder's valuation conditional on $I$, $x$, $j$, and  the event that her equilibrium bid is pivotal.  This is  a quantity sometimes simply called ``pivotal expected value".  Let $w^k := w(S^k,I,x,j), 1 \leq k \leq I$  denote the pivotal expected value of the $k$-th bidder  (whose signal is $S^k$)  in an auction with characteristics $(x,j)$ and $I$ bidders.    The goal here is to identify  the joint distribution of $(w^1,...,w^I)$ in an auction with $(x,I,j)$, along with the distribution $(\lambda_1,...,\lambda_J)$ of the unobserved heterogeneity.  Note that such knowledge is sufficient to address  important questions often asked in practice: see, for example,  footnote 9 of  \citeasnoun{haile2018unobserved} for further discussions.

	The above setting  implies an expression of $w$ of the following form
	\begin{equation}\label{eq:hom-mix}
	w\left(S,I,x,j\right) = \Gamma_j(x)\omega\left(S,I,j\right),
	\end{equation}
	where $\omega\left(S;I,j\right) = E[V^k|S^k = \max_{i \neq k}S^i = S, I,j]$.
	Like the homogenized valuation $\{\{U_j^k\}_{k=1}^I\}_{j=1}^J$, $\omega_j^k := \omega\left(S^k,I,j\right)$ is interpreted as a homogenized pivotal expected value of bidder $k$ in an auction of unobserved type $j$.    It is well-known that the equilibrium bidding function preserves multiplicative separability in \eqref{eq:valuation}, hence \eqref{eq:hom-mix}, for each bidder $k$.  Thus we obtain
	$$
	B^k = \Gamma_j(x)R_j^k,
	$$
	where $R_j^k$ is the homogenized valuation of bidder $k$ in type $j$ auction.  Note that the unobserved auction type can affect equilibrium bids through two channels, that is, the index function $\Gamma_j$ and the homogenized bid $R_j^k$.     Define $b^k = \log B^k$,   $\gamma_j(x) := \log \Gamma_j(x)$ and $r_j^k := \log R_j^k$, then we have 
	\begin{equation}\label{eq:bideq}
b^k = \gamma_j(x) + r_j^k, \quad 1 \leq j \leq J, \quad 1 \leq k \leq I.
	\end{equation}
		Note that \eqref{eq:aucdep} implies 
	\begin{equation}\label{eq:bidindep}
	(r_k^1,\dots,r_k^I) \indep x \quad {\text{for every $j$}}
	\end{equation}
	conditional on $I$.

	We now invoke Lemma \ref{lem:fullidJ} to establish identification of this model.    One of the main objects to be identified is the $I$-dimensional joint distribution of the pivotal expected values $w^1,...,w^I$ conditional on $(x,j,I)$, and our identification strategy works for each value of $I$.    Thus in the rest of this section we treat $I$ as being fixed at a value, and suppress the index $I$  unless necessary.  Let  $c  = (c_1,...,c_I)' \in {\mathbb R}^I$, and define 
	 $b(c) := \sum_{k=1}^{I}c_k b_{k}$, $C(c) :=  \sum_{k=1}^{n}c_k$ and $r_j(c) :=   \sum_{k=1}^{I}c_i r_j^k$.  By \eqref{eq:bideq} and the finite mixture structure of the evaluation  in \eqref{eq:valuation} we have 
	 $$
	 b(c) = C(c)\gamma_j(x) + r_j(c) \quad \text{ with probability } \lambda_j, 1 \leq j \leq J
	 $$
	 where  $r(c) \indep x$  by \eqref{eq:bidindep}. 
	Let  $\left(b(c),\{C(c)\gamma_j(\cdot)\}_{j=1}^J,\{r_j(c)\}_{j=1}^J\right)$  play the role of   $\left(z,\{m_j(\cdot)\}_{j=1}^J,\{\epsilon_j\}_{j=1}^J\right)$ in Lemma \ref{lem:fullidJ}, then $\left(C(c) \gamma_j(\cdot), \lambda_j\right)$ and the distribution of $r_j(c)$ are all identified for every $c \in {\mathbb R}^n$ and each $j \in \{1,...,J\}$.   Moreover, we now know $\gamma_j(\cdot)$,  $j \in \{1,...,J\}$  since $C(c)$ is known.  Note that for each $j$, the marginal distribution of every linear combination $r_j(c)$ of the $I$-vector $(r_j^1,...,r_j^I)$  is identified as $c \in \R^I$ can be chosen arbitrarily.  Then by Cram\'er-Wold the joint distribution of $(r_j^1,...,r_j^I)$ is obtained for each $j$.   Apply this and the knowledge of $\gamma_j$ to equation  \eqref{eq:bideq} to determine  the joint distribution  $(b^i,...,b^I)|x,j,I$.  Using the first order condition for equilibrium bidding (see, e.g. \citeasnoun{haile2003nonparametric},  \citeasnoun{athey2007nonparametric} and Equation (2.4) in \citeasnoun{haile2018unobserved}) we can now back out the joint distribution  of   $(w^1,...,w^I)|x,j,I$  as desired.  Note that the number of (unobserved) auction types $J$ is also identified by  Proposition \ref{prop:idJ}.

\section{Nonparametric estimation for $J=2$}\label{sec:estimation}
This section develops a fully nonparametric estimation procedure based on our third identification result in Section \ref{subsec:third} where the number of mixture components is two.   We first estimate the slopes of $m_1$ and $m_2$ nonparametrically.  
Define 
$\Delta  = m_1(x_1) - m_1(x_0)$ and
$\nabla  = m_2(x_1) - m_2(x_0).$
Let us reintroduce notations. We write, for $j=1,2,$ 
\begin{align*}
\phi^j(s) = \mathbb{E}(e^{isZ}|X=x_j) ,\ \phi_l (s) & = \mathbb{E}(e^{is\epsilon_l}) = \int e^{i \epsilon s} \mathrm{d}F_l(\epsilon),\ \hat{\phi}^j(s) = \frac{\sum_{p=1}^n e^{is Z_p} K(\frac{X_p - x_j}{b_n})} {\sum_{p=1}^n  K(\frac{X_p - x_j}{b_n})},\\
M^j(t)= \mathbb{E}(e^{tZ}|X=x_j),\ M_i (t) & = \mathbb{E}(e^{t\epsilon_i}) = \int e^{\epsilon t} \mathrm{d}F_i(\epsilon),\ \hat{M}^j(t) = \frac{\sum_{p=1}^n e^{t Z_p} K(\frac{X_p - x_j}{h_n})} {\sum_{p=1}^n  K(\frac{X_p - x_j}{h_n})},
\end{align*}
where $F_i$ is the cumulative distribution function of $\epsilon_i,$ $h_n$ and $b_n$ are carefully chosen bandwidths for kernel density estimation. $\hat{M}^j(t)$ and $\hat{\phi}^j(s)$ are  the Nadaraya-Watson regression estimators of respectively the conditional moment generating function and conditional characteristic function of $Z,$ when $X = x_j.$ $X$ being a vector, the kernel function $K$ can have a product form such as $K(X)=\Pi_{l=1}^k k(X^{(l)}).$

Our estimators are
\begin{align*}
\hat{\Delta} & = \frac{1}{t_n} \log \left( \frac{\hat{M}^1(t_n)}{\hat{M}^0(t_n)}  \right), \\
\hat{\nabla} & = \frac{-i}{a_n} \mbox{ Log} \left( \frac{\hat{\phi}^1(s_n+a_n)}{\hat{\phi}^0(s_n+a_n)} \left( \frac{\hat{\phi}^1(s_n)}{\hat{\phi}^0(s_n)} \right)^{-1}  \right),
\end{align*}
where $(a_n)_n, (s_n)_n$ and $(t_n)_n$ are tuning parameters such that  $a_n \to 0,$ $s_n \to \infty$ and $t_n \to \infty$. The notation Log$(\cdot)$ as before corresponds to the principal value of the logarithm of $\cdot$.

We enumerate here the assumptions on the kernel function needed to compute the rates of our estimators.
\begin{ass}\label{ass-kernel}
The kernel function $K(.)$ must satisfy the following conditions,

$\int |K(U)| \, \mathrm{d} U < \infty $ , $\int K(U)\, \mathrm{d}U = 1,$ $\lim_{||U|| \to \infty} UK(U) \to 0,$ 

$\int K(U)^{2} \, \mathrm{d}U < \infty,$ $\int |K(U)| \,  U'U \, \mathrm{d}U < \infty,$ $\int K(U) \, U \, \mathrm{d}U = 0,$ 

$\exists \alpha_0, \alpha \leq \alpha_0\ \Rightarrow \int e^{\alpha ||U||} |K(U)| \, U'U \, \mathrm{d}U < \infty, \int e^{\alpha ||U||} K(U)^2 \, \mathrm{d}U < \infty.$

\end{ass}

We need the following assumptions on the model parameters.
\begin{ass}\label{ass-convergence}

\begin{enumerate}[\textup{(}i\textup{)}]
\item\label{c-xdensity}
$f_X,$ the density of the random variable $X$, has continuous second order partial derivatives. $f_X$ and all its first and second order partial derivatives are bounded on $\R^k.$  $f_X(x_i)>0$, for $i=0,1.$


\item\label{c-differentiability}
$m_i,$ $i=1,2$ have continuous second order partial derivatives, and all their first and second order partial derivatives are bounded on $\R^k.$

\item\label{c-bandwidth}
$h_n  \underset{n \rightarrow \infty }{\to} 0,$ $nh_n^k  \underset{n \rightarrow \infty }{\to} \infty$,   and $b_n  \underset{n \rightarrow \infty }{\to} 0,$ $nb_n^k  \underset{n \rightarrow \infty }{\to} \infty,$

\item\label{c-bias}
$t_n \underset{n \rightarrow \infty }{\to} \infty,$ $t_n h_n \underset{n \rightarrow \infty }{\to} 0,$ and $s_n \underset{n \rightarrow \infty }{\to} \infty,$ $s_n b_n \underset{n \rightarrow \infty }{\to} 0.$

\end{enumerate}
\end{ass}

\begin{ass}\label{ass-ratio rates}
\begin{enumerate}[\textup{(}i\textup{)}]
	
\item\label{indep}		$\epsilon_1|x \sim F_1$ and  $\epsilon_2|x \sim
F_2$ at all $x \in \R^k$ where $F_1$ and $F_2$ do not depend on the value of $x$,
	
\item\label{mgfExist} The domains of $M_1(t)$ and $M_2(t)$ are $[0,\infty)$,

\item\label{mgf}
$\forall \epsilon >0 \mbox{, } e^{\epsilon t} \frac{M_2(t)}{M_1(t)} \underset{t \rightarrow \infty }{=} O(\mu(t)),$ holds for some $\mu (\cdot),$ where $\mu(t) \xrightarrow[t \rightarrow \infty] {} 0,$

\item\label{cf}
$\frac{\phi_1(s)}{\phi_2(s)} \underset{s \rightarrow \infty }{=} O(f(s)),$
holds for some $f(\cdot),$ where $f(t) \xrightarrow [t \rightarrow \infty] {} 0.$

\end{enumerate}
\end{ass}

\begin{prop}\label{prop:estimation}
Suppose Assumptions  \ref{ass-kernel}, \ref{ass-convergence}  and \ref{ass-ratio rates} hold.

Then
\begin{enumerate}[\textup{(}i\textup{)}]

\item\label{delta} $\hat{\Delta} - \Delta = O_{\mathbb{P}} \left[ \frac{\mu (t_n)}{t_n} + \frac{1}{t_n} \left( (t_n h_n)^4 + \frac{1}{n h_n^k} \frac{M_1(2t_n)}{M_1(t_n)^2} \right)^{\frac{1}{2}} \right],$ where we assume $\frac{1}{n h_n^k} \frac{M_1(2t_n)}{M_1(t_n)^2} \underset{n \rightarrow \infty }{\to} 0$

\item\label{nabla} $\hat{\nabla} - \nabla = \frac{1}{a_n} O_{\mathbb{P}}  \left[ f(s_n+a_n) +f(s_n) +  \left((b_n s_n)^4 + \frac{1}{n b_n^{k} |\phi_2(s_n + a_n)|^2 } \right)^{1/2} + \left((b_n s_n))^4 + \frac{1}{n b_n^{k} |\phi_2(s_n)|^2 } \right)^{1/2} \right]$

\end{enumerate}

\end{prop}

\subsection*{Proof 1 part 1}

\begin{proof}[\textupandbold{Proof of Proposition~\ref{prop:estimation} \ref{delta}}]

The estimator can be decomposed as 
$$
\hat{\Delta} = \frac{1}{t_n} \log \left( \frac{\hat{M}^1(t_n)}{\hat{M}^0(t_n)}  \right)=\frac{1}{t_n} \log \left( \frac{M^1(t_n)}{M^0(t_n)}  \right)  +  \frac{1}{t_n} \log \left( \frac{\hat{M^1}(t_n)}{M^1(t_n)}  \right)
- \frac{1}{t_n} \log \left( \frac{\hat{M^0}(t_n)}{M^0(t_n)}  \right).
$$

The first term in the decomposition is deterministic. Using the proof of Lemma \ref{lem:F3rdresult}, this approximation error can be written

\begin{align*}
\frac{M^1(t_n)}{M^0(t_n)} 
= \, e^{t_n \Delta} \, \frac 
  { 1+   \frac {1 - \lambda}{\lambda} 
     e^{t_n [m_2(x_1) - m_1(x_1)] } \frac {M_2(t_n)}{M_1(t_n)}} 
  { 1 + \frac{1 - \lambda}{\lambda}
    e^{t_n[m_2(x_0) - m_1(x_0)]}\frac {M_2(t_n)}{M_1(t_n)}} \,
= \, e^{t_n \Delta} \left[ 1 + O(\mu (t_n)) \right],
\end{align*}
where the last equality holds using Assumption \ref{ass-ratio rates} (\ref{mgf}). This gives
$$ 
 \frac{1}{t_n} \log \left( \frac{M^1(t_n)}{M^0(t_n)}  \right) 
= \Delta +  O(\frac{\mu (t_n)}{t_n}  ).
$$

Let us now focus on the terms $\frac{1}{t_n} \log \left( \frac{\hat{M}^j(t_n)}{M^j(t_n)}  \right),$ the two estimation errors. We write
$$
\hat{M}^j(t_n) = \frac{\frac{1}{nh_n^k} \sum_{p=1}^n e^{t_n Z_p} K(\frac{X_p - x_j}{h_n})} {\frac{1}{nh_n^k}  \sum_{p=1}^n  K(\frac{X_p - x_j}{h_n})} = \frac{\hat{N}^j(t_n)}{\hat{D}^j},
$$
and have  
\begin{equation}\label{estim1}
\frac{\hat{M}^j(t_n)}{M^j(t_n)} =  \frac{\hat{N}^j(t_n)}{\hat{D}^j M^j(t_n)} = \frac{f_X(x_j)}{\hat{D}^j} \frac{\hat{N}^j(t_n)}{ f_X(x_j) M^j(t_n)}. 
\end{equation}

In what follows, we treat separately the two ratios appearing in the last equality in (\ref{estim1}), showing that they both converge to 1. Part of the reasoning will be different from usual kernel regression. Indeed,for the second ratio, we need to keep the denominator to compute the convergence rate to counterbalance the numerator going to infinity, as the parameter $t_n$ goes to infinity.

Under Assumptions \ref{ass-kernel} and \ref{ass-convergence}, we know from usual results on kernel density estimation 
that when computing the Mean Square Error of the term $\frac{\hat{D}^j}{f_X(x_j)},$ the bias is of order $h_n^2$ and the variance of order $\frac{1}{n h_n^k}$, so that 
\begin{equation}
 \frac{\hat{D}^j}{f_X(x_j)} = 1 + O_{\mathbb{P}} \left( h_n^4 + \frac{1}{n h_n^k} \right) ^{\frac{1}{2}}.
 \label{estim2}
\end{equation}

As for the second ratio in the decomposition of (\ref{estim1}) , the dependence in $t_n$ requires new assumptions when computing bias and variance.
For the bias term we denote $G_n(x)= f_X(x) \mathbb{E}(e^{t_n Z}|X=x) $, then by definition of the estimator, 
$$\mathbb{E}(\hat{N}^j(t_n)) = \mathbb{E} \left(\frac{1}{n h_n^k} \sum_{p=1}^n e^{t_n Z_p} K(\frac{X_p - x_j}{h_n}) \right)
= \int_{U \in \R^k} G_n(x_j + h_n U) K(U) \mathrm{d}U.$$

By Assumption \ref{ass-convergence}, $G_n$ is twice continuously differentiable. Since the kernel is of order 2 (Assumption \ref{ass-kernel}), by virtue of the Mean Value Theorem, we have
$$\mathbb{E} \left( \frac{\hat{N}^j(t_n)}{ f_X(x_j) M^j(t_n)} \right) - 1 = \frac{1}{G_n(x_j)} \int \frac{h_n^2}{2} U'. \nabla^2 G_n [x_j + h_n \tau_n(U) U]. U K(U) \mathrm{d}U$$
where $\tau_n(u) \in [0;1]$ and  $\nabla^2 G_n(x)$ is the hessian matrix of the function $G_n$ evaluated at $x$. We know that 
$G_n(x) = f_X(x) [\lambda e^{t_n m_1(x)} M_1(t_n) + (1-\lambda) e^{t_n m_2(x)} M_2(t_n)].$ 
Twice differentiation gives

\begin{align*}
\nabla^2 G_n(x)  =  & \, \lambda e^{t_n m_1(x)} M_1(t_n)\ \mathrm{\{}  t_n^2 f_x(x) \nabla m_1(x) \nabla m_1(x)' \\
& + t_n (\nabla m_1(x) \nabla f_X(x)' + \nabla f_X(x) \nabla m_1(x)' + f_X(x) \nabla^2 m_1(x) )  \\
&  + \nabla^2 f_X(x)  \mathrm{\}} \\
&  + (1-\lambda) e^{t_n m_2(x)} M_2(t_n)\  \mathrm{\{}  t_n^2 f_x(x) \nabla m_2(x) \nabla m_2(x)'  \\
& + t_n (\nabla m_2(x) \nabla f_X(x)' + \nabla f_X(x) \nabla m_2(x)' + f_X(x) \nabla^2 m_2(x) ) \\
& + \nabla^2 f_X(x)  \mathrm{\}}. \\
  = & \, \lambda \, e^{t_n m_1(x)} M_1(t_n) \{  t_n^2 a_1(x) + t_n b_1(x) + c_1(x)\} \\
&  + (1-\lambda) \, e^{t_n m_2(x)} M_2(t_n)  \{  t_n^2 a_2(x) + t_n b_2(x) + c_2(x)\}.
\end{align*}

By boundedness of the first order partial derivatives of $m_i,$ $i=1,2,$ 
$$\exists \delta, \forall (x,U) \in \R^k \times \R^k, |m_i(x+h_n \tau_n(U) U) - m_i(x)| \leq \delta h_n ||U||,$$
implying that $e^{t_n m_i(x+h_n \tau_n(U) U) - m_i(x)} \leq e^{\delta t_n h_n ||U||}.$ 
Therefore, as $G_n(x) \geq f_X(x) \lambda e^{t_n m_1(x)} M_1(t_n) ,$ 
$$\frac{\lambda e^{t_n m_1(x_j + h_n \tau_n(U) U)} M_1(t_n)}{G_n(x_j)} \leq \frac{e^{\delta h_n t_n ||U||}}{f_X(x_j)} \leq \frac{e^{C ||U||}}{f_X(x_j)} ,$$
for some $C \leq \alpha_0,$ for $n$ large enough, under Assumption \ref{ass-convergence} (\ref{c-bias}). The same holds for the $(1-\lambda)$ term.
By Assumption \ref{ass-convergence}, $a_1(x + h_n \tau_n(U) U)$ is bounded by a constant as well as the other coefficients of the $t_n$ polynomial in the expression of $\nabla^2 G_n [x_j + h_n \tau_n(U) U]$. This, together with the previous argument, implies that $\frac{1}{G_n(x_j)} \int U'. \nabla^2 G_n [x_j + h_n \tau_n(U) U]. U K(U) \mathrm{d}U = O(t_n^2).$ The rate of the bias term can therefore be bounded,
$$\mathbb{E} \left( \frac{\hat{N}^j(t_n)}{ f_X(x_j) M^j(t_n)} \right) - 1 = O (t_n h_n)^2. $$

For the variance term, an upper bound is
\begin{align*}
 & \frac{1}{n } \mathbb{E}  \left[ \left(\frac{1}{h_n^k M^j(t_n) f_X(x_j)} e^{t_n Z} K(\frac{X - x_j}{h_n})\right)^2 \right] \\
&= \frac{1}{n [ h_n^k M^j(t_n) f_X(x_j)]^2} \int \mathbb{E}(e^{2t_n Z}|X)K(\frac{X - x_j}{h_n})^2 f_X(X) \mathrm{d}X \\
&= \frac{1}{nh_n^k} \int \frac{\mathbb{E}(e^{2t_n Z}|h_n U + x_j)}{\mathbb{E}(e^{t_n Z}|x_j)^2} \frac{f_X(h_n U + x_j)}{f_X(x_j)^2} K(U)^2 \mathrm{d}U \\
&= \frac{1}{nh_n^k} \int \frac{\lambda e^{2t_n m_1(h_n U + x_j)}M_1(2t_n) + (1 - \lambda) e^{2t_n m_2(h_n U + x_j)}M_2(2t_n)}{\left( \lambda e^{t_n m_1(x_j)}M_1(t_n) + (1 - \lambda) e^{t_n m_2(x_j)}M_2(t_n)\right)^2} \ \frac{f_X(h_n U + x_j)}{f_X(x_j)^2} K(U)^2 \mathrm{d}U \\
& \leq \frac{1}{nh_n^k} \frac{M_1(2t_n)}{M_1(t_n)^2}\int e^{2 \delta t_n  h_n ||U||} \frac{\lambda  + (1 - \lambda) e^{2t_n (m_2(x_j) - m_1(x_j)}\frac{M_2(2t_n)}{M_1(2t_n)}}{\left( \lambda  + (1 - \lambda) e^{t_n (m_2(x_j) - m_1(x_j))} \frac{M_2(t_n)}{M_1(t_n)}\right)^2} \ \frac{f_X(h_n U + x_j)}{f_X(x_j)^2} K(U)^2 \mathrm{d}U
.
\end{align*}
Using Assumption \ref{ass-convergence} (\ref{c-bias}) and Assumption (\ref{ass-ratio rates}) for $n$ large enough, the integrand is bounded above by $ C' e^{C ||U||} K(U)^2, \, \forall \, U \in \R^k,$ for some $C$ independent of $n$, $C \leq \alpha_0,$ $C' > 0.$ Assumption (\ref{ass-kernel}) and (\ref{ass-convergence}) guarantee that the variance is of order $O ( \frac{1}{n h_n^k} \frac{M_1(2t_n)}{M_1(t_n)^2}). $
Therefore, 
\begin{equation}
\frac{\hat{N}^j(t_n)}{ f_X(x_j) M^j(t_n)} = 1 + O_{\mathbb{P}} \left( (t_n h_n)^4 + \frac{1}{n h_n^k} \frac{M_1(2t_n)}{M_1(t_n)^2} \right) ^{\frac{1}{2}}.
\label{estim3}
\end{equation}
With (\ref{estim1}), (\ref{estim2}) and (\ref{estim3}), and given that by Jensen's inequality $\frac{M_1(2t_n)}{M_1(t_n)^2} \geq 1,$ the second ratio in (\ref{estim1}) dominates. 
$\frac{\hat{M}^j(t_n)}{M^j(t_n)} - 1 = O_{\mathbb{P}} \left( (t_n h_n)^4 + \frac{1}{n h_n^k} \frac{M_1(2t_n)}{M_1(t_n)^2} \right) ^{\frac{1}{2}}$.

This finally gives
\begin{align*}
\hat{\Delta} - \Delta &=  O(\frac{\mu (t_n)}{t_n}) + \frac{1}{t_n} \log( \left( 1 +  O_{\mathbb{P}} \left( (t_n h_n)^4 + \frac{1}{n h_n^k} \frac{M_1(2t_n)}{M_1(t_n)^2} \right) ^{\frac{1}{2}} \right)^2 )\\
&= O_{\mathbb{P}} \left[ \frac{\mu (t_n)}{t_n} + \frac{1}{t_n} \left( (t_n h_n)^4 + \frac{1}{n h_n^k} \frac{M_1(2t_n)}{M_1(t_n)^2} \right)^{\frac{1}{2}} \right],
\end{align*}
since we assumed that $\frac{1}{n h_n^k} \frac{M_1(2t_n)}{M_1(t_n)^2} \underset{n \to \infty}{\rightarrow} 0.$
\end{proof}

 \subsection*{Proof 1 part 2}
\begin{proof}[\textupandbold{Proof of Proposition~\ref{prop:estimation} \ref{nabla}}]
The estimator is $\hat{\nabla} = \frac{-i}{a} \mbox{ Log} \left( \frac{\hat{\phi}^1(s_n+a)}{\hat{\phi}^0(s_n+a)}  \left( \frac{\hat{\phi}^1(s_n)}{\hat{\phi}^0(s_n)} \right)^{-1}   \right).$ We first compute the rate of convergence of $\frac{\hat{\phi}^0(s_n)}{\hat{\phi}^1(s_n)}$, in a fashion similar to the proof above. From the identification result in Section \ref{sec:switch}, we know
$$
\lim_{s \rightarrow \infty} \frac{-i}{a}  \mbox{ Log} \left( \frac{\phi(s+a|x_1)}{\phi(s+a|x_0)} \left( \frac{\phi(s|x_1)}{\phi(s|x_0)} \right)^{-1}  \right) = \frac{1}{a} \left( a \nabla + 2 \pi \left \lfloor{\frac{1}{2} - \frac{a \nabla}{2 \pi}}\right \rfloor \right) .
$$
Because we do not know the interval on which the identifying equation will be a constant of $a$, we  plug in a sequence $a_n$ going to zero instead of a fixed $a$. For the approximation error, we have
\begin{align*}
\frac{\phi^1(s_n)}{\phi^0(s_n)}& = \frac{\lambda e^{i s_n m_1(x_1)}\phi_1(s_n) + (1 - \lambda)
    e^{s_n m_2(x_1)}\phi_2(s_n)}{\lambda e^{i s_n m_1(x_0)}\phi_1(s_n) + (1 - \lambda)
    e^{s_n m_2(x_0)}\phi_2(s_n)}\\
& = e^{i s_n \nabla} \frac{\frac{\lambda}{1-\lambda} e^{i s_n (m_1(x_1) - m_2(x_1))} \frac{\phi_1(s_n)}{\phi_2(s_n)} + 1}{\frac{\lambda}{1-\lambda} e^{i s_n (m_1(x_0) - m_2(x_0))} \frac{\phi_1(s_n)}{\phi_2(s_n)} + 1}\\
& = e^{i s_n \nabla} (1 + O(f(s_n))).
\end{align*}

To compute the estimation error, the scheme is initially similar to the previous proof. We write 
$$\hat{\phi}^j(s_n) = \frac{\frac{1}{n h_n^k} \sum_{p=1}^n e^{is_n Z_p} K(\frac{X_p - x_j}{b_n})} {\frac{1}{n h_n^k} \sum_{p=1}^n  K(\frac{X_p - x_j}{b_n})} = \frac{\hat{num}^j(s_n)}{\hat{denom}^j}
$$
and work with an equation similar to (\ref{estim1}), here
$$ \frac{\hat{\phi}^j(s_n)}{\phi^j(s_n)}= \frac{f_X(x_j)}{\hat{denom}^j}  \ \frac{\hat{num}^j(s_n)}{ f_X(x_j) \phi^j(s_n)}
$$
We compute the convergence rate of the ratios in the last equality. 

As in (\ref{estim2}), we know
$\frac{\hat{denom}^j}{f_X(x_j)} = 1 + O_{\mathbb{P}} \left( b_n^4 + \frac{1}{n b_n^k} \right) ^{\frac{1}{2}}.$
Now, let $A_n = \frac{\hat{num}^j(s_n)}{ f_X(x_j) \phi^j(s_n)}:$ $A_n \in \mathbb{C}.$ Working with complex numbers for this proof, we use $|.|$ to denote a modulus. Let us focus on the bias term of $A_n$. We write $g_n(x) = f_X(x) \mathbb{E}(e^{i s_n Z}|X=x)$ so that $A_n = \frac{\hat{num}^j(s_n)}{g_n(x)}$.

Since $g_n(x) = f_X(x) (\lambda e^{i s_n m_1(x)}\phi_1(s_n) + (1 - \lambda) e^{i s_n m_2(x)}\phi_2(s_n)),$ we denote $G_n^{lc}(x)=\cos(s_n m_l(x))f_X(x)$ and $ G_n^{ls}(x)=\sin(s_n m_l(x))f_X(x)$ for $l=1,2.$ Then we have
\begin{align*}
\mathbb{E}(\hat{num}^j(s_n))  = &\mathbb{E} \left(\frac{1}{n b_n^k} \sum_{p=1}^n e^{i s_n Z_p} K(\frac{X_p - x_j}{b_n}) \right)
= \int_{U \in \R^k} g_n(x_j + b_n U) K(U) \mathrm{d}U,\\
= & \lambda \phi_1(s_n) \int [\cos(s_n m_1(x_j + b_n U)) + i \sin(s_n m_1(x_j + b_n U))] f_X(x_j + b_n U) K(U) \mathrm{d}U \\
& + (1 - \lambda) \phi_2(s_n) \int [\cos(s_n m_2(x_j + b_n U)) + i \sin(s_n m_2(x_j + b_n U))] f_X(x_j + b_n U) K(U) \mathrm{d}U \\
= & \lambda \phi_1(s_n) \int [G_n^{1c} (x_j + b_n U) + i G_n^{1s} (x_j + b_n U) ] K(U) \mathrm{d}U \\
& + (1 - \lambda) \phi_2(s_n) \int [G_n^{2c} (x_j + b_n U) + i G_n^{2s} (x_j + b_n U) ] K(U) \mathrm{d}U
\end{align*}

Using the assumption that the kernel is of order 2 (Assumption \ref{ass-kernel}),
\begin{equation}
\int G_n^{1c} (x_j + b_n U) K(U) \mathrm{d}U - G_n^{1c}(x_j), \nonumber \\
= \int \frac{b_n^2}{2} U' \nabla^2 G_n^{1c} [x_j + b_n \tau_n(U) U] \, U K(U) \, \mathrm{d}U, \label{estim4}
\end{equation}
where $\tau_n(U) \in [0;1]$ and  $\nabla^2 G_n^{1c}(x)$ is the hessian matrix of the function $G_n^{1c}$ evaluated at $x$. That is,
\begin{align*}
\nabla^2 G_n^{1c}(x) = & - s_n^2 f_X(x) \cos(s_n m_1(x)) \nabla m_1(x) \nabla m_1(x)' \\
& - s_n \sin(s_n m_1(x)) [\nabla m_1(x) \nabla f_X(x)' + \nabla f_X(x) \nabla m_1(x)' + f_X(x) \nabla^2 m_1(x)] \\
& + \cos(s_n m_1(x)) \nabla^2 f_X(x).
\end{align*}
Similarly to what is done in the first part of this proof, Assumption \ref{ass-convergence} guarantees that
$
\int G_n^{1c} (x_j + b_n U) K(U) \mathrm{d}U - G_n^{1c}(x_j) = O(b_n s_n)^2.
$
The same rate applies for $G_n^{1s}, G_n^{2c}$ and $G_n^{2s}$, implying 
$$
\mathbb{E}(\hat{num}^j(s_n))  = \int g_n(x_j + b_n U) K(U) \mathrm{d}U = g_n(x_j) + O( \, (b_n s_n)^2 \, [\lambda |\phi_1(s_n)| + (1 - \lambda) |\phi_2(s_n)| \, ]\, ),
$$
which gives, for the bias term,
\begin{align*}
\mathbb{E}(A_n) & = \frac{1}{ f_X(x_j) \phi^j(s_n)} \mathbb{E}( \hat{num}^j(s_n)) \\
& = 1 + O \left( (b_n s_n)^2 \frac{1}{f_X(x_j)} \frac{\lambda |\phi_1(s_n)| + (1 - \lambda) |\phi_2(s_n)|}{\lambda e^{i s_n m_1(x_j)}\phi_1(s_n) + (1 - \lambda) e^{i s_n m_2(x_j)}\phi_2(s_n)} \right)\\
& = 1 + O \left( (b_n s_n)^2 \frac{1}{ e^{i s_n m_2(x_j)}} \frac{\lambda \frac{|\phi_1(s_n)|}{|\phi_2(s_n)|} + (1 - \lambda)}{\lambda  \frac{\phi_1(s_n)}{|\phi_2(s_n)|} e^{i s_n (m_1(x_j) - m_2(x_j) )} + (1 - \lambda)} \right) \\
& = 1 + O(b_n s_n)^2,
\end{align*}
where the last equality comes from Assumption \ref{ass-ratio rates} (\ref{cf}).
As for the variance term, we write
\begin{align*}
Var(\frac{\hat{num}^j(s_n)}{ f_X(x_j) \phi^j(s_n)})& = \frac{1}{f_X(x_j)^2 |\phi^j(s_n)|^2} \frac{1}{n b_n^{2k}} Var(e^{i s_n Z} K(\frac{X - x_j}{b_n}))\\
& \leq \frac{1}{f_X(x_j)^2 |\phi^j(s_n)|^2}  \frac{1}{n b_n^{2k}} \mathbb{E}(|e^{i s_n Z} K(\frac{X - x_j}{b_n})|^2) \\
& \leq \frac{1}{ |\phi^j(s_n)|^2}  \frac{1}{n b_n^{k}} \frac{\int f_X(x_j + b_n U) K^2(U) \mathrm{d}U}{f_X(x_j)^2},
\end{align*}
and Assumption \ref{ass-convergence} (\ref{c-bandwidth}) guarantees that in the last equality, the third term in the product converges to $\frac{\int K^2(U) \mathrm{d}U}{f_X(x_j)}.$ Moreover,
\begin{align*}
|\phi^j(s_n)| & =|\lambda e^{i s_n m_1(x_j)}\phi_1(s_n) + (1 - \lambda) e^{i s_n m_2(x_j)}\phi_2(s_n)|\\
& = | \phi_2(s_n)| \ \left| \lambda e^{i s_n m_1(x)} \frac{\phi_1(s_n)}{\phi_2(s_n)} + (1 - \lambda) e^{i s_n m_2(x)} \right| \sim_{n \to \infty} (1- \lambda)|\phi_2(s_n)|,
\end{align*}
therefore implying
$ Var(A_n) = O(\frac{1}{n b_n^{k} |\phi_2(s_n)|^2 }).$

From those two computations, the following reasoning gives a convergence rate for $A_n$ :
$Bias(\Re(A_n)) = \Re(Bias(A_n)) = O(b_n s_n)^2.$ Similarly $Bias(\Im(A_n))=  O(b_n s_n)^2$. Plus, by definition for a complex random variable $Var(A_n)=Var(\Re(A_n)) + Var(\Im(A_n))$: both the variances of the real part and the imaginary part are smaller than the variance of $A_n$. 
An upper bound of the rates of convergence of the Mean Square Error of the real and imaginary parts is therefore obtained,
$\Re(A_n) - 1 = O_{\mathbb{P}} ((b_n s_n)^4 + \frac{1}{n b_n^{k} |\phi^j(s_n)|^2 } )^{1/2},$ and $\Im(A_n) = O_{\mathbb{P}} ((b_n s_n)^4 + \frac{1}{n b_n^{k} |\phi^j(s_n)|^2 } )^{1/2}$. 
This gives,
$$|A_n -1| = O_{\mathbb{P}} ((b_n s_n)^4 + \frac{1}{n b_n^{k} |\phi^j(s_n)|^2 } )^{1/2},$$
so that the estimation error is
\begin{align*}
\frac{\hat{\phi}^j(s_n)}{\phi^j(s_n)}&=1+ O_{\mathbb{P}} \left[ \left( b_n^4 + \frac{1}{n b_n^k} \right) ^{\frac{1}{2}} + \left((b_n s_n)^4 + \frac{1}{n b_n^{k} |\phi^j(s_n)|^2 } \right)^{1/2} \right] \\
& = 1+O_{\mathbb{P}} \left((b_n s_n)^4 + \frac{1}{n b_n^{k} |\phi^j(s_n)|^2 } \right)^{1/2}.
\end{align*}
Finally we obtain
\begin{align*}
\frac{\hat{\phi}^1(s_n)}{\hat{\phi}^0(s_n)} & = \frac{\hat{\phi}^1(s_n)}{\phi^1(s_n)}  \left( \frac{\hat{\phi}^0(s_n)}{\phi^0(s_n)} \right)^{-1} \frac{\phi^1(s_n)}{\phi^0(s_n)}\\
& =  e^{i s_n \nabla} [1 + O_{\mathbb{P}}(f(s_n))] \ \left[ 1 + O_{\mathbb{P}} \left((b_n s_n)^4 + \frac{1}{n b_n^{k} |\phi^j(s_n)|^2 } \right)^{1/2} \right].
\end{align*}  
Plugging in this expression in the definition of the estimator, we obtain
\begin{align*}
\hat{\nabla}  = & \frac{-i}{a_n} \mbox{ Log} \left( \frac{\hat{\phi}^1(s_n+a_n)}{\hat{\phi}^0(s_n+a_n)}  \left( \frac{\hat{\phi}^0(s_n)}{\hat{\phi}^1(s_n)} \right)  \right) \\
 = & \frac{-i}{a_n} \mbox{ Log} \mathrm{\{} e^{i (s_n +a_n) \nabla} [1 + O_{\mathbb{P}}(f(s_n + a_n))] [1 + O_{\mathbb{P}} \left((b_n (s_n + a_n))^4 + \frac{1}{n b_n^{k} |\phi^j(s_n + a_n)|^2 } \right)^{1/2} ] \\
 & e^{-i s_n \nabla} [1 + O_{\mathbb{P}}(f(s_n))] [1 + O_{\mathbb{P}} \left((b_n s_n))^4 + \frac{1}{n b_n^{k} |\phi^j(s_n)|^2 } \right)^{1/2} ] \mathrm{\}}\\
 = &  \frac{-i}{a_n} \mbox{ Log}  \mathrm{\{} e^{i a_n \nabla} [ 1 + O_{\mathbb{P}}(f(s_n+a_n))  + O_{\mathbb{P}}(f(s_n)) + O_{\mathbb{P}} \left((b_n s_n)^4 + \frac{1}{n b_n^{k} |\phi^j(s_n + a_n)|^2 } \right)^{1/2} \\
 & + O_{\mathbb{P}} \left((b_n s_n)^4 + \frac{1}{n b_n^{k} |\phi^j(s_n)|^2 } \right)^{1/2}  ] \mathrm{\}}.
\end{align*}
As the term multiplying $e^{i a_n \nabla}$ in the $\mbox{Log}$ converges to $1$, and eventually $a_n \nabla \in (-\pi; \pi),$ the expression above becomes
\begin{align*}
\hat{\nabla}  = \frac{-i}{a_n} \mathrm{\{} i a_n \nabla + \mbox{ Log} [ 1 + & O_{\mathbb{P}}(f(s_n+a_n))  + O_{\mathbb{P}}(f(s_n)) \\
+ & O_{\mathbb{P}} \left((b_n s_n)^4 + \frac{1}{n b_n^{k} |\phi^j(s_n + a_n)|^2 } \right)^{1/2}
 & + O_{\mathbb{P}} \left((b_n s_n)^4 + \frac{1}{n b_n^{k} |\phi^j(s_n)|^2 } \right)^{1/2}  ] \mathrm{\}},
 \end{align*}
that is, using the first order approximation of the principal value of the log around 1,
$$
\hat{\nabla} = \nabla +\frac{1}{a_n} O_{\mathbb{P}} \left(  f(s_n+a_n) +f(s_n) +  \left((b_n s_n)^4 + \frac{1}{n b_n^{k} |\phi^j(s_n + a_n)|^2 } \right)^{1/2} + \left((b_n s_n)^4 + \frac{1}{n b_n^{k} |\phi^j(s_n)|^2 } \right)^{1/2} \right ).
$$
\end{proof}

The only restriction imposed on the tuning parameter $a_n$ is that it converges to $0$.

For the sake of simplicity, we now write $\hat{\Delta} - \Delta = O_{\mathbb{P}} (\alpha_n)$ and $\hat{\nabla} - \nabla = O_{\mathbb{P}} (\beta_n )$. The rates $\alpha_n$ and $\beta_n$ depend on the distributions of the error terms, and we show here that the rates are polynomial in $n$ if these distributions are normal. 
 
Indeed if $\epsilon_1|x \sim \mathcal{N}(0, \sigma_1^2)$ and $\epsilon_2|x \sim \mathcal{N}(0, \sigma_2^2),$ with $\delta = \sigma_1^2 - \sigma_2^2 >0 ,$ then Assumption \ref{ass-ratio rates} is satisfied.
For the ratio of the mgf, $\forall \epsilon >0 \mbox{, } e^{\epsilon t} \frac{M_2(t)}{M_1(t)}  = e^{\epsilon t - \frac{\delta}{2} t^2} \underset{t \rightarrow \infty }{=} O(\mu(t)),$ with $\mu(t) = e^{-(\frac{\delta}{2}-\nu) t^2} \to 0,$ as $t \rightarrow \infty,$ for some $0 < \nu < \frac{\delta}{2}.$ And as for the ratio of the characteristic functions, $\frac{\phi_1(s)}{\phi_2(s)} = e^{- \frac{1}{2} \delta s^2} \underset{s \rightarrow \infty }{=} O(f(s)),$ with $f(s) = e^{- \frac{1}{2} \delta s^2} \xrightarrow [s \rightarrow \infty] {} 0.$
We take a fixed $a$ in the definition of $\hat{\nabla}$ here to simplify the computations, assuming $a$ is small enough. Applying the results from the estimation proofs, the convergence rates are 
\begin{enumerate}[\textup{(}i\textup{)}]
\item $\hat{\Delta} - \Delta = \frac{1}{t_n} O_{\mathbb{P}} \left[e^{-(\frac{\delta}{2}-\nu) t_n^2} + \left( (t_n h_n)^4 + \frac{1}{n h_n^k} e^{\sigma_1^2 t_n^2} \right)^{\frac{1}{2}} \right],$
\item $\hat{\nabla} - \nabla = O_{\mathbb{P}} \left[ e^{- \frac{1}{2} \delta s_n^2} +  \left((b_n s_n)^4 + \frac{1}{n b_n^{k} } e^{\sigma_2^2(s_n+a)^2} \right)^{\frac{1}{2}} \right].$
\end{enumerate}
One can show that with the appropriate choice of the sequences $t_h, h_n, s_n,$ and $b_n$, the rates are polynomial in $n$. For example, it is the case if $k=1$, $h_n= n^{\frac{-1}{5} + \epsilon}$, $t_n=\frac{1}{\sigma_1} (\epsilon \log(n))^{\frac{1}{2}}$, $s_n=\frac{1}{\sigma_2} ( \beta \log(n))^{\frac{1}{2}}$ and $b_n = n^{\frac{-1}{5} + \beta}$ for $\epsilon, \beta < \frac{1}{5}.$


\medskip

\subsection*{Proof 2} 
 
We now focus on the estimation of the remaining objects. We showed that if $\lambda \in (0,1),$ then
$\lambda = \frac{\mathbb{E}(Z|X=x_1)- \mathbb{E}(Z|X=x_0) - \nabla}{\Delta - \nabla}.$ A natural estimator is therefore $$\hat{\lambda}=\frac{\hat{E}(Z|X=x_1)- \hat{E}(Z|X=x_0) - \hat{\nabla}}{\hat{\Delta} - \hat{\nabla}},$$ where $\hat{E}(Z|X=.)$ is the usual multivariate kernel regression estimator,
$ \hat{E}(Z|X=x) = \frac{\sum_{p=1}^n Z_p K(\frac{X_p - x}{d_n})} {\sum_{p=1}^n  K(\frac{X_p - x_j}{d_n})}$ where the kernel does not have to be the one used for the previous estimators but will be written $K$ for the sake of simplicity. Similarly the point estimation of the regression functions $m_1$ and $m_2$ is derived from Equation (\ref{mID}). Writing $\hat{C}=\frac{1}{2} \left (\hat{E}(Z^2|X=x_0)  -  \hat{E}(Z^2|X=x_1) + \hat{\lambda} \hat{\Delta}^2 + (1 - \hat{\lambda}) \hat{\nabla}^2 \right),$ our estimators of $m_1(x_0)$ and $m_2(x_0)$ are 

\begin{equation}
\begin{bmatrix}
\hat{m}_1(x_0) \\ \hat{m}_2(x_0)
\end{bmatrix}
=
\begin{pmatrix}
\hat{\lambda}^{-1} &  0
\\
0 & (1 - \hat{\lambda})^{-1} 
\end{pmatrix}
\begin{pmatrix}
 -\hat{\Delta} &  -\hat{\nabla}
\\
1 & 1 
\end{pmatrix}
^{-1}
\begin{bmatrix}
\hat{C} \\ \hat{E}(Z|X=x_0)
\end{bmatrix}.
\end{equation}

The convergence rate of these estimators can be computed easily. With the usual assumptions for kernel estimation  and the appropriate choice of bandwidths $d_n=n^{\frac{-1}{k+4}},$ it is known that $\hat{E}(Z|X=x) - \mathbb{E}(Z|X=x) = O_{\mathbb{P}} (n^{-\frac{2}{k+4}})$ and $\hat{E}(Z^2|X=x) - \mathbb{E}(Z^2|X=x) = O_{\mathbb{P}} (n^{-\frac{2}{k+4}})$, see, e.g, \citeasnoun{hardle1994applied}.
Writing $\epsilon_n = n^{-\frac{2}{k+4}} + \alpha_n + \beta_n$, one obtains $\hat{\lambda} = \lambda + O_{\mathbb{P}} (\epsilon_n).$ The estimators of $m_1(x_0)$ and $m_2(x_0)$, being linearizable functions of $\hat{\lambda} $, $\hat{\Delta}$ and $\hat{\nabla}$, their rates of convergence are similarly bounded. This is summarized in the next proposition.

\begin{prop}\label{mrates}
Under Assumptions  \ref{ass-convergence}, \ref{ass-ratio rates}, assuming that $\lambda \in (0,1),$ $K$ satisfies Assumption \ref{ass-kernel}, $d_n \to 0,$ and $n d_n^k \to 0$,
\begin{enumerate}
\item $\hat{\lambda} = \lambda + O_{\mathbb{P}} (\epsilon_n),$
\item $\hat{m_i(x_0)} = m_i(x_0) + O_{\mathbb{P}} (\epsilon_n), \mathrm{ for} \: i=1,2$
\end{enumerate}

\end{prop}

\medskip

\subsection*{Proof 3}

To estimate the CDF of $\epsilon_1$ and $\epsilon_2$, we use Equation (\ref{eq-idF})
and propose the following estimator
\begin{equation*}
\hat{F}_2(z) = 1 - \frac{1}{1-\hat{\lambda}} \sum_{j=0}^{p(n)} \hat{F}(z + j \hat{\delta}(x_1,x_0) + \hat{m}_1(x_1) - \hat{g}(x_0) |x_1)  -  \hat{F}(z + j \hat{\delta}(x_1,x_0) + \hat{m}_2(x_0) |x_0).
\end{equation*}
In this formula $p(n) \in \mathbb{N}$ will be specified later, $\hat{g}(x)=\hat{m}_1(x) - \hat{m}_2(x),$ $\hat{\delta} = \hat{\Delta} - \hat{\nabla}$, and $\hat{F}(.|.)$ is the kernel regression estimator of the conditional cumulative distribution function,
$$\hat{F}(z|x)= \frac{\sum_{j=1}^n \mathbbm{1} (Z_j \leq z) k(\frac{X_j - x}{c_n})} {\sum_{j=1}^n  k(\frac{X_j - x}{c_n})}.$$
The kernel function and the bandwidth may differ from the choices for our previous kernel regression estimators, and will be here written as $k$ and $c_n$ respectively.

\begin{ass}\label{ass-Frate}
\
\begin{enumerate}
\item The probability distribution functions of $\epsilon_1$ and $\epsilon_2$, $f_1$ and $f_2,$ are bounded by a constant $c,$
\item\label{twicediff}  $f_j$ is twice differentiable on $\R,$ and $f_j,f_j', f_j''$ are continuous and bounded for $j=1,2.$
\end{enumerate}
\end{ass}

\begin{ass}\label{ass-Funicv} We assume that $k(.)$ satisfies Assumption \ref{ass-kernel}, and in addition impose,
\begin{enumerate}
\item $||k||_{\infty} < \infty,$
\item The kernel function $k$ has support contained in $[-\frac{1}{2}, \frac{1}{2}]^k,$
\item Assumptions (K-iii) and (K-iv) of \citeasnoun{einmahl2005uniform} hold for $k(.)$,
\item $c_n \geq C' \frac{\log(n)}{n},$ $c_n = O(n^{-\gamma_1}),$ for some $\gamma_1 < 1.$
\end{enumerate}
\end{ass}

The assumptions from \citeasnoun{einmahl2005uniform} are conditions on the covering and measurability properties of the class of functions $\left\lbrace k(\frac{x-.}{c}) ;\ c>0, \, x \in \R^k \right\rbrace$.

\begin{prop}\label{Frate}
Under  \ref{ass-convergence}, \ref{ass-ratio rates}, \ref{ass-Frate} and \ref{ass-Funicv},

$\hat{F}_1(z) - F_1(z)= O_{\mathbb{P}}\left( (p(n)+1)n^{ \frac{-2}{k+4} + a}+ p(n)^2 \epsilon_n + e^{-\gamma_0 p(n)} \right),$ and

$\hat{F}_2(z) - F_2(z)= O_{\mathbb{P}}\left( (p(n)+1)n^{ \frac{-2}{k+4} + a}+ p(n)^2 \epsilon_n + e^{-\gamma_0 p(n)} \right).$
\end{prop}

\begin{proof}
Fix $z \in \R.$
Write $\xi_j^0=z + j \delta(x_1,x_0) + m_2(x_0),$ $\hat{\xi}_j^0 = z + j \hat{\delta}(x_1,x_0) + \hat{m}_2(x_0)$, $\xi_j^1=z + j \delta(x_1,x_0) + m_1(x_1) - g(x_0)$ and $\hat{\xi}_j^1=z + j \hat{\delta}(x_1,x_0) + \hat{m}_1(x_1) - \hat{g}(x_0)$.

\begin{align}\label{Fsplit}
\hat{F}_2(z) - F_2(z) = & \frac{1}{1-\lambda} \sum_{j=0}^{p(n)} \left\lbrace \left[\hat{F}(\hat{\xi}_j^1|x_1) - F(\xi_j^1|x_1) \right] - \left[\hat{F}(\hat{\xi}_j^0|x_0) - F(\xi_j^0|x_0) \right] \right\rbrace  \nonumber \\
& - \frac{1}{1-\lambda} \sum_{j=p(n)+1}^\infty \left[ F(\xi_j^1|x_1) - F(\xi_j^0|x_0) \right] \\
& + (\frac{1}{1-\hat{\lambda}} - \frac{1}{1-\lambda}) \sum_{j=0}^{p(n)} \hat{F}(\hat{\xi}_j^1|x_1) - \hat{F}(\hat{\xi}_j^0|x_0). \nonumber
\end{align}
We write $\hat{F}_2(z) - F_2(z) = I_1 - I_2 + I_3,$ and the convergence rate of each part in the right hand side of (\ref{Fsplit}) will be computed separately.

For $I_1$, we write 
\begin{align*}
(1 - \lambda) \, I_1= & \sum_{j=0}^{p(n)}  \left[\hat{F}(\hat{\xi}_j^1|x_1) - F(\hat{\xi}_j^1|x_1) \right] + \sum_{j=0}^{p(n)} \left[ F(\hat{\xi}_j^1|x_1) - F(\xi_j^1|x_1) \right] \\
& - \sum_{j=0}^{p(n)}  \left[\hat{F}(\hat{\xi}_j^0|x_0) - F(\hat{\xi}_j^0|x_0) \right] - \sum_{j=0}^{p(n)} \left[ F(\hat{\xi}_j^0|x_0) - F(\xi_j^0|x_0) \right] .
\end{align*}

We know $\frac{\partial F(z|x)}{\partial z} = f(z|x) = \lambda f_1(z-m_1(x)) + (1-\lambda)f_2(z-m_2(x)),$ and Assumption \ref{ass-Frate} guarantees that $f(y|x)$ is bounded by $c, \, \forall \, (x,y)$  therefore $ y \mapsto F(y|x)$ is Lipschitz continuous with constant $c.$ That is, for $v=0,1,$
$|F(\hat{\xi}_j^v|x_v) - F(\xi_j^v|x_v)| \leq c \, |\hat{\xi}_j^v - \xi_j^v|$
implying
\begin{align*}
|\sum_{j=0}^{p(n)} F(\hat{\xi}_j^v|x_v) - F(\xi_j^v|x_v) | = O_{\mathbb{P}}(p(n)^2 \epsilon_n).
\end{align*}
For the two other terms  in $I_1,$
we write 
$$F_n(.|x_i) = \frac{\mathbb{E}(\mathbbm{1} (Z \leq z) k(\frac{X - x_i}{c_n}))}{\mathbb{E}( k(\frac{X - x_i}{c_n}))}.$$
Then under Assumption \ref{ass-Funicv}, we apply Theorem 3 of \citeasnoun{einmahl2005uniform}, which gives the rate of the supremum of $|| \hat{F}(.|x) - F_n(.|x) ||_{\infty}$ over a certain range of bandwidths and over $x \in I$ where $I$ is a compact subset of $\R^k.$ For the specific bandwidth $b_n$ and taking $I = \left\lbrace x_0,x_1 \right\rbrace$ we then have 
\begin{align}\label{sup1}
\limsup_{n \to \infty}  (nc_n^k)^{1/2} || \hat{F}(.|x) - F_n(.|x) ||_{\infty} & = O_{\mathrm{a.s}}\left(\max(\log \log n,-\log(c_n))^{1/2}\right) \nonumber \\
& = O_{\mathrm{a.s}}\left( (-\log c_n)^{1/2}\right).
\end{align}

We now examine $ F_n(.|x) - F(.|x).$ Write $f(.,.)$ the joint density of $(Z,X)$, then we define
\begin{align*}
F_X(x,z)& = \int_{ z' \leq z} f(z',x)\mathrm{d}z' = F(z|x) f_X(x) = [ \lambda F_1(z - m_1(x)) + (1 - \lambda) F_2(z - m_2(x)) ] f_X(x)
\end{align*}
and  write
$$ F_n(z|x_i) - F(z|x_i)= \frac{\frac{1}{c_n^k} \mathbb{E}(\mathbbm{1} (Z \leq z) k(\frac{X - x_i}{c_n})) - F_X(x_i,z)}{\frac{1}{c_n^k} \mathbb{E}( k(\frac{X - x_i}{c_n})) } + F_X(x_i,z) \left( \frac{1}{\frac{1}{c_n^k} \mathbb{E}( k(\frac{X - x_i}{c_n}))} - \frac{1}{f_X(x_i)}\right).$$
Under Assumption \ref{ass-convergence}, \ref{ass-Frate} and \ref{ass-Funicv}, we know that $\frac{1}{c_n^k} \mathbb{E}( k(\frac{X - x_i}{c_n})) - f_X(x_i) = O(c_n^2).$ 
Similarly, 
$$\frac{1}{c_n^k} \mathbb{E}\left[ \mathbbm{1} (Z \leq z) \, k\left( \frac{X - x_i}{c_n}\right) \right]  - F_X(x_i,z)= \frac{c_n^2}{2} \int_{U \in \R^k}  U' \, \nabla_X^2  F_X (x_i + b_n \tau_n(U) U, z) \, U k(U) \mathrm{d}U ,$$
and Assumption \ref{ass-convergence}, \ref{ass-Frate} and \ref{ass-Funicv} guarantee that $\nabla_X^2  F_X (.,.)$ is uniformly bounded over $\R^{k+1}$. Therefore 
$$\sup_{z \in \R} \left| \ \frac{1}{c_n^k}   \mathbb{E}\left[ \mathbbm{1} (Z \leq z) \, k\left( \frac{X - x_i}{c_n}\right) \right]  - F_X(x_i,z)\right| = O(c_n^2).$$ 
which gives, for $i=0,1,$
\begin{equation}\label{sup2}
|| F_n(.|x_i) - F(.|x_i)||_{\infty} = O(c_n^2).
\end{equation}

Equations (\ref{sup1}) and (\ref{sup2}) give
$||\hat{F}(.|x) - F(.|x) ||_{\infty} = O_{\mathbb{P}}((-\log(c_n))^{1/2} (nc_n^k)^{-1/2} + c_n^2).$ 
For the appropriate choice of $\gamma_1$ in Assumption \ref{ass-Funicv} and for any small $a>0,$
$$\sup_{y \in \R} | \hat{F}(y|x_i) - F(y|x_i) | = O_{\mathbb{P}}( n^{-\frac{2}{k+4} + a}), \; i=0,1.$$
This implies that
$$ \sum_{j=0}^{p(n)}  \left[\hat{F}(\hat{\xi}_j^1|x_i) - F(\hat{\xi}_j^1|x_i) \right]= O_{\mathbb{P}}((p(n)+1)n^{-\frac{2}{k+4} + a}), \; i=0,1. $$
Therefore,
$$I_1 = O_{\mathbb{P}} \left( (p(n)+1)n^{-\frac{2}{k+4} + a} + p(n)^2 \epsilon_n \right).$$

Looking at $I_2,$ by construction the second sum appearing in the right hand side of (\ref{Fsplit}) simplifies to 
$$\frac{1}{1-\lambda} \sum_{j=p(n)+1}^\infty \left[ F(\xi_j^1|x_1) - F(\xi_j^0|x_0) \right] = 1 - F_2(z + (p(n)+1) \delta(x_1,x_0)).$$

Using the exponential version of the Chebyshev's inequality, we have 
$1- F_2(C)= \mathbb{P}(\epsilon_2 >C) \leq e^{-tC}M_2(t)$ using the assumption that the moment generating functions are finite.
Fixing $t_0 \in \R_{+},$ $1 - F_2[z + (p(n)+1) \delta(x_1,x_0)] \leq e^{t_0 z + (p(n)+1) \delta(x_1,x_0) t_0}$ which guarantees the existence of $\gamma_0 > 0$ such that
$I_2 = O(e^{-\gamma_0 \, p(n)}).$

As for $I_3,$ we showed in our computation for $I_1$ that
$\sum_{j=0}^{p(n)} \hat{F}(\hat{\xi}_j^1|x_1) - \hat{F}(\hat{\xi}_j^0|x_0) \xrightarrow [n \to \infty]{} F_2(z).$ 
As
$\frac{1}{1-\hat{\lambda}} - \frac{1}{1-\lambda}=  O_{\mathbb{P}}(\epsilon_n),$ we have
$$I_3 = O_{\mathbb{P}} (\epsilon_n).$$ 
Adding these three parts, we obtain
$$\hat{F}_2(z) - F_2(z) = O_{\mathbb{P}} \left( (p(n)+1)n^{-\frac{2}{k+4} + a}+ p(n)^2 \epsilon_n + e^{-\gamma_0 p(n)} \right). $$

Using the equation $F(z|x)= \lambda F_1(z - m_1(x)) + (1 - \lambda) F_2(z - m_2(x)),$ an estimator of $F_1(z)$ is
$$ \hat{F}_1(z) = \frac{1}{\hat{\lambda}} \left[ \hat{F}(z + \hat{m}_1(x)) - (1 - \hat{\lambda}) \hat{F}_2(z + \hat{m}_1(x) - \hat{m}_2(x)) \right], $$ 
which will converge to $F_1(z)$ at the same rate.

\end{proof}

\textit{
In the case where $\epsilon_n$ is slower than $n^{-2 \: \frac{1 - 2a}{k+4}},$ for some $a,$ which happens when for instance the error terms are normally distributed, then $p_n$ is solution to $\epsilon_n p_n = t_0 e^{-t_0 p_n}.$}

\section{Conclusion}\label{sec:conclusion}
New nonparametric identification results for finite mixture models are
developed.  These open up the possibility of flexibly modeling
economic behavior in the presence of unobserved heterogeneity.

\section{Appendix}

\label{sec:appendix}

\parindent=0em

This Appendix presents the proofs of some of the results presented
in the previous sections.

\begin{proof}[\textupandbold{Proof of Lemma~\ref{lem:thmFE}}]
	
	Define $\delta(x) : = m_2(x) - m_1(x)$, $\dot m_1(x) := m_1(x) - m_1(x_0)$,    $\dot m_2(x) := m_2(x) - m_2(x_0)$,   
	$$
	r(+\infty,x) := \lim_{t \rightarrow + \infty} \frac 1 t \log R(x,t), \quad r(-\infty,x) := \lim_{t \rightarrow - \infty} \frac 1 t \log R(x,t)
	$$	
	and
	$$
	\tilde \lambda_c(x) := \frac{1 - K_{-\infty}(x)+ c}{K_{+\infty}(x) - K_{-\infty}(x) + c}.  
	$$
	In what follows we show that the slopes of $m_1$ and $m_2$ over the interval connecting $x$ and $x_0$, as well as the values of  $\lambda(\cdot)$  at these two points, are all recovered from $r(+\infty,x)$, 	$r(-\infty,x)$ and $\lim_{c \downarrow 0} \tilde \lambda_c(x)$.

	\
	
	\noindent {\it Case (1): $\lambda(x)  = \lambda(x_0) = 1$.}  
	
	With the given structure of the model we have $m_1(x) = m_2(x)$,  $m_1(x_0) = m_2(x_0)$, and $M_1 \equiv M_2$ in this case.  Thus 
	$$
	R(x,t) = \frac{e^{tm_1(x)}}{e^{tm_1(x_0)}} = e^{t \dot m_1(x)}
	$$ 				
	and 
	$$
	\frac 1 t \log R(x,t) = \dot m_1(x), 
	$$
	therefore Condition \ref{cond:FE}\eqref{ass:thm4id1} fails.  On the other hand this means 
	$$
	K_{+\infty}(x) = K_{-\infty}(x) = 1,
	$$ 	
	yielding 
	$$
	\tilde \lambda_c(x) = \frac{1 - 1 + c}{1 - 1 +c} = 1,
	$$
	therefore Condition \ref{cond:FE}\eqref{ass:thm4id2} holds in this case.  Moreover, the values of $\lambda$ are identifiable  from $\lim_{c \downarrow 0}\lambda_c(x)$.   
	
	\
	
	\noindent {\it Case (2): $\lambda(x) < 1, \lambda(x_0) < 1$.  Condition \ref{cond:FE}\eqref{ass:thm4id1} holds.}  
	
	In this case the two slopes $(m_1(x) - m_1(x_0),m_2(x) - m_2(x_0))$ are identified as in the proof of Lemma  \ref{lem:slope}. 
	
	Take $\delta'$ as in the proof of Lemma  \ref{lem:slope}.  We first consider the case with  $t$ tending to $+\infty$.  If $h(\pm\epsilon,t) = O(1)$ holds, then according to the proof of Lemma  \ref{lem:slope} we have 
	$$
	\lim_{t \rightarrow \infty} \frac 1 t \log R(x,t) = m_1(x) - m_1(x_0)
	$$
	for $x \in N^1(x_0,\delta')$
	and consequently  
	$$
	K_{+\infty}(x) = \frac{\lambda(x)}{\lambda(x_0)}.  
	$$	
	If $1/h(\pm\epsilon,t) = O(1)$ then 
	$$
	\lim_{t \rightarrow \infty} \frac 1 t \log R(x,t) = m_2(x) - m_2(x_0)
	$$
	and then 
	$$
	K_{+\infty}(x) =  \frac{1 - \lambda(x)}{1 - \lambda(x_0)}.  
	$$
	With these results we see $K_{+\infty}(x) \neq K_{-\infty}(x)$ iff $\lambda(x) \neq \lambda(x_0)$.  With 
	\begin{eqnarray*}
		\lim_{c \downarrow 0}\tilde \lambda_c(x) &=& \frac{1 - K_{-\infty}(x)}{K_{+\infty}(x) - K_{-\infty}(x) }\\
		&=& \lambda(x).  
	\end{eqnarray*}
	By continuity $\lambda(x_0)$ is identified as $\lim_{x \rightarrow x_0}\lambda(x)$.
	If   $K_{+\infty}(x) = K_{-\infty}(x)$  we can obtain the value of $\lambda(x)$ (and thus $\lambda(x_0)$) as $\lim_{c \downarrow 0}\lambda_c$, as noted in the proof of Lemma \ref{lem:thm1}.

	Now we let $t \rightarrow -\infty$.   If $h(\pm\epsilon,t) = O(1)$ and $1/h(\pm\epsilon,t) = O(1)$  as  $t \rightarrow -\infty$ we have 
	$$
	\lim_{t \rightarrow -\infty} \frac 1 t \log R(x,t) = m_1(x) - m_1(x_0), \quad K_{-\infty}(x) =  \frac{\lambda(x)}{\lambda(x_0)}  
	$$
	and 
	$$
	\lim_{t \rightarrow -\infty} \frac 1 t \log R(x,t) = m_2(x) - m_2(x_0), \quad  K_{-\infty}(x) =  \frac{1 - \lambda(x)}{1 - \lambda(x_0)}  
	$$
	respectively, so once again we identify $\lambda(x)$ and the two slopes by switching $\lambda(x)$ and $\lambda(x_0)$ and $m_1$ and $m_2$.  
	
	If  both $h(\pm\epsilon,t) = O(1)$ and $1/h(\pm\epsilon,t) = O(1)$ hold, $D(x_0)$ = 0.     If $D(x) >0$, for example,  then $r(x,-\infty) = \dot m_1(x)$ and $r(x,+\infty) = \dot m_2(x)$.   (In this case Condition \ref{cond:FE}\eqref{ass:thm4id1} is automatically satisfied.)   $\lambda(x)$  is identified, hence $\lambda(x_0)$ too,  as above.

	\

	\noindent {\it Case (3): $\lambda(x) < 1, \lambda(x_0) < 1$.  Condition \ref{cond:FE}\eqref{ass:thm4id1} fails.}
	
	Wlog suppose $r(x,+\infty) = \dot m_1(x)$, then 
	$$
	K_{+\infty,t}(x) = \frac{\lambda(x) + (1 - \lambda(x))  e^{t\delta(x)} \frac{M_2(t)}{M_1(t)} }{\lambda(x_0) + (1 - \lambda(x_0))  e^{t\delta(x_0)} \frac{M_2(t)}{M_1(t)} },
	$$    
	so for  Condition \ref{cond:FE}\eqref{ass:thm4id2} to hold we need
	$$
	{\lambda(x) + (1 - \lambda(x))  e^{t\delta(x)} \frac{M_2(t)}{M_1(t)} } = {\lambda(x_0) + (1 - \lambda(x_0)))  e^{t\delta(x_0)} \frac{M_2(t)}{M_1(t)} }
	$$
	or 
	$$
	\frac{\lambda(x) - \lambda(x_0)}{1 - \lambda(x_0)} = \left[ e^{t\delta(x_0)} + \frac{1 - \lambda(x) }{1 - \lambda(x_0)} e^{t\delta(x)}   \right] \frac{M_2(t)}{M_2(t)}.
	$$
	Take $x_1 \neq x_0$  in $N^1(x_0,\delta')$.      Since the right hand side of the above equation is positive, we have $\lambda(x_1)  \neq  \lambda(x_0)$.  Then 
	\begin{eqnarray*}
		\frac{  \frac{\lambda(x) - \lambda(x_0)}{1 - \lambda(x_0)}    }{ \frac{\lambda(x_1) - \lambda(x_0)}{1 - \lambda(x_0)} } &=& \frac{  1 + \frac{1 - \lambda(x) }{1 - \lambda(x_0)} e^{t[\delta(x) - \delta(x_0)]}  }
		{   
			1 + \frac{1 - \lambda(x_1) }{1 - \lambda(x_0)} e^{t[\delta(x_1) - \delta(x_0)]}
		}
		\\
		& = &
		\frac{  1 + \frac{1 - \lambda(x) }{1 - \lambda(x_0)} e^{t[  \dot m_2(x)- \dot m_1(x)   ]}  }
		{   
			1 + \frac{1 - \lambda(x_1) }{1 - \lambda(x_0)} e^{t[   \dot m_2(x_1)- \dot m_1(x_1)   ]}
		}.
	\end{eqnarray*}
	In view of the non-parallel  assumption, the right hand does not depend of $t$ only if $\lambda(x) =0$, which is a contradiction.  Thus Case (3) is (correctly) precluded by Condition \ref{cond:FE}.     
	
	\
	
	\noindent {\it Case (4): $\lambda(x) < 1, \lambda(x_0) = 1$.   Condition \ref{cond:FE}\eqref{ass:thm4id1} holds.}
	
	Note that $\lambda(x_0) = 1$ means $m_1(x_0) = m_2(x_0)$, and moreover, with Assumption \ref{ass:FE1}, $M_1$ and $M_2$ are identical.  Then 
	\begin{eqnarray*}
		R(x,t) &=& \frac{\lambda(x) e^{tm_1(x)}M_1(t) + (1 - \lambda(x) )   e^{t m_2(x) } M_1(t)  }{  e^{ tm_1(x_0) }  M_1(t) }
		\\
		&=& \lambda(x) e^{t \dot m_1} + ( 1 - \lambda(x)  e^{t \dot m_2(x) }. 
	\end{eqnarray*}
	If, for example, $\dot m_1(x) > \dot m_2(x)$, $r(x,+\infty) = \dot m_1(x)$ and  $r(x,-\infty) = \dot m_2(x)$, and moreover, 
	\begin{eqnarray*}
		K_{+\infty,t}(x) &=& R(x,t) e^{-t\dot m_1(x)}
		\\
		&=&  \lambda(x)  + ( 1 - \lambda(x)  e^{t [ \dot m_2(x) - \dot  m_1(x) ] }
		\\
		&\rightarrow& \lambda(x) \text{ as } t \rightarrow \infty, 
	\end{eqnarray*}
	that is, $\lambda(x) = K_{+\infty}(x)$.  Proceeding analogously,  we have  
	$
	\lambda(x) = K_{-\infty}(x).
	$
	Use these values in the definition of $\tilde \lambda$, we see that $\lambda(x)$ is identified from $\tilde \lambda(x)$.   Analysis of the case with $\dot m_1(x) < \dot m_2(x)$ is analogous.  And of course $\dot m_1(x) = \dot m_2(x)$ cannot happen.
	
	\
	
	\noindent {\it Case (5): $\lambda(x) < 1, \lambda(x_0) = 1$.   Condition \ref{cond:FE}\eqref{ass:thm4id1} fails.}
	
	As seen in Case (4), in this case we have 
	$$
	R(x,t) = \lambda(x) e^{t \dot m_1} + ( 1 - \lambda(x) )  e^{t \dot m_2(x)}, 
	$$
	and if, for example, $\dot m_1(x) > \dot m_2(x)$ 
	$$
	K_{+\infty,t}(x) = \lambda(x)  + ( 1 - \lambda(x) ) e^{t [ \dot m_2(x) - \dot  m_1(x) ] }.
	$$
	Thus Condition \ref{cond:FE}\eqref{ass:thm4id2} fails and this case is (correctly) precluded.   Analysis of the case with $\dot m_1(x) < \dot m_2(x)$ is analogous, and $\dot m_1(x) \neq \dot m_2(x)$ as above.
	
	\
	
	\noindent Finally, note that $\lambda(x_0) < 1$ then by continuity $\lambda(x) <1$ for every $x \in N^1(x_0,\delta')$ for sufficiently small $\delta'$, so this reduces to either Case (2) or (3).
	
	%
	
\end{proof}

\begin{proof}[\textupandbold{Proof of Proposition~\ref{prop:rational}}]
	The recursive formula in Lemma \ref{lem:qrep}
	then becomes  (NOTE THE USE OF x, not $x_a$)
	\begin{equation}\label{Prec}
	R_{k+1}^j = D_{x^1}\left(\frac{R_k^j}{R_k^k}\right)
	+ t\frac{R_k^j}{R_k^k}D_{x^1}m_{j,k}, j = 1,...,J
	\end{equation}
	with initial conditions
	\begin{equation}\label{Ini}
	R_{2}^j = t D_{x^1}m_{j1}, j = 1,...,J.
	\end{equation}
	For $k = 3$, 
	\begin{align*}
	R_{3}^j 
	&= D_{x^1}\left(\frac{R_2^j}{R_2^2}\right)
	+ t\frac{R_2^j}{R_2^2}D_{x^1}m_{j,2}
	\\
	&=  D_{x^1}\left(\frac{D_{x^1}m_{j1}}{D_{x^1}m_{21}}\right)
	+ t\frac{D_{x^1}m_{j1}}{D_{x^1}m_{21}}D_{x^1}m_{j,2}
	\\
	&=
	\frac{D_{x^1}^2m_{j,1}D_{x^1}m_{2,1} - D_{x^1}^2m_{2,1}D_{x^1}m_{j,1}
		+ t D_{x^1}m_{j,1}D_{x^1}m_{j,2}D_{x^1}m_{2,1}}{(D_{x^1}m_{2,1})^2}
	\\
	&= \frac{P_3^j}{(D_{x^1}m_{2,1})^2} 
	\end{align*}
	where 
	$$
	P_3^j = D_{x^1}^2m_{j,1}D_{x^1}m_{2,1} - D_{x^1}^2m_{2,1}D_{x^1}m_{j,1}
	+ t D_{x^1}m_{j,1}D_{x^1}m_{j,2}D_{x^1}m_{2,1}.
	$$
	Note that $P_3^j$ depends on $x$ (where $m$'s are evaluated) and $t$, so
	it can be interpreted as shorthand for $P_3^j(x,t)$.
	Then 
	\begin{align*}
	R_{4}^j 
	&= D_{x^1}\left(\frac{R_3^j}{R_3^3}\right)
	+ t\frac{R_3^j}{R_3^3}D_{x^1}m_{j,3}
	\\
	&= D_{x^1}\left(\frac{P_3^j}{P_3^3}\right)
	+ t\frac{P_3^j}{P_3^3}D_{x^1}m_{j,3}
	\\
	&= \frac {D_{x^1}P_{3}^jP_{3}^{3} - P_{3}^jD_{x^1}P_{3}^{3} + tP_{3}^jP_{3}^{3} D_{x^1}m_{j,3}}
	{(P_{3}^{3})^2} 
	\\
	&= \frac{P_4^j}{(P_{3}^{3})^2}
	\end{align*}
	where
	$$
	P_4^j = D_{x^1}P_{3}^jP_{3}^{3} - P_{3}^jD_{x^1}P_{3}^{3} + tP_{3}^jP_{3}^{3} D_{x^1}m_{j,3}.
	$$
	Note that $P_3^3 \neq 0$ at least for large $t$, therefore the above
	representation of $R_4^j$ is valid. From here we can argue by
	induction.  Suppose $P_{h-1}^{h-1} \neq 0$ (which will be justified
	shortly): also assume that for $k = h$, $R_h^j$ 
	can be written as
	\begin{equation}\label{hypo1}
	R_h^j =  \frac{P_h^j}{(P_{h-1}^{h-1})^2}, 
	\end{equation}
	where $P_h^j$ and $P_{h-1}^j$, $j = 1,...,J$ satisfy the following relationship
	\begin{equation}\label{hypo2}
	P_h^j = D_{x^1}P_{h-1}^jP_{h-1}^{h-1} - P_{h-1}^jD_{x^1}P_{h-1}^{h-1} + tP_{h-1}^jP_{h-1}^{h-1} D_{x^1}m_{j,h-1}.
	\end{equation}
	Then as in the case of $h = 4$ above, 
	\begin{align*}
	R_{h+1}^j 
	&= D_{x^1}\left(\frac{R_h^j}{R_h^h}\right)
	+ t\frac{R_h^j}{R_h^h}D_{x^1}m_{j,h}
	\\
	&= D_{x^1}\left(\frac{P_h^j}{P_h^h}\right)
	+ t\frac{P_h^j}{P_h^h}D_{x^1}m_{j,h}
	\\
	&= \frac {D_{x^1}P_{h}^jP_{h}^{h} -
		P_{h}^jD_{x^1}P_{h}^{h} + tP_{h}^jP_{h}^{h}
		D_{x^1}m_{j,h}} 
	{(P_{h}^{h})^2} 
	\\
	&= \frac{P_{h+1}^j}{(P_{h}^{h})^2}
	\end{align*}
	with 
	$$
	P_{h+1}^j = D_{x^1}P_{h}^jP_{h}^{h} -
	P_{h}^jD_{x^1}P_{h}^{h} + tP_{h}^jP_{h}^{h}
	D_{x^1}m_{j,h},
	$$
	i.e., if (\ref{hypo1}) and (\ref{hypo2}) hold for $k = h$, they also
	hold for $k = h+1$.  In short, the original system of equations
	(\ref{Prec}) and (\ref{Ini}) that determine $R_k^j$ can be rewritten in terms of $P_k^j$s as
	follows:
	\begin{align}\label{Rformula}
	P_1^j &= 1, \quad  
	P_{h+1}^j = D_{x^1}P_{h}^jP_{h}^{h} -
	P_{h}^jD_{x^1}P_{h}^{h} + tP_{h}^jP_{h}^{h}
	D_{x^1}m_{j,h},
	\\
	R_k^j &= \frac{P_{k}^j}{(P_{k-1}^{k-1})^2}, \qquad 1 \leq k,j \leq J.
	\nonumber
	\end{align}
	(The fact that $P_1^j = 1, j = 1,...J$ are appropriate initial
	conditions can be easily verified.) In particular, (\ref{Rformula})
	implies that 
	\begin{equation}\label{Precursive}
	P_{h+1}^{h+1} = D_{x^1}P_{h}^{h+1}P_{h}^{h} -
	P_{h}^{h+1}D_{x^1}P_{h}^{h} + tP_{h}^{h+1}P_{h}^{h}
	D_{x^1}m_{h+1,h}.
	\end{equation}
	Note that (\ref{Precursive}) with initial values $P_1^1 = P_1^2$ recursively
	generates expressions of $P_k(\cdot,\cdot) =
	P_k^j(\cdot,\cdot), k = 2,...,J, j = k,...,J$ that have some useful
	properties including
	
	\medskip
	
	\noindent (Replacement Property of $P_h^j$): $P_h^{j}, j = h+1,...,J$ are obtained by replacing
	$m_{h}$ in the expression for $P_h^h$ with $m_{j}, j = 1,...,J$.  
	
	\medskip
	
	To see this,
	first note that $P_2^j = t
	D_{x^1}m_{j,1}, j = 2,...,J$ according to (\ref{Rformula}), therefore this claim applies to the case of $k =
	2$.  But (\ref{Rformula}) also shows that if the
	claim applies to $k = h$, it holds for $k = h+1$ as well.  The
	property holds for all k by induction.
	
	Noting this property, it is easy to see that  $P_k(\cdot,\cdot) =
	P_k^k(\cdot,\cdot), k = 2,...,J$ are polynomials in $t$ where their
	coefficients are functions of derivatives of $m$'s.  First, it trivially holds for $k=2$ since  $P_2^j = t
	D_{x^1}m_{j,1}, j = 2,...,J$.  Now, 
	suppose the claim holds for $k = h$.  Then by (\ref{Precursive})
	$P_{h+1}^{h+1}$ is a polynomial with the stated property, and by the
	replacement property, so are $P_{h+1}^{j}, j = h+2,...,J$.  That is,
	the claim holds for $k = h+1$.  By induction, the claim holds for $k =
	2,...,J$.  In particular, 
	we now know that $P_k = P_k^k, k = 3,...,J$ are polynomials in $t$, as claimed
	in the Proposition.  
	
	It remains to verify the formulae for $\deg_t(P_k)$ and
	lc$_t(P_k)$ given in the Proposition.   
	Start with
	$k=3$.  It implies that 
	$$
	P_3^3 = D_{x^1}^2m_{3,1}D_{x^1}m_{2,1} - D_{x^1}^2m_{2,1}D_{x^1}m_{3,1}
	+ t D_{x^1}m_{3,1}D_{x^1}m_{3,2}D_{x^1}m_{2,1},
	$$
	therefore $\deg_tP_3 = 1$ and $\text{lc}_t(P_3) =
	D_{x^1}m_{3,1}D_{x^1}m_{3,2}D_{x^1}m_{2,1}$, which are certainly consistent with
	the proposition.  Now suppose the Proposition holds for $k=l$: $P_l$ is a polynomial with $\deg_t(P_l)
	= 2^{l-2} - 1$ and 
	$
	\mathrm{lc}_t(P_l(t,x)) =
	(\Pi_{g=1}^{l-1}D_{x^1}m_{l,g})\Pi_{j=2}^{l-1}\{(\Pi_{h=1}^{j-1}D_{x^1}m_{j,h})^{2^{l
			- j - 1}}\}. 
	$

	Since $x \in N^1(x_a,\delta')$, $(D_1m_l)_{l=1..J}$ take $J$ distinct values, and $\mathrm{lc}_t(P_l(t,x)) \neq
	0$.  Also, the above observation that $P_l^l$ and $P_l^{j}$ are identical except for the
	replacement of $m_l$ with $m_{j}$, for all $ j\geq l$ implies that 
	\begin{equation}
	\deg_t(P_l^l) = \deg_t(P_l^{j}), \, \forall j\geq l,
	\end{equation}
	and 
	\begin{align*}
	\mathrm{lc}_t(P_l^{l+1}(t,x)) 
	& =
	(\Pi_{g=1}^{l-1}D_{x^1}m_{l+1,g})\Pi_{j=2}^{l-1}\{(\Pi_{h=1}^{j-1}D_{x^1}m_{j,h})^{2^{l
			- j - 1}}\} 
	\\
	& \neq 0. 
	\end{align*}
	Using the recursion formula (\ref{Precursive}) with $h = l$ and
	noting that $\deg_t(D_{x^1}P_{h}^{h+1}P_{h}^{h} -
	P_{h}^{h+1}D_{x^1}P_{h}^{h}) \leq \deg_t(P_l^{l+1}P_l^l)$, we have
	\begin{equation}\label{degrecursive}
	\deg_t(P_{l+1}^{l+1}) = 2\deg_t(P_l^l) + 1
	\end{equation}
	and
	\begin{align*}
	\text{lc}_t(P_{l+1}^{l+1}) 
	&= \text{lc}_t(P_l^{l+1})\text{lc}_t(P_l^l)  D_{x^1}m_{l+1,l}
	\\
	&= (\Pi_{g=1}^{l}D_{x^1}m_{l,g})\Pi_{j=2}^{l}\{(\Pi_{h=1}^{j-1}D_{x^1}m_{j,h})^{2^{l
			- j}}\}.
	\end{align*}
	Moreover, solving the difference equation (\ref{degrecursive}) under the initial
	condition $\deg_t(P_3^3) = 1$,
	\begin{align*}
	\deg_t(P_k^k) &= \sum_{j=0}^{k - 4}2^j +  2^{k-3}
	\\
	&= 2^{k-3} - 1 + 2^{k-3}
	\\
	&= 2^{k-2} - 1.
	\end{align*} 
	
	Since $P_k = P_k^k, k = 1,...,J$ are polynomials, they are nonzero for
	sufficiently large $t$.  This justifies division by $P_k$ used
	throughout the current proof for sufficiently large $t$. 
\end{proof}

\begin{prop}\label{prop:zeros}  There exists 
	$X^{(J)}=(x_1^{(J)},...,x_{J-1}^{(J)}) \in B(x_0, \delta')^{J-1}$ such that 
	$$\mathcal{Z}= \left\lbrace t \in \R | \det D(t, x_0, x_1^{(J)},...,x_{J-1}^{(J)})=0 \right\rbrace
	$$ is a finite set.
\end{prop}

\begin{proof}[\textupandbold{Proof of Proposition~\ref{prop:zeros}}]
	$$D(t, c_1,...,c_J)=(e^{tm_j(c_i)})_{1 \leq i,j \leq J}$$
	Writing $S_n$ the set of permutations of the first $n$ natural numbers and $sign(\sigma)$ the signature of a permutation $\sigma$, we have $\det D(t, c_1,...,c_J) = \sum_{\sigma \in S_J} sign(\sigma) \; e^{t \sum_{i=1}^J m_{\sigma(i)}(c_i)}$.

	\noindent {\bf{Step 1:}} We call $V(\sigma,c)=  \sum_{i=1}^J m_{\sigma(i)}(c_i)$, where $c=(c_1,...,c_J) \in N^1(x_0, \delta')$, and our goal is now to construct a vector $c^{(J)}=(c_J^{(J)},...,c_J^{(J)})$ such that there is a unique permutation maximizing $V(\cdot, c^{(J)})$: what follows explain how to.
	
	We fix $c^{(1)}=(c^{(1)}_1,...,c^{(1)}_J) \in N^1(x_0, \delta')$, $A_1 = \displaystyle\max_{\sigma \in S_J} V(\sigma,c^{(1)})$, $\Sigma_1 = \left\lbrace \sigma \in S_J | V(\sigma, c^{(1)})= A_1 \right\rbrace$ (and $\Sigma_1 \neq \emptyset$), and $B_1 = \displaystyle \max_{\sigma \in S_J \backslash \Sigma_1} V(\sigma,c^{(1)})$ (if $B_1$ does exist, then $B_1 < A_1$).
	We consider a change of the first component of $c^{(1)}$, that is a vector $c^{(2)}$ which differs from $c^{(1)}$ only in the first component: the first component of $c^{(1)}$ is a point in $\R^n$, we consider a variation in its first covariate, with respect to which we know that the $(m_i)_{i=1...J}$ are J times differentiable.
	$$\forall \sigma \in S_J, V(\sigma, c^{(2)})=V(\sigma, c^{(1)}) + m_{\sigma(1)}(c_1^{(2)}) - m_{\sigma(1)}(c_1^{(1)}).$$
	
	We know that for all $x \in N^1(x_0, \delta')$, $(D_1 m_j(x))_{j=1..J}$ take distinct values: $\displaystyle\argmax_{s \in \left\lbrace\sigma(1) | \sigma \in \Sigma_1 \right\rbrace } D_1m_s(c_1^{(1)})$ is a singleton set $\left\lbrace s_1 \right\rbrace$. Hence, since the $m_i$ functions are at least twice differentiable, they are continuously differentiable, we can choose $c_1^{(2)}$ close enough from $c_1^{(1)}$ so that 
	$$ m_{s_1}(c_1^{(2)}) - m_{s_1}(c_1^{(1)}) = \max_{\sigma \in \Sigma_1 } \; m_{\sigma(1)}(c_1^{(2)}) - m_{\sigma(1)}(c_1^{(1)}), $$
	$$ c_1^{(2)} \in N^1(x_0, \delta'), $$
	and if $B_1$ exists, 
	$$m_i(c_1^{(2)}) - m_i(c_1^{(1)}) < \frac{A_1 - B_1}{2}, \forall i \leq J. $$
	
	Therefore, constructing $\Sigma_2 = \left\lbrace \sigma \in \Sigma_1 | \sigma(1)=s_1 \right\rbrace$ ($\Sigma_2 \neq \emptyset$ by construction), $A_2 = \displaystyle\max_{\sigma \in S_J} V(\sigma, c^{(2)})$, and $B_2 = \displaystyle\max_{\sigma \in S_J \backslash \Sigma_2} V(\sigma, c^{(2)})$, we know that $B_2$ exists and $B_2 < A_2$. We repeat the same process with the second component of $c^{(2)}$ and construct $s_2$, $\Sigma_3$, $c^{(3)}$, $A_3$ and $B_3$, and then we repeat it with the third component of $c^{(3)}$ and so on, until $|\Sigma_i|=1$ for some $i$. 
	If this is not the case for some $i<J$, then constructing each of the elements until $i=J$, we have
	$$\Sigma_J = \left\lbrace \sigma \in \Sigma_1 | \sigma(1)=s_1,..., \sigma(J-1)=s_{J-1} \right\rbrace, $$ implying $|\Sigma_J|=1$. The vector and the permutation obtained at the end that we call $c^{(J)}$ and $\sigma_J$ whatever the final number of steps is, are such that 
	$$ V(\sigma_J, c^{(J)}) = \displaystyle\max_{\sigma \in S_J} V(\sigma, c^{(J)}) \text{ and } \forall \sigma \neq \sigma_J, V(\sigma, c^{(J)}) < V(\sigma_J,c^{(J)}),$$ which is the result we wanted.

	\noindent {\bf{Step 2:}} Note that in the previous step, the last component of the vector $c^1$ did not change during the whole process: we could have chosen $c_J^{(1)} = x_0$. Since the order of those components do not matter, the previous result hold for some $c^{(J)}=(x_0,x_1^{(J)},...,x_{J-1}^{(J)})$. That is, 
	$$ \exists \sigma_J, \forall \sigma \in S_J, \sigma \neq \sigma_J \Rightarrow V(\sigma, c^{(J)})<V(\sigma_J, c^{(J)}).$$
	Since $\det D(t, x_0,x_1^{(J)},...,x_{J-1}^{(J)}) = \sum_{\sigma \in S_J} sign(\sigma) \; e^{t V(\sigma, c^{(J)}) }$, and $sign(\sigma) \in \left\lbrace -1,1 \right\rbrace$, $\det D(\cdot, x_0,x_1^{(J)},...,x_{J-1}^{(J)})$ is a finite sum of exponential functions multiplied by scalars where at least one of the scalars is nonzero. This implies that $\det D(\cdot, x_0,x_1^{(J)},...,x_{J-1}^{(J)})$ has a finite number of zeros (see, e.g, \citeasnoun{tossavainen2006zeros}).
	
\end{proof}

\newpage
\bibliographystyle{econometrica}
\bibliography{bibmix}
\end{document}